\documentclass[a4paper,onecolumn,allowtoday,accepted=2024-11-25]{quantumarticle}
\pdfoutput=1
\usepackage{standalone}
\usepackage[page,header]{appendix}
\usepackage{titletoc}
\usepackage{dsfont}
\usepackage[utf8]{inputenc}
\usepackage[english]{babel}
\usepackage[T1]{fontenc}
\usepackage{float}
\usepackage[10pt]{moresize}
\usepackage{amsmath}
\usepackage{amssymb}
\usepackage{amsthm}
\usepackage{wrapfig}
\usepackage{mathtools}
\usepackage{hyperref}
\usepackage{cleveref}
\usepackage{pgfplots}
\usepackage{delarray}
\pgfplotsset{width=10cm,compat=1.9}
\usepgfplotslibrary{groupplots, fillbetween}

\DeclareMathOperator{\tr}{tr}

\usepackage{tikz}
\usetikzlibrary{backgrounds,fit,hobby,decorations,snakes,decorations.markings,positioning,shadows.blur,calc}

\usepackage{array}
\usepackage{multirow,colortbl}
\newcommand{\PreserveBackslash}[1]{\let\temp=\\#1\let\\=\temp}
\newcolumntype{C}[1]{>{\PreserveBackslash\centering}p{#1}}
\newcolumntype{R}[1]{>{\PreserveBackslash\raggedleft}p{#1}}
\newcolumntype{L}[1]{>{\PreserveBackslash\raggedright}p{#1}}
\makeatletter
\newcommand{\thickhline}{%
	\noalign {\ifnum 0=`}\fi \hrule height 1pt
	\futurelet \reserved@a \@xhline
}
\newcolumntype{"}{@{\hskip\tabcolsep\vrule width 1pt\hskip\tabcolsep}}
\makeatother
\usepackage{caption}
\usepackage{subcaption}
\usepackage{mleftright}

\newcommand{\ket}[1]{\lvert #1 \rangle}
\newcommand{\bra}[1]{\langle #1 \rvert}
\newcommand{\ketbra}[2]{\lvert #1 \rangle\langle #2 \rvert}
\newcommand{\braket}[2]{\langle #1 \vert #2 \rangle}

\definecolor{myred}{RGB}{236, 17, 0}
\definecolor{myblue}{RGB}{10, 88, 153}
\definecolor{mygreen}{RGB}{26, 152, 81}
\definecolor{myorange}{RGB}{236, 137, 0}
\definecolor{LightGray}{RGB}{220,220,220}
\definecolor{LinkColor}{RGB}{167,20,49}

\pgfplotscreateplotcyclelist{mycolors}{
    myblue,mark=* \\
    myred,mark=square* \\
    mygreen,mark=otimes* \\
    myorange,mark=diamond* \\
}

\usepackage[framemethod=TikZ]{mdframed}
\mdfdefinestyle{MyFrame}{%
	linecolor=black,
	outerlinewidth=0.5pt,
	roundcorner=2pt,
	innertopmargin=15pt,
	innerbottommargin=15pt,
	innerrightmargin=15pt,
	innerleftmargin=15pt,
	backgroundcolor=gray!10!white}

\DeclareRobustCommand\ground{\begin{tikzpicture}[scale=0.5]
		\draw[thick](-0.4,0)--(0.4,0);
		\draw[thick](-0.3,0.1)--(0.3,0.1);
		\draw[thick](-0.2,0.2)--(0.2,0.2);
		\draw[thick](-0.1,0.3)--(0.1,0.3);
	\end{tikzpicture}}

\theoremstyle{plain}
\newtheorem{thm}{Theorem}[section]
\newtheorem{lem}[thm]{Lemma}

\newtheorem{cond}[thm]{Condition}
\theoremstyle{definition}
\newtheorem{defn}[thm]{Definition}

\crefname{defn}{Definition}{Definitions}
\Crefname{defn}{Definition}{Definitions}
\crefname{thm}{Theorem}{Theorems}
\Crefname{thm}{Theorem}{Theorems}
\crefname{claim}{Claim}{Claims}
\Crefname{claim}{Claim}{Claims}
\crefname{lem}{Lemma}{Lemmas}
\Crefname{lem}{Lemma}{Lemmas}
\crefname{rem}{Remark}{Remarks}
\Crefname{rem}{Remark}{Remarks}
\crefname{prop}{Proposition}{Propositions}
\Crefname{prop}{Proposition}{Propositions}
\crefname{cor}{Corollary}{Corollaries}
\Crefname{cor}{Corollary}{Corollaries}
\crefname{section}{Section}{Sections}
\Crefname{section}{Section}{Sections}
\crefname{equation}{}{}
\Crefname{equation}{}{}
\crefname{figure}{Figure}{Figures}
\Crefname{figure}{Figure}{Figures}
\crefname{appendix}{Appendix}{Appendices}
\Crefname{appendix}{Appendix}{Appendices}
\crefname{table}{Table}{Tables}
\Crefname{table}{Table}{Tables}
\crefname{exmp}{Example}{Examples}
\Crefname{exmp}{Example}{Examples}
\crefname{footnote}{Footnote}{Footnotes}
\Crefname{footnote}{Footnote}{Footnotes}
\crefname{cond}{Condition}{Conditions}
\Crefname{cond}{Condition}{Conditions}

\mleftright

\hypersetup{
    colorlinks,
    linkcolor=myblue,
    citecolor=myblue,
    filecolor=myblue,
    urlcolor=myblue,
    pdftitle={Security of DPS QKD from relativistic principles},
    pdfauthor={Marcus Haberland, Martin Sandfuchs, V. Vilasini, and Ramona Wolf}
}

\newcommand{\todo}[1]{{\color{myred} TODO: #1}}

\begin{document}

\title{Security of differential phase shift QKD from relativistic principles}

\author{Martin Sandfuchs}
\email{martisan@phys.ethz.ch}
\affiliation{Institute for Theoretical Physics, ETH Zürich, Wolfgang-Pauli-Str.\ 27, 8093 Zürich, Switzerland}
\author{Marcus Haberland}
\email{marcus.haberland@aei.mpg.de}
\affiliation{Institute for Theoretical Physics, ETH Zürich, Wolfgang-Pauli-Str.\ 27, 8093 Zürich, Switzerland}
\affiliation{Max Planck Institute for Gravitational Physics (Albert Einstein Institute), Am M\"uhlenberg 1, 14476 Potsdam, Germany}
\author{V. Vilasini}
\email{vilasini@inria.fr}
\affiliation{Institute for Theoretical Physics, ETH Zürich, Wolfgang-Pauli-Str.\ 27, 8093 Zürich, Switzerland}
\affiliation{Université Grenoble Alpes, Inria, 38000 Grenoble, France}
\author{Ramona Wolf}
\email{ramona.wolf@uni-siegen.de}
\affiliation{Institute for Theoretical Physics, ETH Zürich, Wolfgang-Pauli-Str.\ 27, 8093 Zürich, Switzerland}
\affiliation{Naturwissenschaftlich-Technische Fakultät, Universität Siegen, 57068 Siegen, Germany}

\maketitle

\fontfamily{lmr}\selectfont

\begin{abstract}
The design of quantum protocols for secure key generation poses many challenges: On the one hand, they need to be practical concerning experimental realisations. On the other hand, their theoretical description must be simple enough to allow for a security proof against all possible attacks. Often, these two requirements are in conflict with each other, and the differential phase shift (DPS) QKD protocol exemplifies these difficulties: It is designed to be implementable with current optical telecommunication technology, which, for this protocol, comes at the cost that many standard security proof techniques do not apply to it. After about 20 years since its invention, this work presents the first full security proof of DPS QKD against general attacks, including finite-size effects. The proof combines techniques from quantum information theory, quantum optics, and relativity. We first give a security proof of a QKD protocol whose security stems from relativistic constraints. We then show that security of DPS QKD can be reduced to security of the relativistic protocol. In addition, we show that coherent attacks on the DPS protocol are, in fact, stronger than collective attacks. Our results have broad implications for the development of secure and reliable quantum communication technologies, as they shed light on the range of applicability of state-of-the-art security proof techniques.
\end{abstract}

%

\section{Introduction}

The art of encryption is as old as the concept of writing systems. For thousands of years, people have invented sophisticated cryptographic techniques to hide the content of messages for various purposes, such as secret communication between governments or militaries. However, a look back at history suggests that cryptography is caught in a vicious circle: Cryptanalysts have always been quick to find ways to break any supposedly secure encryption method, prompting cryptographers to invent even more sophisticated schemes to hide information, and so on. In today's society, secure communication is a highly relevant issue as a large amount of sensitive data is transmitted over the internet. Quantum key distribution (QKD) \cite{Bennett1984,Ekert1991} offers a possibility to break the vicious circle by providing information-theoretically secure encryption, which is based (almost) solely on the laws of physics. Nonetheless, caution is still advised in this case: even these protocols can only break the circle if they come with a complete security proof against all possible attacks.

While QKD, in principle, offers a way to achieve unbreakable encryption, it comes with a number of challenges, in particular when turning theoretical ideas into practical applications. A crucial issue in this transformation is that actual devices, such as quantum sources and measurements, rarely conform to their corresponding description in the theoretical protocol. For example, a typical information carrier in QKD protocols is single photons. However, perfect single photon sources and detectors do not exist in practice. Since the security proof only applies to the assumptions made in the protocol description, these deviations open up the possibility of side-channel attacks such as the photon number splitting (PNS) attack \cite{Bennett1992b,Brassard2000}. This attack exploits that in an implementation, information is typically encoded into weak coherent pulses instead of single photons. These pulses have a small non-zero probability that more than one photon is emitted in one pulse, which allows the adversary to split off one of these photons without influencing the second one. Since this photon contains all information encoded in the pulse, the adversary can obtain information on the key without being detected.

\begin{figure}[t]
	\centering
	\newcommand{\BS}[1]{
    \begin{scope}[shift={(#1)}]
        \draw[fill=gray!30] (-0.25,-0.25) rectangle (0.25,0.25);
        \draw (-0.25,-0.25) rectangle (0.25,0.25);
        \draw (-0.25,-0.25) -- (0.25,0.25);
    \end{scope}
}

\begin{tikzpicture}[
    boxnode/.style={rectangle, draw=black, minimum width=0.8cm, thick, fill=black!20},
    emptynode/.style={},
    mirror/.pic={
        \draw[decorate,decoration={markings, mark=between positions 0.015 and 0.98 step 0.1072 with {
            \draw (0,0)--(90:3pt);}}] (-0.2,-0.2) -- (0.2,0.2);
        \draw[thick] (-0.2,-0.2) -- (0.2,0.2);
    },
    detector/.pic={\draw[draw=black,fill=gray!20] (0,0.2) arc(90:-90:0.2cm and 0.2cm) -- cycle;},
    scale=0.9,
    ]
    \node[boxnode] (source) at (-4.5, 0) {source};
    \node[boxnode] (PM) at (-2, 0) {PM};

    \node (BSBob1) at (3, 0) {};
    \BS{BSBob1.center};

    \node (BSBob2) at (6, 0) {};
    \begin{scope}[rotate=90]
        \BS{BSBob2.center};
    \end{scope}

    \node (MirrorBob1) at (3, 1.5) {};
    \draw (MirrorBob1.center) pic {mirror};

    \node (MirrorBob2) at (6, 1.5) {};
    \draw (MirrorBob2.center) pic[rotate=-90] {mirror};

    \begin{scope}[shift={(0.3, 0)}]
        \begin{scope}[shift={(-1, 0)}]
            \draw[fill=myred,draw=none] (0, 0) to[out=0,in=180,looseness=0.3] (0.3, 1) to[out=0,in=180,looseness=0.3] (0.6, 0) -- cycle;
        \end{scope}
        \begin{scope}
            \draw[fill=myred,draw=none] (0, 0) to[out=0,in=180,looseness=0.3] (0.3, 1) to[out=0,in=180,looseness=0.3] (0.6, 0) -- cycle;
        \end{scope}
        \begin{scope}[shift={(1, 0)}]
            \draw[fill=none,draw=myred,thick] (0, 0) to[out=0,in=180,looseness=0.3] (0.3, 1) to[out=0,in=180,looseness=0.3] (0.6, 0) -- cycle;
        \end{scope}
        \draw[thick,|<->|] (-0.7, 1.2) -- node[above]{\small $\Delta \phi = 0$} (0.3, 1.2);
        \draw[thick,|<->|] (+0.3, -0.2) -- node[below=0.05]{$\small \Delta \phi = \pi$} (1.3, -0.2);

        \node[text=myred] at (2.0, 0.5) {$\cdots$};
        \node[text=myred] at (-1.3, 0.5) {$\cdots$};
    \end{scope}

    \draw[thick,draw=myred] (source.east) -- (PM.west);
    \draw[thick,draw=myred] (PM.east) --  (BSBob1.center) -- node[above]{} (BSBob2.center);
    \draw[thick,draw=myred] (BSBob1.center) -- (MirrorBob1.center) -- node[above]{} (MirrorBob2.center) -- (BSBob2.center);

    \draw[thick,draw=myred] (BSBob2.center) -- ([xshift=+1.0cm] BSBob2.center) pic{detector} node[right,xshift=0.2cm]{};
    \draw[thick,draw=myred] (-1,0) (BSBob2.center) -- ([yshift=-1.0cm] BSBob2.center) pic[rotate=-90]{detector} node[below,yshift=-0.2cm]{};
\end{tikzpicture}
    \caption{\label{fig:dps_overview} Overview of the differential phase shift QKD protocol. A phase modulator (PM) is used to apply a random phase $\phi \in \{0, \pi\}$ (represented by the shading in the diagram) to each pulse in a train of coherent states. Alice's key bit is determined by the relative phase $\Delta \phi$ between subsequent pulses. Bob obtains his key by measuring the relative phase using a Mach-Zehnder interferometer.}
\end{figure}

A way to get around these kinds of problems is to design protocols whose theoretical description is closer to an experimentally feasible implementation, an approach that is followed by the differential phase shift (DPS) QKD protocol originally proposed in \cite{Inoue2002}. Already on the level of the theoretical description, this protocol employs coherent states as information carriers instead of single photons, which allows for an implementation with readily available optical telecommunication equipment. However, as explained above, using weak coherent pulses opens up the possibility of the PNS attack. DPS QKD counteracts this attack by combining coherent states with encoding information in the relation between two consecutive rounds rather than into single rounds (see \cref{fig:dps_overview} for an overview of DPS QKD). This directly rules out attacks that extract information from individual pulses, which includes the PNS attack. However, designing a protocol with a focus on implementations comes at a cost: While the protocol is simpler concerning its experimental realisation, the security proof poses two significant challenges: 
\begin{enumerate}
	\item Using coherent states instead of single photons means we have to deal with states in an infinite-dimensional Fock space instead of a qubit (or some other finite-dimensional) Hilbert space. This renders any numerical method for calculating secure key rates infeasible if one tries to apply it directly to states in the Fock space.
	\item The fact that information is encoded into the relation between two consecutive rounds rather than the individual rounds directly rules out some of the standard security proof techniques such as the quantum de Finetti theorem \cite{Renner2007,Renner2008} and the postselection technique \cite{Christandl2009}. This is because these techniques require the protocol rounds to be permutation invariant.
\end{enumerate}

In light of these challenges, it is perhaps not surprising that a full security proof of DPS QKD has not been achieved yet. Instead, the security of DPS has been proven in various simplified scenarios. These efforts of proving the security of the DPS protocol generally fall into two categories: In the first, additional assumptions are made about the possible attacks that an eavesdropper can carry out. Consequently, these proofs only provide conditional (rather than unconditional) security of the protocol. One example in this category is a security proof that only applies to the class of individual attacks \cite{Waks2006}. In the second category, typically, a modified version of the protocol with a (block) iid\footnote{``iid'' stands for ``independent and identically distributed'' and describes attacks where the eavesdropper applies the same strategy to each signal and, in particular, does not exploit correlations between signals.} structure is analysed. Notable examples include the security proof of single photon DPS \cite{Wen2009} and security proofs for versions of DPS with phase-randomised blocks \cite{Tamaki2012,Mizutani2017,Mizutani2023}. In addition to the security proofs, attacks on DPS QKD have been devised which provide upper bounds on its performance \cite{Curty2006,Curty2008,Curty2009}.

\begin{figure}[t]
    \centering
    \begin{tikzpicture}[scale=0.95]
       	\draw[fill=LightGray!75,draw=black,rounded corners,blur shadow={shadow blur steps=10}] (-2.5,-3.8) rectangle (2.5,-2.6);
    	\draw[fill=LightGray!75,draw=black,rounded corners,blur shadow={shadow blur steps=10}] (-2,-6.1) rectangle (2,-4.9);
    	\draw[fill=LightGray!75,draw=black,rounded corners,blur shadow={shadow blur steps=10}] (-5.4,-0.8) rectangle (-2.9,0.8);
    	\draw[fill=LightGray!75,draw=black,rounded corners,blur shadow={shadow blur steps=10}] (2.9,-0.8) rectangle (5.1,0.8);
    	\draw[fill=LightGray!75,draw=black,rounded corners,blur shadow={shadow blur steps=10}] (-1.1,-0.8) rectangle (1.1,0.8);
        \node[align=center] (DPS) at (0, -5.5) {Security of DPS QKD\\\cref{sec:SecurityDPS}};
        \node[align=center] (relativistic) at (0, -3.2) {Security of relativistic QKD\\\cref{sec:rel_qkd}};
        \node[align=center] (EAT) at (-4.15, 0) {Entropy\\accumulation\\\cref{sec:qkd_security}};
        \node[align=center] (causality) at (0, 0) {Causality\\principles\\\cref{sec:causality}};
        \node[align=center] (squash) at (4, 0) {Quantum\\optics\\\cref{sec:quantum_optics}};
        \draw[->,>=stealth,thick] ([xshift=0.75em] EAT.east) to[out=-10,in=100] ([xshift=-2.5em,yshift=0.75em] relativistic.north);
        \draw[->,>=stealth,thick] ([yshift=-0.75em] causality.south) to[out=-90,in=90] ([yshift=0.75em]relativistic.north);
        \draw[->,>=stealth,thick] ([xshift=-0.75em] squash.west) to[out=190,in=80] ([xshift=2.5em,yshift=0.75em] relativistic.north);
        \draw[->,>=stealth,thick] ([yshift=-0.75em] relativistic.south) to[out=-90,in=90] ([yshift=0.75em]DPS.north);
    \end{tikzpicture}
    \caption{\label{fig:proofingrerdients}Overview of the ingredients of the security proofs presented in this work, together with their corresponding sections in the paper.}
\end{figure}
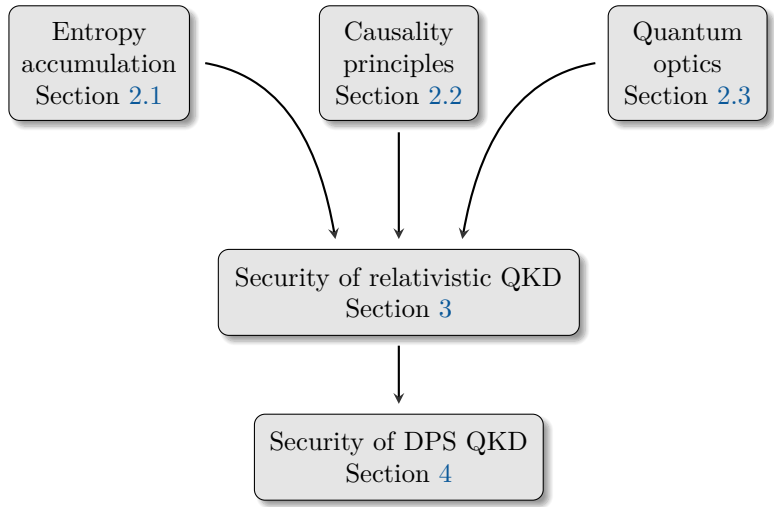

In this work, we provide a security proof of DPS QKD against general attacks, which combines ideas from quantum information theory, quantum optics, and relativity. As such, it is the first full security proof of a protocol where information is encoded in between rounds instead of into individual rounds, which does not require modifying the protocol to recover a (block) iid structure. The method we employ to achieve this task is a generalisation of the \emph{entropy accumulation theorem} (EAT) \cite{Dupuis2020,Depuis2019,Metger2022,Metger2023}, which allows us to derive secure key rates against general attacks taking into account finite-size effects. In contrast to methods such as the de Finetti theorem, it does not require the rounds of the protocol to be symmetric under permutation. It can hence be applied to protocols where information is encoded between two rounds. To account for the problem of the infinite-dimensional Fock space that describes the involved quantum states, we use a method called \emph{squashing} \cite{Gottesman2002, Tsurumaru2008, Beaudry2008, Gittsovich2014}. The general idea here is to formulate a protocol on a low-dimensional Hilbert space that is analogous (with respect to its security) to the actual protocol. On the level of the low-dimensional space, we can then apply numerical techniques for calculating the secure key rate. In order to construct an appropriate squashing map that can be applied together with the generalised EAT to prove the security of DPS QKD, we are in need of one missing ingredient: To meet the requirements of the generalised EAT, we need a well-defined sequence of channels. This can be enforced if Alice sends her signal states with a sufficient time delay, which provides a natural connection to relativistic principles, particularly causality, since it implies that certain systems cannot signal to each other. This then allows us to reduce the security analysis of DPS to the security analysis of relativistic protocols \cite{Molotkov2011,Molotkov2012,Radchenko2014,Kravtsov2018}. The ingredients and overall structure of our security proof are sketched in \cref{fig:proofingrerdients}. This timing condition plays a vital role throughout our security proof, particularly in constructing the squashing map. Remarkably, we can show that this condition is not just a technical requirement for our proof technique but an inherent feature of the DPS protocol, without which the protocol would be insecure.

In addition to a general security proof, our analysis reveals new insights into the power of different classes of attacks for DPS QKD. For the vast majority of QKD protocols (which have a full security proof), general attacks are not stronger than collective attacks, a restricted class of attacks where the adversary has to act identically and independently in each round, thus cannot exploit any correlations between rounds. Although there are a few exceptions, see \cite{Thinh2016,Sandfuchs2023}, they are very limited in number.
Upon comparing our findings on the DPS protocol with existing attacks, we demonstrate that it serves as a new example of a protocol where general attacks indeed surpass collective attacks in strength. This result enhances our understanding of the security characteristics of QKD protocols that do not have iid structure, as well as the limits of current proof techniques.

Lastly, we note that the techniques presented in this work are also of interest for QKD protocols other than DPS QKD. Since many of our tools concern themselves with relativistic constraints, this leads to a natural connection to relativistic QKD protocols \cite{Molotkov2011,Molotkov2012,Radchenko2014,Kravtsov2018}. These protocols have been developed separately of DPS QKD and are therefore of independent interest. Using our techniques, we are able to derive a full security proof for such relativistic protocols, including finite size effects, thus broadening the scope of our work.

In summary, in this paper, we give a security proof for DPS QKD and relativistic QKD against general attacks, including finite-size effects. This proof combines concepts from different areas, namely quantum information theory, quantum optics, and relativistic principles. To make it accessible for readers with different scientific backgrounds, we first introduce all necessary concepts from these areas in \cref{sec:preliminaries}. With these concepts at hand, we can then introduce relativistic QKD protocols and prove their security in \cref{sec:rel_qkd}. In \cref{sec:SecurityDPS}, we explain the DPS QKD protocol and show how to reduce its security to that of relativistic QKD protocols. Some aspects of these proofs deserve a more in-depth discussion, in particular the assumptions we use, which is provided in \cref{sec:Discussion}. In \cref{sec:outlook}, we conclude by providing some perspective on how our techniques can be used or modified for security proofs of related protocols.

\section{Preliminaries and techniques}
\label{sec:preliminaries}

In this section we cover the necessary background knowledge that is required to understand the security proofs. As hinted at in the introduction, there are three central ingredients: The entropy accumulation theorem, causality, and the squashing technique. In the following, we will cover these topics in that order. The notation and some basic definitions that are used throughout the section are listed in \cref{tab:notation}. The more technical definitions can be found in \cref{sec:tech_defs}.

\begin{table}[t]
    \centering
    \begin{tabular}{c|l}
        \hline
        Symbol & Definition \\ \hline
        $\mathcal{S}(A)$ & Density operators on the system $A$ \\
        $\mathcal{S}_\leq(A)$ & Sub-normalised density operators on the system $A$ \\
        $\mathcal{I}_A$ & The identity channel on system $A$ \\
        $\mathbb{P}_\mathcal{X}$ & Probability distributions over the alphabet $\mathcal{X}$ \\
        $\| M \|_1$ & Trace norm of $M$ \\
        $M^*$ & The adjoint of $M$ \\
        $A^n$ & Concatenation of the systems $A_1 \ldots A_n$ \\
        $\log(x)$ & The logarithm of $x$ to base 2 \\
        $\oplus$ & Addition modulo 2 \\
        $[n]$ & The set $\{1, 2, \ldots, n\}$  \\
        $\Omega^c$ & The complement of the set $\Omega$ \\
        \hline
    \end{tabular}
    \caption{\label{tab:notation}Various symbols and their definitions}
\end{table}


\subsection{Security of quantum key distribution}
\label{sec:qkd_security}

The security proof of any quantum key distribution protocol has to guarantee that the resulting key can be used in any application, for example in an encryption scheme. This is called \emph{composable security} \cite{Maurer2011,Portmann2022}, and achieving it boils down to deriving a security definition that ensures the security statement holds in any context. In this section, we explain the composable security definition we use in this work that goes back to \cite{Renner2008}, including notions of correctness, secrecy, and completeness, and discuss what a security proof must entail to meet it.

The general idea of the security definition is that we aim to quantify how far the actual key resource whose security we want to prove is from an ideal key resource. The goal of a QKD protocol is for two spatially distant parties (called Alice and Bob) to establish a shared secret key, hence the ideal key resource should fulfil two properties: (i) the resulting key has to be the same for Alice and Bob, and (ii)
an adversary must not have any knowledge of it. A secure QKD protocol will either produce a key that fulfils these properties or abort. This is captured by the following definition:

\begin{defn}
	\label{def:QKDsecurity}
	Consider a QKD protocol which can either produce a key of length $l$ or abort, and let $\rho_{K_A^l K_B^l E}$ be the final quantum state at the end of the protocol, where $K_A^l$ and $K_B^l$ are Alice's and Bob's version of the final key, respectively, and $E$ is the quantum system that contains all knowledge available to an adversary Eve. The protocol is said to be $\varepsilon^\mathrm{cor}$-correct, $\varepsilon^\mathrm{sec}$-secret, and $\varepsilon^\mathrm{comp}$-complete if the following holds:
	\begin{enumerate}
		\item \emph{Correctness:} For any implementation of the protocol and any behaviour of the adversary, 
			\begin{equation}
				\Pr[K_A^l\neq K_B^l \wedge  \mathrm{accept} ] \le\varepsilon^\mathrm{cor}.
			\end{equation}
		\item \emph{Secrecy:} For any implementation of the protocol and any behaviour of the adversary, 
			\begin{equation}
				\frac{1}{2} \left\|(\rho_{\wedge\Omega})_{K_A^l E}-\tau_{K_A^l}\otimes (\rho_{\wedge\Omega})_E\right\|_1\le\varepsilon^\mathrm{sec},
			\end{equation}
		where $\rho_{\wedge\Omega}$ is the sub-normalized state after running the protocol conditioned on the event $\Omega$ of not aborting, and $\tau_{K_A^l}=\frac{1}{2^{l}}\sum_k\ketbra{k}{k}_{K_A^l}$ is the maximally mixed state on the system $K_A^l$ (i.e., a uniformly random key for Alice).
		\item \emph{Completeness:} There exists an honest implementation of the protocol such that
			\begin{equation}
				\Pr[\mathrm{abort}]\le\varepsilon^\mathrm{comp}.
			\end{equation}
	\end{enumerate}
\end{defn}

Note that in the above definition, correctness and secrecy must be fulfilled for any behaviour of the adversary. These conditions ensure that Alice and Bob's probability of getting different or insecure keys without detecting it (i.e., without the protocol aborting) is low. They are often summarized to a single condition called \emph{soundness}:

\begin{defn}
	\label{def:soundness_subnorm}
	Consider a QKD protocol which can either produce a key of length $l$ or abort, and let $\rho_{K_A^l K_KE}$ be the final quantum state at the end of the protocol, where $K_A^l$ and $K_B^l$ are Alice's and Bob's version of the final key, respectively, and $E$ is the quantum system that contains all knowledge available to an adversary Eve. The protocol is said to be $\varepsilon^\mathrm{snd}$-sound if
	\begin{equation}
		\frac{1}{2}\left\|(\rho_{\wedge\Omega})_{K_A^l K_B^l E}-\tau_{K_A^l K_B^l}\otimes(\rho_{\wedge\Omega})_E\right\|_1\le\varepsilon^\mathrm{snd},
	\end{equation}
	where $\rho_{\wedge\Omega}$ is the sub-normalized state after running the protocol conditioned on the event $\Omega$ of not aborting, and $\tau_{K_A^l K_B^l}=\frac{1}{2^{l}}\sum_k\ketbra{kk}{kk}_{K_A^l K_B^l}$ is the maximally mixed state on the system $K_A^l K_B^l$ (i.e., an identical pair of uniformly random keys for Alice and Bob).
\end{defn}
It is straightforward to show that if a protocol is $\varepsilon^\mathrm{cor}$-correct and $\varepsilon^\mathrm{sec}$-secret, then it is $\varepsilon^\mathrm{snd}$-sound with $\varepsilon^\mathrm{snd}=\varepsilon^\mathrm{cor}+\varepsilon^\mathrm{sec}$ (see, for example, \cite{Portmann2022}). It is possible to either show correctness and secrecy separately or to show soundness directly. For the differential phase shift protocol, we choose to show soundness directly.
In contrast to soundness, completeness is concerned only with the honest implementation, i.e., the case where the adversary is not trying to corrupt the execution of the protocol. For instance, a protocol that always aborts fulfils the first two conditions of the definition, but it is not a useful protocol. These kinds of protocols are excluded by imposing the completeness condition.

To prove that a QKD protocol fulfils the conditions in \cref{def:QKDsecurity}, we typically employ two-universal hash functions and randomness extractors (see, for example, the protocol described in \cref{subsec:DPSprotocol}). Here, we give a rough sketch of what a security proof entails and briefly recall the definitions of the required primitives. In \cref{sec:security_general,sec:security_rel_qkd,sec:security_proof_dps}, you can find a detailed security proof.

\subsubsection*{Completeness} To show completeness, one has to show that there exists an honest implementation of the protocol such that it aborts with low probability. This is usually straightforward to show as the honest behaviour typically has an IID structure. As long as we allow for enough tolerance in the parameter estimation step, one can choose the amount of resources used for the error correction step such that the probability of aborting is low (see \cref{sec:security_general}).

\subsubsection*{Correctness} Showing correctness means deriving a bound on the probability that the error-corrected strings are not equal but the protocol does not abort. In case the error-correction procedure includes a step where Bob checks whether his guess of Alice's string is correct, this is also straightforward to show. The checking step can be implemented using two-universal hash functions:
	\begin{defn}[Two-universal hash function]
		\label{def:twounivhash}
		Let $\mathcal{F}$ be a family of hash functions between sets $\mathcal{X}$ and $\mathcal{Z}$. We call $\mathcal{F}$ two-universal if for all $x,x'\in\mathcal{X}$ with $x\neq x'$ it holds that
		\begin{equation}
			\Pr_{f\in\mathcal{F}}\left[f(x)=f(x')\right]\le\frac{1}{|\mathcal{Z}|},
		\end{equation}
		where $f\in\mathcal{F}$ is chosen uniformly at random.
	\end{defn}

Alice and Bob can hence choose a hash function $f\in\mathcal{F}$ and compare the outputs of Alice's key and Bob's guess of her key. If their bit strings are not equal, they will detect it with probability $1-1/|\mathcal{Z}|$, which means that by choosing the size of the output set $\mathcal{Z}$ one can ensure that this probability is high. This checking step is independent of the actual error correction procedure that is employed, hence it allows us to decouple the proof of correctness from all properties of the error correction step (except its output).

\subsubsection*{Secrecy} Showing secrecy is the most difficult part of the security proof, as one has to take into account any possible behaviour of the adversary. Secrecy is ensured in the last step of the protocol, privacy amplification. A possible procedure to implement this step is again based on two-universal hashing \cite{Renner2008}: As in the checking step after the error correction procedure, Alice and Bob choose a function from a family of two-universal hash functions and apply it to their respective strings. The following lemma (taken from \cite{Tomamichel2017}) then ensures that the resulting key fulfils the properties described in \cref{def:QKDsecurity}:


\begin{lem}[Quantum leftover hashing]
	\label{lem:leftoverhashing}
	Let $\rho_{f(X)FE}\in\mathcal{S}_\le(ZFE)$ be the (sub-normalized) state after applying a function $f$, randomly chosen from a family of two-universal hash functions $\mathcal{F}$ from $\mathcal{X}$ to $\mathcal{Z}$, to the bit string $X$. Then, for every $\varepsilon>0$ it holds that 
	\begin{equation}
		\frac{1}{2}\left\|\rho_{f(X)FE}-\tau_Z\otimes\rho_{FE}\right\|_1\le 2\varepsilon+2^{-\frac{1}{2}\left(H_\mathrm{min}^\varepsilon(X|E)-l+2\right)},
	\end{equation}
	where $l=|\mathcal{Z}|$, $\tau_Z$ is the maximally mixed state on $Z$, and $F$ is the register that holds the choice of the hash function.
\end{lem}

Note that the smooth min-entropy $H_{\min}^\varepsilon(X|E)$ is evaluated on the state $\rho_{XE}$, i.e., the state of the
system before applying the hash function. \Cref{lem:leftoverhashing} then states that if there is a sufficient amount of initial smooth min-entropy, applying a random hash function results in a state that is almost product with the adversary’s information. This means that to prove secrecy we need to find a sufficiently large lower bound on the smooth min-entropy which holds for general attacks of the adversary.

There are several techniques for finding such a bound. The typical strategy in a security proof is to find a bound that is valid if the adversary is limited to collective attacks, i.e., they apply the same attack in every round, and the individual rounds are uncorrelated. From this, a bound that is valid for general attacks can be inferred via techniques based on the quantum de Finetti theorem \cite{Renner2008,Christandl2009} or the entropy accumulation theorem (EAT) \cite{Dupuis2020,George2022}. From the bound on the smooth min-entropy we can then obtain a lower bound on the \emph{key rate}
	\begin{equation}
		r=\frac{l}{n},
	\end{equation}
where $l$ is the length of the final key, and $n$ is the number of rounds via \cref{lem:leftoverhashing} (more details about this can be found in \cref{sec:soundness}). It is often easier to calculate this bound in the asymptotic case where the number of rounds $n$ goes to infinity. However, for a full security proof and to obtain meaningful bounds for practical protocols, it is necessary to also include finite-size effects, which occur because the protocol consists only of a finite number of rounds.

The technique we employ in this work is the EAT in its recently developed generalised form \cite{Metger2022,Metger2023}. Apart from guaranteeing security against general attacks it allows us to take into account finite-size effects. The EAT relates the smooth min-entropy of $n$ rounds in the case of general attacks to the von Neumann entropy of a single round in the case of collective attacks, which, in general, is much easier to bound. The general setting in which we can apply the generalised EAT is depicted in \cref{fig:dps_setup_testing}. Before we can state the theorem, we need to introduce some definitions that describe the setup to which it applies, in particular, the notion of EAT channels. For the technical definitions we refer to \cref{sec:tech_defs}.
 

\begin{defn}[EAT channel]
\label{def:eat_channel}
    Let $\{ \mathcal{M}_i : \mathcal{S}(R_{i-1}E_{i-1}) \rightarrow \mathcal{S}(R_iE_iA_iC_i) \}_{i \in [n]}$ be a sequence of CPTP maps, where $C_i$ are classical registers with common alphabet $\mathcal{C}$. We call the channels $\{\mathcal{M}_i\}_i$ \emph{EAT channels} if they satisfy the following conditions:
	\begin{enumerate}
        \item There exists a CPTP map $\mathcal{R}_i : \mathcal{S}(E_{i-1}) \rightarrow \mathcal{S}(E_i)$ such that $\tr_{A_iR_iC_i} \circ \mathcal{M}_i = \mathcal{R}_i \circ \tr_{R_{i-1}}$.
        \item Let $\mathcal{M}_i' = \tr_{C_i} \circ \mathcal{M}_i$. Then there exists a CPTP map $\mathcal{T}: \mathcal{S}(A^nE_n) \rightarrow \mathcal{S}(C^nA^nE_n)$ of the form
		\begin{align}
			\mathcal{T}(\omega_{A^nE_n}) = \sum_{y \in \mathcal{Y}, z \in \mathcal{Z}} (\Pi_{A^n}^{(y)} \otimes \Pi_{E_n}^{(z)}) \omega_{A^nE_n} (\Pi_{A^n}^{(y)} \otimes \Pi_{E_n}^{(z)}) \otimes \ketbra{r(y,z)}{r(y,z)}_{C^n},
		\end{align}
		such that $\mathcal{M}_n \circ \ldots \circ \mathcal{M}_1 = \mathcal{T} \circ \mathcal{M}_n' \circ \ldots \circ \mathcal{M}_1'$. The operators $\{\Pi_{A^n}^{(y)}\}_y$ and $\{\Pi_{E_n}^{(z)}\}_z$ are mutually orthogonal projectors and $r: \mathcal{Y}\times\mathcal{Z} \rightarrow \mathcal{C}$ is a (deterministic) function.
	\end{enumerate}
\end{defn}

The first of these conditions states that the map $\mathcal{M}_i$ does not signal from $R_{i-1}$ to $E_i$. This non-signalling constraint is required as part of the EAT and is distinct from the other non-signalling constraints that will arise in the analysis of the protocol. For a more detailed discussion of non-signalling maps we refer to \cref{sec:causality}. We note that the second condition above is always satisfied if $C_i$ is computed from classical information in $A^n$ and $E_n$. A diagram of the channels is shown in \cref{fig:dps_setup_testing}.

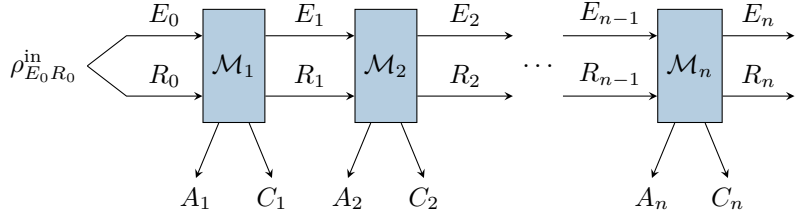
\begin{figure}[t]
	\centering
	\begin{tikzpicture}[
		boxnode/.style={rectangle, draw=black, minimum height=1.5cm, minimum width=0.8cm, fill=myblue!30},
	]
		\node (rho) at (-2.5, 0) {$\rho_{E_0R_0}^\mathrm{in}$};
		\node[boxnode] (M1) at (0, 0) {$\mathcal{M}_1$};
		\node[boxnode] (M2) at (2, 0) {$\mathcal{M}_2$};
		\node (ellipsis) at (4, 0) {$\ldots$};
		\node[boxnode] (Mn) at (6, 0) {$\mathcal{M}_n$};
		\node (end) at (7.5, 0) {};
		
		\draw[->,>=stealth] (rho.east) -- ([xshift=-1.0cm,yshift=0.4cm] M1.west) -- node[above]{$E_0$} ([yshift=0.4cm] M1.west);
		\draw[->,>=stealth] (rho.east) -- ([xshift=-1.0cm,yshift=-0.4cm] M1.west) -- node[above]{$R_0$} ([yshift=-0.4cm] M1.west);
		
		\draw[->,>=stealth] ([yshift=0.4cm] M1.east) -- node[above]{$E_1$} ([yshift=0.4cm] M2.west);
		\draw[->,>=stealth] ([yshift=-0.4cm] M1.east) -- node[above]{$R_1$} ([yshift=-0.4cm] M2.west);
		
		\draw[->,>=stealth] ([yshift=0.4cm] M2.east) -- node[above]{$E_2$} ([yshift=0.4cm] ellipsis.west);
		\draw[->,>=stealth] ([yshift=-0.4cm] M2.east) -- node[above]{$R_2$} ([yshift=-0.4cm] ellipsis.west);
		
		\draw[->,>=stealth] ([yshift=0.4cm] ellipsis.east) -- node[above]{$E_{n-1}$} ([yshift=0.4cm] Mn.west);
		\draw[->,>=stealth] ([yshift=-0.4cm] ellipsis.east) -- node[above]{$R_{n-1}$} ([yshift=-0.4cm] Mn.west);
		
		\draw[->,>=stealth] ([yshift=0.4cm] Mn.east) -- node[above]{$E_n$} ([yshift=0.4cm] end.west);
		\draw[->,>=stealth] ([yshift=-0.4cm] Mn.east) -- node[above]{$R_n$} ([yshift=-0.4cm] end.west);
		
		\node (A1) at ([xshift=-0.5cm,yshift=-1.0cm] M1.south) {$A_1$};
		\node (C1) at ([xshift=+0.5cm,yshift=-1.0cm] M1.south) {$C_1$};
		\draw[->,>=stealth] ([xshift=-0.2cm] M1.south) -- (A1.north);
		\draw[->,>=stealth] ([xshift=+0.2cm] M1.south) -- (C1.north);
		
		\node (A2) at ([xshift=-0.5cm,yshift=-1.0cm] M2.south) {$A_2$};
		\node (C2) at ([xshift=+0.5cm,yshift=-1.0cm] M2.south) {$C_2$};
		\draw[->,>=stealth] ([xshift=-0.2cm] M2.south) -- (A2.north);
		\draw[->,>=stealth] ([xshift=+0.2cm] M2.south) -- (C2.north);
		
		\node (An) at ([xshift=-0.5cm,yshift=-1.0cm] Mn.south) {$A_n$};
		\node (Cn) at ([xshift=+0.5cm,yshift=-1.0cm] Mn.south) {$C_n$};
		\draw[->,>=stealth] ([xshift=-0.2cm] Mn.south) -- (An.north);
		\draw[->,>=stealth] ([xshift=+0.2cm] Mn.south) -- (Cn.north);
	\end{tikzpicture}
	\caption{Setup of the generalised EAT with testing. In each round $i$ the channels take quantum inputs $E_{i-1}$ and $R_{i-1}$ and produce classical outputs $A_i$ and $C_i$. The registers $C_i$ are used to collect statistics to constrain the set of allowed channels $\{\mathcal{M}_i\}_i$.}
	\label{fig:dps_setup_testing}
\end{figure}

\begin{defn}[Min-tradeoff function] \label{defn:min_tradeoff} Let $\{\mathcal{M}_i\}_i$ be a sequence of EAT channels. For $i \in [n]$ and $q \in \mathbb{P}_\mathcal{C}$ we define the set $\Sigma_i(q)$ of all states that are compatible with the statistics $q$ under the map $\mathcal{M}_i$:
	\begin{align}
        \Sigma_i(q) = \{ \nu_{C_iA_iR_iE_i\tilde{E}_{i-1}} = \mathcal{M}_i(\omega_{R_{i-1}E_{i-1}\tilde{E}_{i-1}}) \; | \; \omega \in \mathcal{S}(R_{i-1}E_{i-1}\tilde{E}_{i-1}) \; \mathrm{and} \; \nu_{C_i} = q \},
	\end{align}
	where $\nu_{C_i}$ represents the distribution over $\mathcal{C}$ given by $\mathrm{Pr}[c]=\bra{c}\nu_{C_i}\ket{c}$ and $\tilde{E}_{i-1}$ is a system isomorphic to $R_{i-1}E_{i-1}$. 
	An affine function $f: \mathbb{P}_\mathcal{C} \rightarrow \mathbb{R}$ is called a \emph{min-tradeoff function} for $\{\mathcal{M}_i\}_i$ if it satisfies
	\begin{align}
		f(q) \leq \inf_{\nu \in \Sigma_i(q)} H(A_i|E_i\tilde{E}_{i-1})_\nu \qquad \forall q \in \mathbb{P}_\mathcal{C}, i \in [n].
	\end{align}
\end{defn}

The second-order corrections in the EAT will depend on some properties of the min-tradeoff function (see \cref{defn:min_tradeoff}), which are given in the following definition:
\begin{defn}[Min, Max and Var]
	Let $\{ \mathcal{M}_i \}_i$ be a sequence of EAT channels and let $f: \mathbb{P}_\mathcal{C} \rightarrow \mathbb{R}$ be an affine function. We define
	\begin{equation}
		\begin{aligned}
			\mathrm{Max}(f) &= \max_{q \in \mathbb{P}_\mathcal{C}} f(q), \\
			\mathrm{Min}_{\Sigma}(f) &= \min_{q: \Sigma(q) \neq \emptyset} f(q), \\
			\mathrm{Var}(f) &= \max_{q: \Sigma(q) \neq \emptyset} \left\{ \sum_{c \in \mathcal{C}} q(c)f(\delta_c)^2 - \left( \sum_{c \in \mathcal{C}} q(c)f(\delta_c) \right)^2 \right\},
		\end{aligned}
	\end{equation}
	where $\Sigma(q) = \bigcup_{i} \Sigma_i(q)$ and $\delta_c$ is the distribution with deterministic output $c$.
\end{defn}

\begin{defn}
	Let $\mathcal{C}$ be some finite alphabet and let $C^n \in \mathcal{C}^n$ for some $n \in \mathbb{N}$. Then $\mathrm{freq}(C^n) \in \mathbb{P}_\mathcal{C}$ is defined as the probability distribution given by:
	\begin{align}
		\mathrm{freq}(C^n)(c) = \frac{|\{i | C_i = c\}|}{n} \qquad \forall c \in \mathcal{C}.
	\end{align}
\end{defn}

With this in hand we are now able to state the theorem:
\begin{thm}[Generalised EAT \cite{Metger2022}] \label{thm:gen_EAT}
    Let $\{ \mathcal{M}_i \}_i$ be a sequence of EAT channels and let $f$ be a min-tradeoff function for those channels. Furthermore, let $\Omega \subseteq \mathcal{C}^n$ and $\rho_{A^nC^nR_nE_n} = \mathcal{M}_n \circ \ldots \circ \mathcal{M}_1 (\rho_{R_0E_0}^{\mathrm{in}})$ be the output state for some initial state $\rho_{R_0E_0}^{\mathrm{in}} \in \mathcal{S}(R_0E_0)$. Then for all $\alpha \in (1, 3/2)$ and $\varepsilon > 0$,
	\begin{equation}
		\begin{aligned}
			H_\mathrm{min}^{\varepsilon}(A^n|E_n)_{\rho_{|\Omega}}
			\geq n t - n \frac{\alpha-1}{2-\alpha}\frac{\ln(2)}{2}V^2 - \frac{g(\varepsilon) + \alpha \log\left(\frac{1}{\rho[\Omega]}\right)}{\alpha-1} - n \left(\frac{\alpha-1}{2-\alpha}\right)^2 K(\alpha),
		\end{aligned}
	\end{equation}
	where $\rho[\Omega]$ is the probability of observing the event $\Omega$, and
	\begin{equation}
		\begin{aligned}
			t &= \min_{c^n \in \Omega} f(\mathrm{freq}(c^n)), \\
			g(\varepsilon) &= \log \frac{1}{1 - \sqrt{1-\varepsilon^2}}, \\
			V &= \log(2 d_A^2 + 1) + \sqrt{2 + \mathrm{Var}(f)}, \\
			K(\alpha) &= \frac{(2-\alpha)^3}{6(3-2\alpha)^3 \ln 2} 2^{\frac{\alpha-1}{2-\alpha}(2\log d_A + \mathrm{Max}(f) - \mathrm{Min}_\Sigma(f))} \ln^3 \left(2^{2\log d_A + \mathrm{Max}(f) - \mathrm{Min}_\Sigma(f)} + e^2 \right),
		\end{aligned}
	\end{equation}
	with $d_A=\max_i \dim (A_i)$.
\end{thm}

\begin{proof}
	See \cite{Metger2022}.
\end{proof}

\subsection{Relativistic principles and causality}
\label{sec:causality}

Relativistic principles of causation prohibit signalling outside the future light-cone. In order to incorporate such principles into quantum protocols in a space-time, we must consider how non-signalling conditions can be formulated at the level of quantum operations. Consider a quantum operation (a completely positive and trace preserving linear map) $\mathcal{E}_{SR\rightarrow S'R'}: \mathcal{S}(SR) \rightarrow \mathcal{S}(S'R')$, where $S$, $R$, $S'$ and $R'$ are quantum systems of arbitrary (possibly infinite) dimensions. Suppose the input quantum system $S$ and the output system $R'$ are associated with spacelike separated locations. We would then desire that $S$ does not signal to $R'$. Operationally speaking, the choice of a local operation $\mathcal{M}_S: \mathcal{S}(S) \rightarrow \mathcal{S}(S)$ performed on the input $S$ of $\mathcal{E}_{SR\rightarrow S'R'}$ must never be detectable when accessing the system $R'$ alone. This is captured by the following definition:

\begin{defn}
\label{defn:nonsignalling}
  We say that $S$ does not signal to $R'$ in a linear CPTP map $\mathcal{E}_{SR\rightarrow S'R'}: \mathcal{S}(SR) \rightarrow \mathcal{S}(S'R')$ if and only if for all local operations $\mathcal{M}_S: \mathcal{S}(S) \rightarrow \mathcal{S}(S)$ on $S$, the following holds
  \begin{equation}
  \label{eq: nonsignalling1}
  \tr_{S'}  \circ \mathcal{E}_{SR\rightarrow S'R'} =  \tr_{S'}  \circ \mathcal{E}_{SR\rightarrow S'R'} \circ \big(\mathcal{M}_S\otimes \mathcal{I}_R \big).
  \end{equation}
\end{defn}


\begin{figure}[t]
    \centering
 \begin{subfigure}[b]{0.425\textwidth}
 	\centering
 	\begin{tikzpicture}[trace/.pic={\draw [thick](-0.4,0)--(0.4,0);\draw [thick](-0.3,0.1)--(0.3,0.1);\draw [thick](-0.2,0.2)--(0.2,0.2);\draw [thick](-0.1,0.3)--(0.1,0.3);}, scale=0.6, transform shape]
 
\draw[thick, fill=LightGray, opacity=0.75 ] (0,0) rectangle node[align=center]{\Large{$\mathcal{E}_{SR\rightarrow S'R'}$}} (3,2);

 \draw[thick,black,->] (0.5,-2)--node[anchor=east]{\large$S$}(0.5,0);  \draw[thick,black,->] (2.5,-2)--node[anchor=east]{\large$R$}(2.5,0); \draw[thick,black,->] (0.5,2)--node[anchor=east]{\large$S'$}(0.5,4); \draw[thick,black,->] (2.5,2)--node[anchor=east]{\large$R'$}(2.5,4);  \draw (0.5,4) pic {trace};

   \node[align=center, black] at (4.5,1) {\Huge{$\mathbf{=}$}};

\begin{scope}[xshift=-1cm]
	   \draw[thick, fill=LightGray, opacity=0.75 ] (7,0) rectangle node[align=center]{\Large{$\mathcal{E}_{SR\rightarrow S'R'}$}} (10,2);
\draw[thick,black,->] (7.5,-1)--(7.5,0);  
\draw[thick,black,->] (9.5,-2)--(9.5,0); \draw[thick,black,->] (7.5,2)--node[anchor=east]{\large$S'$}(7.5,4); \draw[thick,black,->] (9.5,2)--node[anchor=east]{\large$R'$}(9.5,4);  \draw (7.5,4) pic {trace};

   \draw[thick, fill=LightGray, opacity=0.75 ] (6.8,-2) rectangle node[align=center]{\large$\mathcal{M}_S$} (8.2,-1);
\draw[thick,black,->] (7.5,-3)--node[anchor=east]{\large$S$}(7.5,-2);
\draw[thick,black] (9.5,-3)--node[anchor=east]{\large$R$}(9.5,-2);
\end{scope}
\end{tikzpicture}
 \caption{Diagrammatic representation of \cref{eq: nonsignalling1}. The above equality must hold for all local maps $\mathcal{M}_S$ on $S$.}
 \end{subfigure}
\hfill
 \begin{subfigure}[b]{0.425\textwidth}
 	\centering
 	\begin{tikzpicture}[trace/.pic={\draw [thick](-0.4,0)--(0.4,0);\draw [thick](-0.3,0.1)--(0.3,0.1);\draw [thick](-0.2,0.2)--(0.2,0.2);\draw [thick](-0.1,0.3)--(0.1,0.3);}, scale=0.6, transform shape]
   
\draw[thick, fill=LightGray, opacity=0.75 ] (0,0) rectangle node[align=center]{\Large{$\mathcal{E}_{SR\rightarrow S'R'}$}} (3,2);

 \draw[thick,black,->] (0.5,-2)--node[anchor=east]{\large$S$}(0.5,0);  \draw[thick,black,->] (2.5,-2)--node[anchor=east]{\large$R$}(2.5,0); \draw[thick,black,->] (0.5,2)--node[anchor=east]{\large$S'$}(0.5,4); \draw[thick,black,->] (2.5,2)--node[anchor=east]{\large$R'$}(2.5,4);  \draw (0.5,4) pic {trace};

   \node[align=center, black] at (4.5,1) {\Huge{$\mathbf{=}$}};

\begin{scope}[xshift=-1cm]
	 \draw[thick,black,->] (7.5,-2)--node[anchor=east]{\large$S$}(7.5,0);  \draw[thick,black,->] (9.5,-2)--node[anchor=east]{\large$R$}(9.5,0); 
 \draw[thick,black,->] (9.5,2)--node[anchor=east]{\large$R'$}(9.5,4); 
\draw (7.5,0) pic {trace};

\draw[thick, fill=LightGray, opacity=0.75 ] (8.5,0) rectangle node[align=center]{\Large{$\mathcal{E}_{R\rightarrow R'}$}} (10.5,2);
\end{scope}

\end{tikzpicture}
  \caption{Diagrammatic representation of \cref{eq: nonsignalling2}. There must exist a quantum CPTP map $\mathcal{E}_{R\rightarrow R'}$ such that the above equality holds.}
 \end{subfigure}
\caption{\label{fig:nonsig2} Diagrammatic representation of two equivalent definitions of non-signalling in a quantum map. The ground symbol \ground\ denotes the trace operation.}
\end{figure}
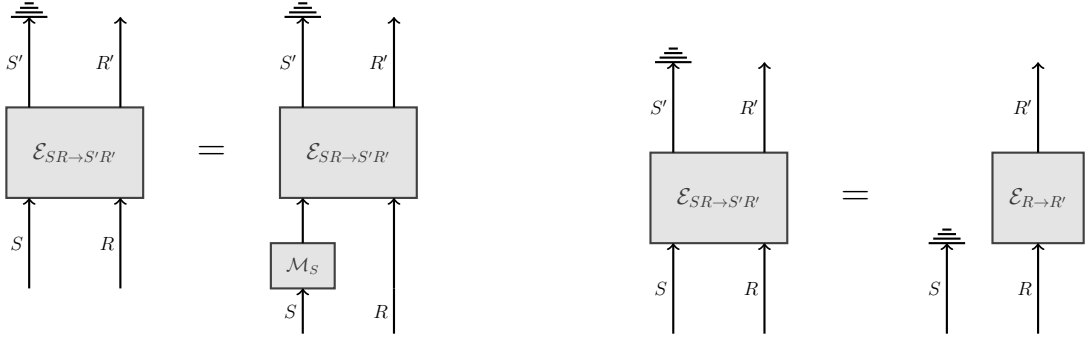

Another natural way to define signalling would be to require that once we trace out $S'$ and only observe the output at $R'$, then we can also trace out the input at $S$ and only use the input at $R$. That is, there exists a quantum channel  $\mathcal{E}_{R\rightarrow R'}: \mathcal{S}(R)\rightarrow \mathcal{S}(R')$ such that 
\begin{equation}
\label{eq: nonsignalling2}
   \tr_{S'}  \circ \mathcal{E}_{SR\rightarrow S'R'}  = \tr_S \otimes \mathcal{E}_{R\rightarrow R'} .
\end{equation}
In fact, it turns out that the two definitions of signalling, \cref{eq: nonsignalling1} and \cref{eq: nonsignalling2} are equivalent \cite{Ormrod2023}.\footnote{While this result is stated only for unitary $\mathcal{E}_{SR\rightarrow S'R'}$ in \cite{Ormrod2023}, their proof applies to arbitrary quantum CPTP maps.}

The following lemma provides another equivalent condition to non-signalling in the CPTP map $\mathcal{E}_{SR\rightarrow S'R'}$, in terms of its Choi state $\mathcal{C}(\mathcal{E}_{SR\rightarrow S'R'})\in \mathcal{S}(\bar{S}\bar{R}S'R')$, in the case where $S$, $R$, $S'$ and $R'$ are finite dimensional quantum systems. The Choi state of a CP map on finite dimensional systems
 is defined as follows. 
\begin{equation}
\label{eq: choi}
    \mathcal{C}(\mathcal{E}_{SR\rightarrow S'R'})\coloneqq (\mathcal{I}_{\bar{S}\bar{R}}\otimes \mathcal{E}_{SR\rightarrow S'R'}) \ket{\Phi}\bra{\Phi}_{\bar{S}\bar{R}SR},
\end{equation}
where $\bar{S}$ and $\bar{R}$ have isomorphic state spaces to the systems $S$ and $R$, respectively, and $\ket{\Phi}_{\bar{S}\bar{R}SR}=\frac{1}{\sqrt{d_Sd_R}}\sum_{ij}\ket{ijij}_{\bar{S}\bar{R}SR}$ is the normalised maximally entangled state on the bi-partition $\bar{S}\bar{R}$ and $SR$ with respect to a chosen basis $\{\ket{i}\}_i$ of the isomorphic systems $S$ and $\bar{S}$, and the basis $\{\ket{j}\}$ of the isomorphic systems $R$ and $\bar{R}$. While a Choi representation for the infinite dimensional case can be defined, it does not correspond to a state, but to a sesquilinear positive-definite form \cite{Holevo2011}. In this paper, we will only require the finite-dimensional Choi representation which is captured by the Choi state of a CP map. Note however that the definitions of signalling defined above also apply to the infinite dimensional case.

We now state the lemma. It is based on the idea of encoding channel decomposition properties in the Choi state of the channel  which is commonly employed in the quantum causality literature (see for instance \cite{Chiribella2009, Araujo2015}). Here we formulate the lemma in terms the non-signalling constraint on channels which is of relevance to us.
\begin{lem}
\label{lem:choi}
    $S$ does not signal to $R'$ in a linear CPTP map $\mathcal{E}_{SR\rightarrow S'R'}: \mathcal{S}(SR) \rightarrow \mathcal{S}(S'R')$ on finite-dimensional quantum systems $S$, $R$, $S'$ and $R'$ if and only if
    \begin{equation}
        \label{eq: choi_condition}
        \tr_{S'}\left[\mathcal{C}(\mathcal{E}_{SR\rightarrow S'R'})\right]= \frac{\mathds{1}_{\bar{S}}}{d_{\bar{S}}}\otimes \tr_{\bar{S}S'}\left[\mathcal{C}(\mathcal{E}_{SR\rightarrow S'R'})\right],
    \end{equation}
    where $\mathcal{C}(\mathcal{E}_{SR\rightarrow S'R'})$ is the Choi state of the map $\mathcal{E}_{SR\rightarrow S'R'}$, given by \cref{eq: choi}.
\end{lem}
\begin{proof} See Appendix \ref{sec:proof_choi}.
\end{proof}

The above form of the non-signalling condition in terms of the Choi state \cref{eq: choi_condition} derives its usefulness from the fact that it no longer involves any quantifiers, in contrast to \cref{eq: nonsignalling1} and \cref{eq: nonsignalling2}. Furthermore, it is a linear constraint on the Choi state. Both these features are beneficial for conveniently encoding relativistic constraints within the numerical procedure of our QKD security proofs, as we will see in \cref{sec:security_rel_qkd}. 

It is important to note that while relativistic principles associated with a spacetime may motivate us to impose certain non-signalling conditions (e.g., between $S$ and $R'$ when they are spacelike separated), these conditions are independent of the spacetime locations of the systems involved and rely on the information-theoretic structure of the associated CPTP map.
Thus, such no-signalling conditions may be of interest even in scenarios where $S$ and $R'$ are timelike separated, but where we wish to restrict the information flow from $S$ to $R'$.

The above results are relevant for a single round of the cryptographic protocols that we consider in this paper. In order to prove security against general attacks for these protocols, we need to study non-signalling conditions at the level of the full attack channel of the adversary over multiple rounds. The following theorem will be important for this purpose, as it will enable us to decompose the full attack channel into the sequential form required for proving security through the generalised EAT (\cref{def:eat_channel}).

\begin{thm}
\label{thm:causality_sequence}
Consider a linear CPTP map $\mathcal{E}:\mathcal{S}(E_0S_1S_2...S_n)\rightarrow \mathcal{S}(S_1'S_2'...S_n'E_n)$. If $S_i$ does not signal to $S_1'S_2'...S_{i-1}'$ for all $i\in\{2,...,n\}$, then $\mathcal{E}$ admits a decomposition of the form shown in \cref{fig:causality_sequence}, i.e., there exists a set of CPTP maps $\{\mathcal{E}_j:\mathcal{S}(E_{j-1}S_j)\rightarrow \mathcal{S}(S_j'E_j)\}_{j=1}^n$ (for some systems $E_1,...,E_{n-1}$) such that 
\begin{equation}
\label{eq: causality_sequence}
    \mathcal{E}=\big(\mathcal{I}_{E_0S_1...S_{n-1}}\otimes\mathcal{E}_n\big)\circ\dots\circ \big(\mathcal{I}_{E_0S_1}\otimes \mathcal{E}_2\otimes \mathcal{I}_{S_3...S_n}\big)\circ\big(\mathcal{E}_1\otimes \mathcal{I}_{S_2...S_n}\big),
\end{equation}
where the sequence in between the $\mathcal{E}_2$ and $\mathcal{E}_n$ terms consists of the remaining maps $\mathcal{E}_3$,...,$\mathcal{E}_{n-1}$ (in that order) appropriately tensored with identities on the remaining systems.
\end{thm}
\begin{proof} See Appendix \ref{sec:proof_causality_sequence}.
\end{proof}

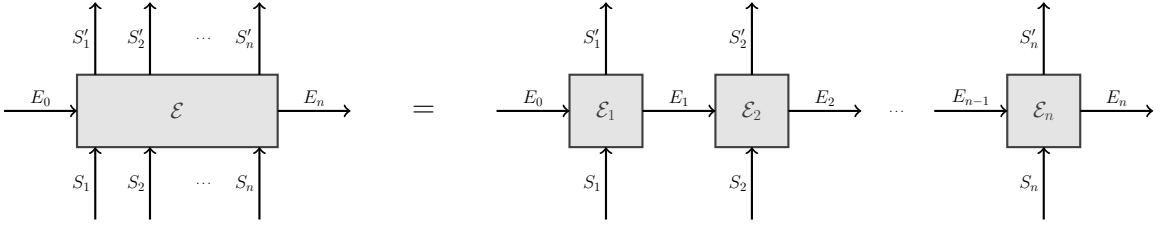
\begin{figure}
    \centering
    \begin{tikzpicture}[trace/.pic={\draw [thick](-0.4,0)--(0.4,0);\draw [thick](-0.3,0.1)--(0.3,0.1);\draw [thick](-0.2,0.2)--(0.2,0.2);\draw [thick](-0.1,0.3)--(0.1,0.3);}, scale=0.48, transform shape]
 
\draw[thick, fill=LightGray, opacity=0.75 ] (0,0) rectangle node[align=center]{\LARGE{$\mathcal{E}$}} (5.5,2);

 \draw[thick,black,->] (0.5,-2)--node[anchor=east]{\Large$S_1$}(0.5,0); 

  \draw[thick,black,->] (2,-2)--node[anchor=east]{\Large$S_2$}(2,0); 
  
 \draw[thick,black,->] (5,-2)--node[anchor=east]{\Large$S_n$}(5,0); 

\node at (3.5,-1) {$\dots$};

  \draw[thick,black,->] (0.5,2)--node[anchor=east]{\Large$S_1'$}(0.5,4); 

  \draw[thick,black,->] (2,2)--node[anchor=east]{\Large$S_2'$}(2,4); 
  
 \draw[thick,black,->] (5,2)--node[anchor=east]{\Large$S_n'$}(5,4); 

 \node at (3.5,3) {$\dots$};

  \draw[thick,black,->] (-2,1)--node[anchor=south]{\Large$E_0$}(0,1); 
    \draw[thick,black,->] (5.5,1)--node[anchor=south]{\Large$E_n$}(7.5,1); 

    \node[align=center, black] at (9.5,1) {\Huge{$\mathbf{=}$}};

  \draw[thick, fill=LightGray, opacity=0.75 ] (13.5,0) rectangle node[align=center]{\LARGE{$\mathcal{E}_1$}} (15.5,2);  
  \draw[thick,black,->] (11.5,1)--node[anchor=south]{\Large$E_0$}(13.5,1); 
    \draw[thick,black,->] (14.5,-2)--node[anchor=east]{\Large$S_1$}(14.5,0); 
 \draw[thick,black,->] (14.5,2)--node[anchor=east]{\Large$S_1'$}(14.5,4); 

   \draw[thick, fill=LightGray, opacity=0.75 ] (17.5,0) rectangle node[align=center]{\LARGE{$\mathcal{E}_2$}} (19.5,2);   \draw[thick,black,->] (15.5,1)--node[anchor=south]{\Large$E_1$}(17.5,1); 
 \draw[thick,black,->] (18.5,-2)--node[anchor=east]{\Large$S_2$}(18.5,0); 
 \draw[thick,black,->] (18.5,2)--node[anchor=east]{\Large$S_2'$}(18.5,4); 

\draw[thick,black,->] (19.5,1)--node[anchor=south]{\Large$E_2$}(21.5,1); 
\node at (22.5,1) {$\dots$};
\draw[thick,black,->] (23.5,1)--node[anchor=south]{\Large$E_{n-1}$}(25.5,1);

  \draw[thick, fill=LightGray, opacity=0.75 ] (25.5,0) rectangle node[align=center]{\LARGE{$\mathcal{E}_n$}} (27.5,2);  
    \draw[thick,black,->] (26.5,-2)--node[anchor=east]{\Large$S_n$}(26.5,0); 
 \draw[thick,black,->] (26.5,2)--node[anchor=east]{\Large$S_n'$}(26.5,4); 
 \draw[thick,black,->] (27.5,1)--node[anchor=south]{\Large$E_{n}$}(29.5,1);

\end{tikzpicture}
    \caption{Diagrammatic representation of the sequence decomposition of the map $\mathcal{E}$ given in \cref{eq: causality_sequence}.}
    \label{fig:causality_sequence}
\end{figure}

\subsection{Quantum optics}
\label{sec:quantum_optics}

In this section, we turn to more practical considerations when studying photonic implementations of QKD protocols. In particular, we introduce the necessary background knowledge required to formulate the photonic implementations of our protocols. Our protocol will make use of the two most common devices in quantum optics, namely the beam splitter (BS) and the threshold detector. Therefore, in the following, we introduce the mathematical language required to describe the operation of these two devices.


The first of these components is the beam splitter. A $50/50$ BS can be described using the following transformation of creation operators: \\
\begin{minipage}{0.3\textwidth}
    \begin{flushright}
        \begin{tikzpicture}
            \draw[thick,fill=gray!30] (-0.4,-0.4) rectangle (0.4,0.4);
            \draw[thick] (-0.4,+0.4) -- (0.4,-0.4);
            \begin{scope}[very thick,decoration={
                markings,
                mark=at position 0.5 with {\arrow{>}}}
                ] 
                \draw[myred,postaction={decorate},>=stealth] (-1,0) node[left,text=black] {$S$} -- (0,0);
                \draw[myred,postaction={decorate},>=stealth] (0,1) node[above,text=black] {$R$} -- (0,0);
                \draw[myred,->,>=stealth] (0,0) -- (1,0) node[right,text=black] {$A$};
                \draw[myred,->,>=stealth] (0,0) -- (0,-1) node[below,text=black] {$B$};
            \end{scope}
        \end{tikzpicture}
    \end{flushright}
\end{minipage}
\begin{minipage}{0.7\textwidth}
    \begin{equation}
    \begin{aligned}
        a^\dagger_A &= \frac{1}{\sqrt{2}}\left(a_S^\dagger + a_R^\dagger\right), \\
        a^\dagger_B &= \frac{1}{\sqrt{2}}\left(a_S^\dagger - a_R^\dagger\right),
    \end{aligned}
    \end{equation}
\end{minipage}
where the modes are as shown in the picture. These operators create states with photons in a given mode ($S$, $R$, $A$ or $B$). This allows us to view the above transformations as a transformation acting on states by writing
\begin{equation}
\begin{aligned}
    \ket{N, 0}_{AB} &= \frac{1}{\sqrt{N!}} \left(a_A^\dagger\right)^N \ket{0,0} = \frac{1}{\sqrt{2^N N!}}\left(a_S^\dagger + a_R^\dagger\right)^N\ket{0,0}, \\
    \ket{0, N}_{AB} &= \frac{1}{\sqrt{N!}} \left(a_B^\dagger\right)^N \ket{0,0} = \frac{1}{\sqrt{2^N N!}}\left(a_S^\dagger - a_R^\dagger\right)^N\ket{0,0}.
    \label{eq:BS_state_trsf}
\end{aligned}
\end{equation}
Of particular importance is the action of a BS on a coherent state:
\begin{equation}
    \ket{\alpha} = e^{-|\alpha|^2/2} \sum_{n=0}^\infty \frac{\alpha^n}{\sqrt{n!}}\ket{n}.
\end{equation}
The parameter $\alpha\in\mathbb{C}$ quantifies the amplitude of a laser pulse, as the state has an average photon number $\langle n \rangle_{\ket{\alpha}}=|\alpha|^2$. Under the action of the BS, coherent states are mapped to coherent states, more precisely
\begin{equation}
    \ket{\alpha}_S\ket{\beta}_R \mapsto \ket{(\alpha+\beta)/\sqrt{2}}_A \otimes \ket{(\alpha-\beta)/\sqrt{2}}_B.
\label{eq:BS_coherent_states}
\end{equation}
Furthermore, the probability of detecting no photon when measuring a coherent state $\ket{\alpha}$ is given by 
	\begin{equation}
		\mathrm{Pr}[n=0|\alpha] = \left|\braket{0}{\alpha}\right|^2 = e^{-|\alpha|^2}.
		\label{eq:loss}
	\end{equation}
When a coherent state $\ket{\alpha}$ travels through a lossy beam line with transmittance $\eta\in[0,1]$, it is mapped to the state $\ket{\sqrt{\eta}\alpha}$.

We will be interested in the scenario where two threshold detectors are placed at the output ports of a beam splitter. There are four different measurement outcomes: only the first detector clicks, only the second detector clicks, both detectors click, or neither detector clicks. These four events are described by the following POVM:
\begin{equation}
\begin{aligned}
    \tilde{M}^{(0)} &= \sum_{N=1}^{\infty} \ketbra{N, 0}{N, 0}_{AB}, \\
    \tilde{M}^{(1)} &= \sum_{N=1}^{\infty} \ketbra{0, N}{0, N}_{AB}, \\
    \tilde{M}^{(\mathrm{dc})} &= \sum_{N=2}^{\infty} \sum_{n=1}^{N-1} \ketbra{N - n, n}{N - n, n}_{AB}, \\
    \tilde{M}^{(\bot)} &= \ketbra{0,0}{0,0} = \mathds{1} - \tilde{M}^{(0)} - \tilde{M}^{(1)} - \tilde{M}^{(\mathrm{dc})}.
\end{aligned}
\end{equation}
Since the pair of threshold detectors is placed behind a BS, the states in the above expressions can also be understood as photonic states in the respective mode before the BS (compare with \cref{eq:BS_state_trsf}). In our protocols we will post-process the measurement outcomes by randomly assigning the double-click outcomes to the outcome $0$ or $1$. We describe this via the post-processed POVM:
\begin{equation}
\label{eq:mmt_ops_Bob}
\begin{aligned}
    M^{(0)} &= \tilde{M}^{(0)} + \frac{1}{2} \tilde{M}^{(\mathrm{dc})}, \\
    M^{(1)} &= \tilde{M}^{(1)} + \frac{1}{2} \tilde{M}^{(\mathrm{dc})}, \\
    M^{(\bot)} &= \tilde{M}^{(\bot)}.
\end{aligned}
\end{equation}
Studying photonic implementations directly is infeasible due to the infinite dimension of the Fock space. This issue can be addressed using a theoretical tool known as squashing maps \cite{Gottesman2002, Tsurumaru2008, Beaudry2008, Gittsovich2014}:

\begin{defn}[Squashing map] \label{def:squashing} Let $A$ and $A'$ be two quantum systems. Let $\{M_A^{(x)}\}_x$ and $\{N_{A'}^{(x)}\}_x$ be two POVMs for the systems $A$ and $A'$, respectively. A CPTP map $\Lambda: \mathcal{S}(A) \rightarrow \mathcal{S}(A')$ is called  a \emph{squashing map from $\{M_A^{(x)}\}_x$ to $\{N_{A'}^{(x)}\}_x$} if for all $x$ and all $\rho_A \in \mathcal{S}(A)$, 
\begin{align}
    \tr \left[ M_A^{(x)} \rho_A \right] = \tr \left[ N_{A'}^{(x)} \Lambda(\rho_A) \right].
\end{align}
\end{defn}

A squashing map is an essential tool in our security proof because it allows us to reduce the analysis of photonic QKD implementations to qubit-based implementations, which are much simpler to analyse. The argument goes as follows: Since the measurement statistics are preserved by the squashing map, introducing an artificial squashing map before Bob's detectors does not change the post-measurement state. We can now see this squashing map as part of Eve's attack channel. Hence, any attack on the large system implies an equally strong attack on the reduced system. Thus, the key rate of the qubit protocol is a lower bound on the key rate of the original protocol. Finding a squashing map for the given detectors is usually good enough for security proofs. Unfortunately, this is not sufficient in our case because we have a non-signalling constraint on Eve's attack, which needs to be preserved by the squashing map. Therefore, we want a squashing map that is non-signalling. This is achieved by the following theorem:

\begin{thm} \label{thm:squashing} 
Let $S$ and $R$ be two Fock spaces and let $S'$ and $R'$ be two qubit systems. The target POVM is
\begin{equation}
    \begin{aligned}
        \label{eq:mmt_ops_Bob_qubits}
        N_{S'R'}^{(0)} &= \ketbra{\phi^+}{\phi^+} + \frac{1}{2}\ketbra{11}{11}, \\
        N_{S'R'}^{(1)} &= \ketbra{\phi^-}{\phi^-} + \frac{1}{2}\ketbra{11}{11}, \\
        N_{S'R'}^{(\bot)} &= \ketbra{00}{00},
    \end{aligned}
\end{equation}
    where $\ket{\phi^\pm} = (\ket{01} \pm \ket{10})/\sqrt{2}$ are Bell states. Note that the POVM above is simply the restriction of the POVM defined in \cref{eq:mmt_ops_Bob} to the one-photon subspaces. Let $N \in \mathbb{N}_{>0}$, $0 \leq k \leq N$ and $0 \leq l \leq N$, where $k$ is an odd and $l$ is an even number in $\mathbb{N}$.
    Define the Kraus operators
    \begin{align}
        K^{(0)} &= \ket{00}_{S'R'}\bra{0,0}_{SR}, \\
        K^{(N)}_{k,l} &= \frac{\sqrt{2}}{\sqrt{2^N}} \left( \sqrt{\binom{N}{l}} \ket{01}_{S'R'}\bra{N-k, k}_{SR} + \sqrt{\binom{N}{k}} \ket{10}_{S'R'}\bra{N-l, l}_{SR} \right).
    \end{align}
    Then the map
    \begin{align}
        \Lambda(\rho_{SR}) = K^{(0)} \rho_{SR} \left(K^{(0)}\right)^* + \sum_{N = 1}^{\infty} \sum_{k, l}^N K^{(N)}_{k,l} \rho_{SR} \left( K^{(N)}_{k,l} \right)^*
    \end{align}
    is a squashing map from the POVM given in \cref{eq:mmt_ops_Bob} to the POVM given in \cref{eq:mmt_ops_Bob_qubits}. Furthermore $\Lambda$ is non-signalling from $S$ to $R'$.
\end{thm}
\begin{proof} See \cref{sec:proof_squashing}.
\end{proof}



\begin{figure}[t]
	\centering
	\newcommand{\BS}[1]{
		\begin{scope}[shift={(#1)}]
			\draw[fill=gray!30] (-0.25,-0.25) rectangle (0.25,0.25);
			\draw (-0.25,-0.25) rectangle (0.25,0.25);
			\draw (-0.25,-0.25) -- (0.25,0.25);
		\end{scope}
	}
	\begin{tikzpicture}[
		boxnode/.style={rectangle, draw=black, minimum width=0.8cm, thick, fill=black!20},
		emptynode/.style={rectangle},
		mirror/.pic={
			\draw[decorate,decoration={markings, mark=between positions 0.015 and 0.98 step 0.1072 with {
					\draw (0,0)--(90:3pt);}}] (-0.2,-0.2) -- (0.2,0.2);
			\draw[thick] (-0.2,-0.2) -- (0.2,0.2);
		},
		detector/.pic={\draw[draw=black,fill=gray!20] (0,0.2) arc(90:-90:0.2cm and 0.2cm) -- cycle;},
		scale=0.9,
		]
		\draw[draw=myblue,dashed] (0, 0) rectangle (4, 4.5);
		\node[text=myblue,align=left] at (0.8, 3.9) {Alice's\\lab};
		
		\draw[draw=mygreen,dashed] (4.5, 0) rectangle (7, 4.5);
		\node[text=mygreen,align=left] at (5.8, 2) {Eve};
		
		\draw[draw=myblue,dashed] (7.5, 0) rectangle (11, 5.5);
		\node[text=myblue,align=right] at (10.3, 0.6) {Bob's\\lab};
		
		\node[boxnode] (laser) at (1, 0.5) {Source};
		
		\node[emptynode] (BSAlice) at (3,0.5) {};
		\BS{BSAlice.center};
		
		\node[emptynode] (MirrorAlice) at (3, 3.5) {};
		\draw (MirrorAlice.center) pic {mirror};
		
		\draw[thick,myred] (laser.east) -- (BSAlice.center);
		\draw[thick,myred] (BSAlice.center) -- (MirrorAlice.center);
		
		\node[boxnode] (PM) at (3, 1.4) {PM};
		\draw[thick,->,>=stealth] ([xshift=-0.5cm] PM.west) node[left]{$V_i$} -- (PM.west);
		
		\draw[thick,myred] (2.5, 2.2) ellipse (0.5 and 0.3);
		\draw[thick,myred] (2.5, 2.4) ellipse (0.5 and 0.3);
		\draw[thick,myred] (2.5, 2.6) ellipse (0.5 and 0.3);
		\draw[thick,myred] (2.5, 2.8) ellipse (0.5 and 0.3);
		\node[text=myred] at (1.3, 2.5) {Delay};
		
		\node[emptynode] (MirrorBob) at (8.5, 0.5) {};
		\draw (MirrorBob.center) pic[rotate=180] {mirror};
		
		\node[emptynode] (BSBob) at (8.5,3.5) {};
		\BS{BSBob.center}
		
		\draw[thick,myred] (MirrorBob.center) -- (BSBob.center);
		
		\draw[thick,myred] (9, 1.8) ellipse (0.5 and 0.3);
		\draw[thick,myred] (9, 2.0) ellipse (0.5 and 0.3);
		\draw[thick,myred] (9, 2.2) ellipse (0.5 and 0.3);
		\draw[thick,myred] (9, 2.4) ellipse (0.5 and 0.3);
		\node[text=myred] at (10.2, 2.1) {Delay};
		
		\node[emptynode] (DetectorBob1) at ([xshift=1cm] BSBob.east) {};
		\draw (DetectorBob1.center) pic{detector} node[right=0.2cm]{``0''};
		\draw[thick,myred] (BSBob.center) -- (DetectorBob1.center);
		
		\node[emptynode] (DetectorBob2) at ([yshift=1cm] BSBob.north) {};
		\draw (DetectorBob2.center) pic[rotate=90]{detector} node[above=0.2cm]{``1''};
		\draw[thick,myred] (BSBob.center) -- (DetectorBob2.center);
		
		\draw[thick,myred] (MirrorAlice.center) -- node[above]{$\ket{(-1)^{V_i}\alpha}_S$} (BSBob.center);
		\draw[thick,myred] (BSAlice.center) -- node[above]{$\ket{\alpha}_R$} (MirrorBob.center);
	\end{tikzpicture}
	\caption{\label{fig:rel_QKD_setup} The experimental setup of our novel relativistic QKD protocol, which boils down to a shared Mach-Zehnder interferometer between Alice and Bob. In each round of the protocol, Alice chooses a uniformly random bit $V_i\in\{0,1\}$. She then sends a weak coherent state through a BS, creating a reference state and a signal state. The reference state $\ket{\alpha}_R$ is sent to Bob immediately. Additionally, Alice uses a phase modulator (PM) to apply a phase $(-1)^{V_i}$ to the signal state $\ket{\alpha}_S$, producing the state $\ket{(-1)^{V_i} \alpha}_S$. This state is delayed by a time $\Delta t$ before Alice sends it to Bob (depicted via a delay line). Bob correspondingly first receives the reference state and delays it by the same amount $\Delta t$. Upon receiving the signal state, he measures the relative phase between the reference and signal state using a BS on his side.}
\end{figure}
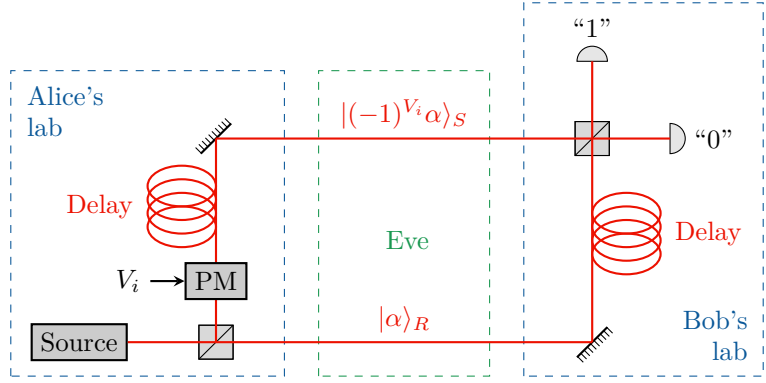

\section{Security of relativistic QKD}
\label{sec:rel_qkd}

As a first step towards proving the security of DPS QKD, we introduce a novel relativistic QKD protocol, based on ideas from \cite{Radchenko2014,Kravtsov2018}. This serves two purposes: As we will show in \cref{sec:SecurityDPS}, the security of DPS QKD is inherited from the security of the relativistic QKD protocol. The reason is that both protocols share the same measurement operators and similar relativistic constraints, even if their experimental setups differ. Secondly, the relativistic QKD protocol we introduce in this section may be of independent interest since it comes with a complete security proof against general attacks.

The setup of the relativistic QKD protocol is as follows: Broadly speaking, Alice and Bob share a Mach-Zehnder interferometer with two delay lines as shown in \cref{fig:rel_QKD_setup}. A single round of the protocol contains the following steps: At the time $t_A^{(i)}$ Alice prepares two states, a weak coherent \emph{reference} pulse $\ket{\alpha}_R$ that she immediately sends to Bob, and a weak coherent \emph{signal} state that she delays by a time $\Delta t$ before sending it to Bob. Alice encodes her uniformly random raw key bit $V_i \in \{0,1\}$ in the phase of the signal state, i.e., she sends $\ket{(-1)^{V_i}\alpha}_S$ to Bob. The delay $\Delta t$ is chosen such that the following condition holds: 

\begin{cond}
	\label{cond:relativistic}
	Eve does not signal from the signal state to the reference state.
\end{cond}

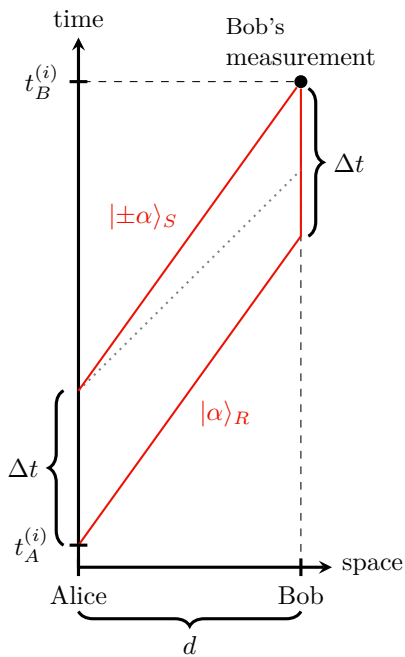
\begin{wrapfigure}[34]{l}{0.4\textwidth}
		\centering
		\resizebox{0.37\textwidth}{!}{
			\begin{tikzpicture}[
				boxnode/.style={rectangle, draw=black, minimum width=0.8cm, thick, fill=black!20},
				emptynode/.style={rectangle},
				mirror/.pic={
					\draw[decorate,decoration={markings, mark=between positions 0.015 and 0.98 step 0.1072 with {
							\draw (0,0)--(90:3pt);}}] (-0.2,-0.2) -- (0.2,0.2);
					\draw[thick] (-0.2,-0.2) -- (0.2,0.2);
				},
				detector/.pic={\draw[draw=black,fill=gray!20] (0,0.2) arc(90:-90:0.2cm and 0.2cm) -- cycle;},
				scale=0.6,
				]
				\draw[->,>=stealth,very thick] (0, 0) -- (5.7, 0) node[right]{space};
				\draw[->,>=stealth,very thick] (0, 0) -- (0, 12) node[above]{time};
				
				\draw[dashed] (5, 0) -- (5, 10);
				\draw[very thick] (0, -0.2) node[below]{Alice} -- (0, 0.2);
				\draw[very thick] (5, -0.2) node[below]{Bob} -- (5, 0.2);
				
				\node[circle, fill=black, inner sep=0em, minimum size=0.5em] (mmtBob) at (5, 11) {};
				\draw[] (mmtBob) node[above=0.3em, align=left]{\small Bob's\\measurement};
				
				\draw[thick, myred] (0, 0.5) -- node[below right]{$\ket{\alpha}_R$} ++(5, 7) -- (mmtBob);
				\draw[dotted, thick, gray] (0, 4) -- (5, 9.0);
				\draw[thick, myred] (0, 4) -- node[above left]{$\ket{\pm \alpha}_S$} (mmtBob);
				\draw[dashed] (mmtBob) -- ++(-5,0);
				
				\draw[very thick, decorate, decoration={brace,amplitude=0.2cm,raise=0.2cm}] (0, 0.5) -- (0, 4) node[pos=0.5,left=0.4cm]{$\Delta t$};
				\draw[very thick, decorate, decoration={brace,amplitude=0.2cm,raise=0.6cm}] (5, 0) -- (0, 0) node[pos=0.5,below=0.8cm]{$d$};
				\draw[very thick, decorate, decoration={brace,amplitude=0.2cm,raise=0.1cm}] (mmtBob) -- (5, 7.5) node[pos=0.5,right=0.3cm]{$\Delta t$};
				
				\draw[very thick] (0.2, 0.5) -- (-0.2, 0.5) node[left=0.4em]{$t_A^{(i)}$};
				\draw[very thick] (0.2, 11) -- (-0.2, 11) node[left]{$t_B^{(i)}$};
		\end{tikzpicture}}
		\caption{A spacetime diagram depicting how \protect{\cref{cond:relativistic}} can be enforced. Alice's lab is depicted as the left world line, and Bob's lab is separated by a distance $d$. The (red) world lines of the (not necessarily lightlike) reference state $\ket{\alpha}_R$ and the signal state $\ket{\pm \alpha}_S$ are separated by the time shift $\Delta t$. The dotted line depicts the future light cone of Alice revealing information about the signal state. For large enough $\Delta t$, Eve therefore can't influence the reference based on this information before it enters Bob's lab.}
		\label{fig:rel_QKD_spacetime}
\end{wrapfigure}

\cref{fig:rel_QKD_spacetime} shows how \cref{cond:relativistic} can be enforced in an experimental setup via the delay of the signal by $\Delta t$, which is implemented in \cref{fig:rel_QKD_setup} via the delay lines. In \cref{sec:seq_condition}, we further elaborate on how to enforce this condition by choosing an appropriate value for $\Delta t$.
In line with the concepts introduced in \cref{sec:causality}, we interpret this condition as a non-signalling constraint on Eve's possible attacks. 

After the delay, Alice sends the modulated signal pulse $\ket{(-1)^{V_i}\alpha}_S$ to Bob. He correspondingly first receives the reference state, which he delays by the same amount $\Delta t$. At time $t_B$, he receives the signal state and interferes both states through his own BS. Using \cref{eq:BS_coherent_states}, this transformation can be written as
\begin{equation}
\begin{aligned}
	\ket{(-1)^0 \alpha}_S \ket{\alpha}_R &\mapsto \ket{+ \sqrt{2}\alpha}_A \otimes \ket{0}_B, \\
	\ket{(-1)^1 \alpha}_S \ket{\alpha}_R &\mapsto \ket{0}_A \otimes \ket{-\sqrt{2}\alpha}_B. 
    \label{eq:BS_trsf_rel}
\end{aligned}
\end{equation}
We see that, depending on the phase of the signal state, only one of Bob's detectors will click. This then allows Bob to recover Alice's raw key bit $V_i$. Bob's measurement can be seen as optimal unambiguous state discrimination between $\ket{\pm\alpha}_S$. Bob also records the time $t_B$ at which his detector clicked. 

If both detectors click due to the presence of noise or the interaction of an adversary, Bob randomly reassigns the measurement outcome to either $0$ or $1$. This leaves him with the measurement operators as defined in \cref{eq:mmt_ops_Bob}, where single detector-click outcomes $0$ or $1$ correspond to his guess for Alice's raw key bit $V_i$, and the inconclusive outcome $\bot$ represents that no detector has clicked.

Afterwards, Alice communicates the time $t_A^{(i)}$ at which she dispatched the reference state over an authenticated classical channel to Bob. If Bob determines the time of interference $t_B^{(i)}$ to be above a threshold given explicitly by $t_A^{(i)} + 2\Delta t + d/c$, he aborts the protocol\footnote{It may appear drastic to abort the whole protocol if the timing of a single round was off. This, however, only allows Eve to abort the protocol at her will, which is a possibility she has in any QKD protocol. For instance, she could block the quantum transmission line such that none of Alice's states arrive at Bob's lab. In practice, one should only count detector clicks up to $t_B^{(i),\max} = t_A^{(i)} + 2\Delta t + d/c$ to realize \cref{cond:relativistic}.}. Through this abort condition, Alice and Bob know that \cref{cond:relativistic} is satisfied if the protocol didn't abort.

To be able to apply the generalised EAT in the security proof of the relativistic QKD protocol, it is necessary that the assumptions of the theorem are fulfilled. In particular, we have to enforce a sequential form of our protocol which is formalized through the following additional condition:

\begin{cond} 
	\label{cond:sequential_relQKD}
	Eve does not signal from round $i+1$ to the rounds $1, \ldots, i$.
\end{cond}
This condition is easily satisfied by requiring that Alice starts the $i+1$-th round at a time $t_A^{(i+1)} = t_A^{(i)} + 2\Delta t$ by the same argument we used to ensure that \cref{cond:relativistic} is fulfilled (see \cref{sec:seq_condition}). Conceptually, Alice should therefore send a (reference or signal) pulse every $\Delta t$, as agreed upon by Alice and Bob beforehand.

\subsection{Protocol}

Next, we formalise the protocol as described above and include the classical post-processing steps after repeating $n\in\mathbb{N}$ rounds of the protocol.

For a fraction $\gamma\in(0,1)$ of the rounds, Bob publicly announces his measurement outcome $B_i$, which allows Alice to compute statistics in order to upper-bound Eve's knowledge. We refer to these rounds as \emph{test rounds}. These statistics take values in the alphabet $\mathcal{C} = \{\mathrm{corr}, \mathrm{err}, \bot, \varnothing\}$. The first three correspond to Bob determining the value for Alice's raw key bit $V_i$ correctly, incorrectly, or not at all, respectively. The value $\varnothing$ denotes that the round was not a test round and hence no statistics have been collected. Correspondingly, we define the evaluation function $\mathrm{EV}: \{0,1, \bot\} \times \{0,1,\bot,\varnothing\} \rightarrow \mathcal{C}$ with inputs from Alice's raw key bit $A$ and Bob's measurement outcome $J$ as
\begin{align}
    \mathrm{EV}(A,J) = \begin{cases}
        \mathrm{corr}, & \text{if $J \in \{0, 1 \}$ and $A = J$} \\
        \mathrm{err}, & \text{if $J \in \{0, 1 \}$ and $A \neq J$} \\
        \bot, & \text{if $J = \bot$} \\
        \varnothing, & \text{if $J = \varnothing$}.
    \end{cases}
    \label{eq:eval}
\end{align}
With these definitions we are now able to formally state the relativistic QKD protocol with $\Delta t$ chosen such that \cref{cond:relativistic,cond:sequential_relQKD} are satisfied as described in \cref{sec:seq_condition}. The structure of this protocol (summarised in~\hyperref[protocol:relativisticQKD]{Protocol 1}) is then the same one as the general prepare-and-measure protocol in \cite{Metger2023}.

\subsection{Sketch of security proof}
\label{sec:proofsketch}
Here, we present a brief sketch of the security proof of the relativistic QKD protocol. The interested reader can find the details of the proof in \cref{sec:security_general,sec:security_rel_qkd}.
The main steps of the security proof can be summarised as follows:
\begin{enumerate}
	\item Cast the soundness condition into a form that matches the conditions of the leftover hashing lemma (\cref{lem:leftoverhashing}). This lemma ensures that the trace-distance between the ideal state and the state that describes the actual protocol can be upper-bounded, given a lower bound on the smooth min-entropy.
    \item Ensure that all requirements for applying the generalised EAT are fulfilled: via appropriate entropic chain rules, we can to bring the smooth min-entropy into the form that appears in the generalised EAT, and \cref{cond:sequential_relQKD} together with \cref{thm:causality_sequence} ensures the existence of well-defined EAT channels $\mathcal{M}_i$.
    \item To get a bound on $H_\mathrm{min}^\varepsilon$ out of the generalised EAT we need a min-tradeoff function. This requires lower-bounding Eve's uncertainty about the raw key, i.e., finding a lower-bound on $H(A|EIJ)$.
    \begin{enumerate}
	    \item Use \cref{cond:relativistic} and \cref{thm:squashing} to squash the relativistic protocol into a qubit protocol that still satisfies \cref{cond:relativistic} (as Eve could have applied the squashing map herself). The measurement operators of the squashed protocol are then given by \cref{eq:mmt_ops_Bob_qubits}.
        \item The numerical optimization requires us to minimize the conditional entropy over all possible attacks of Eve that are non-signalling (compare \cref{defn:nonsignalling}). This can be conveniently included in the optimization constraints by optimizing over Choi states and applying \cref{lem:choi}.
    \end{enumerate}
\end{enumerate}

\begin{figure}[H]
	\captionsetup{name=Box}
	\begin{mdframed}[style=MyFrame]
		\small
		{\large \textsf{Protocol 1: Relativistic QKD}}\label{protocol:relativisticQKD}\\[0.2cm]
		The protocol is defined in terms of the following parameters, which are chosen before the protocol begins: \\[0.1cm]
		\begin{tabular}{rl}
			$\alpha \in \mathbb{C}$: & amplitude of the laser light \\
			$n\in\mathbb{N}$: & number of protocol rounds \\
			$\gamma\in (0,1)$: & testing frequency \\
			$\mathrm{leak_{EC}}$: & maximum length of error correction \\
			$\varepsilon_\mathrm{EC}$: & error tolerance during error correction \\
			$f: \mathds{P}_\mathcal{C} \rightarrow \mathds{R}$: & collective attack bound \\
			$H_\mathrm{exp}$: & minimum expected single-round entropy \\
			$l \in \mathbb{N}$: & length of the final secret key \\
		\end{tabular}\\[0.5cm]
		\hrule\vspace{10pt}
		\begin{enumerate}
			\item \textbf{Quantum Phase:} For $i\in [n]$:
			\begin{enumerate}
				\item Alice chooses a bit $V_i \in \{0, 1\}$ uniformly at random, prepares a reference state $\ket{\alpha}_R$ and a signal state $\ket{(-1)^{V_i} \alpha}_S$ and sends them to Bob such that \cref{cond:relativistic} is enforced.
                \item Bob receives a joint state $\rho_{SR}$ and performs a measurement given by the POVM $\{M_{SR}^{(b)}\}_b$ of \cref{eq:mmt_ops_Bob}. He records his measurement outcome in the register $B_i \in \{0, 1, \bot\}$.
				\item If $B_i = \bot$ then Bob sets $I_i = \bot$ and $I_i = \top$ otherwise.
				\item Bob chooses $T_i \in \{0,1\}$ randomly with $\mathrm{Pr}[T_i=1] = \gamma$. If $T_i=1$ Bob records $J_i = B_i$ to Alice and $J_i = \varnothing$ otherwise.
				\item Alice waits to enforce \cref{cond:sequential_relQKD}.
			\end{enumerate}
            \item \textbf{Public announcement:} Bob announces $I^nJ^n$.
			\item \textbf{Sifting:} For all $i \in [n]$ Alice sets $A_i = V_i$ if $I_i \neq \bot$ and $A_i = \bot$ otherwise.
			\item \textbf{Error correction:}
			\begin{enumerate}
				\item Alice and Bob use their outputs $A^{n}$ and $B^{n}$ to perform error correction by communicating at most $\mathrm{leak_{EC}}$ number of bits. Bob stores his guess for Alice's key in $\tilde{A}^{n}$.
				\item Alice chooses a hash function $h \in \mathcal{F}$ uniformly at random from a family of two-universal hash functions of length $\lceil \log(1/\varepsilon_\mathrm{EC}) \rceil$ and applies it to her raw key. She sends the output $h(A^{n})$ and her choice of hash function to Bob.
				\item Bob applies the same hash function to his guess $\tilde{A}^{n}$. If the two hashes disagree, Alice and Bob abort the protocol.
			\end{enumerate}
			\item \textbf{Parameter estimation:} For all $i \in [n]$ Alice computes $C_i=\mathrm{EV}(A_i, J_i)$. If $f({\mathrm{freq}(C^n)}) \leq H_\mathrm{exp}$ they abort the protocol.
			\item \textbf{Privacy amplification:} Alice and Bob perform privacy amplification on $A^{n}$ and $\tilde{A}^{n}$ to obtain raw keys $K_A^l$ and $K_B^l$.
		\end{enumerate}
	\end{mdframed}
\end{figure}

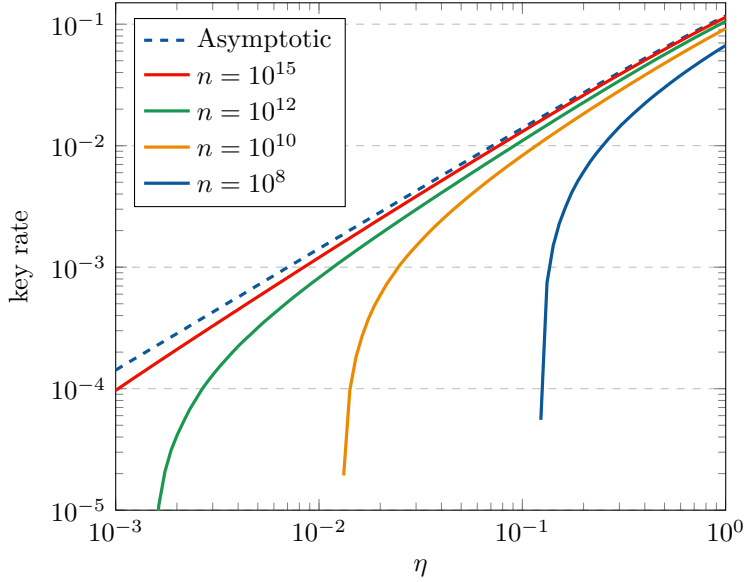
\begin{figure}[t]
	\centering
	\begin{tikzpicture}
		\begin{axis}[
			xmode=log,
			ymode=log,
			width = 0.7\linewidth,
			xlabel={$\eta$},
			ylabel={key rate},
			xmin=1e-3, xmax=1,
			ymin=1e-5, ymax=0.15,
			ymajorgrids=true,
			grid style=dashed,
			legend pos=north west,
			every mark=none,
			cycle list name=mycolors,
			legend cell align={left},
			scale=0.9,
			]
			\addplot+[dashed,very thick,mark=none] table[x=eta, y=rate, col sep=comma]{data/rel_qkd/opt_alphas.csv};
			\addlegendentry{Asymptotic}
			
			\addplot+[very thick,mark=none] table[x=eta, y=min_ent, col sep=comma]{data/rel_qkd/opt_grads_n=1e15.csv};
			\addlegendentry{$n=10^{15}$}
			
			\addplot+[very thick,mark=none] table[x=eta, y=min_ent, col sep=comma]{data/rel_qkd/opt_grads_n=1e12.csv};
			\addlegendentry{$n=10^{12}$}
			
			\addplot+[very thick,mark=none] table[x=eta, y=min_ent, col sep=comma]{data/rel_qkd/opt_grads_n=1e10.csv};
			\addlegendentry{$n=10^{10}$}
			
			\addplot+[very thick,mark=none] table[x=eta, y=min_ent, col sep=comma]{data/rel_qkd/opt_grads_n=1e8.csv};
			\addlegendentry{$n=10^{8}$}
		\end{axis}
	\end{tikzpicture}
	\caption{The key rates for a finite number of rounds $n$ as given in the legend in dependence of different transmittances $\eta\in[0,1]$ of the lossy beam line without considering QBER and for optimal $\alpha$. Note that asymptotically one can distil a secret key for arbitrarily low transmittances.}
	\label{fig:losses_relQKD}
\end{figure}

 \begin{figure}[H]
     \centering
 	\begin{tikzpicture}
 		\begin{axis}[
 			xmode=normal,
 			ymode=log,
 			width = 0.6\textwidth,
 			xlabel={QBER},
 			ylabel={key rate},
 			xmin=0, xmax=0.15,
 			ymin=0, ymax=0.13,
 			ymajorgrids=true,
 			grid style=dashed,
 			legend pos=south west,
 			every mark=none,
 			cycle list name=mycolors,
 			legend cell align={left},
             xticklabel style={
                 /pgf/number format/.cd,
                 fixed,
                 fixed zerofill,
                 precision=2,
             },
             anchor=west,
 		]
 			\addplot+[very thick,mark=none] table[x=q, y=key_rate, col sep=comma]{data/rel_qkd/rates_qber.csv};
            \addlegendentry{$\eta=1$};

 			\addplot+[very thick,mark=none] table[x=q, y=key_rate, col sep=comma]{data/rel_qkd/rates_qber_low_eta.csv};
            \addlegendentry{$\eta=0.1$};

 			\addplot+[very thick,mark=none] table[x=q, y=key_rate, col sep=comma]{data/rel_qkd/rates_qber_very_low_eta.csv};
            \addlegendentry{$\eta=0.01$};
 		\end{axis}
 	\end{tikzpicture}
     \caption{Asymptotic key rate in the limit $n\rightarrow \infty$ for different QBERs and transmittances. We choose $\alpha$ to optimize the key rates. Note that a secret key can be distilled up to a threshold $\mathrm{QBER}\approx13\%$ independent of transmittance.}
 	\label{fig:qber_plot_relQKD}
 \end{figure}
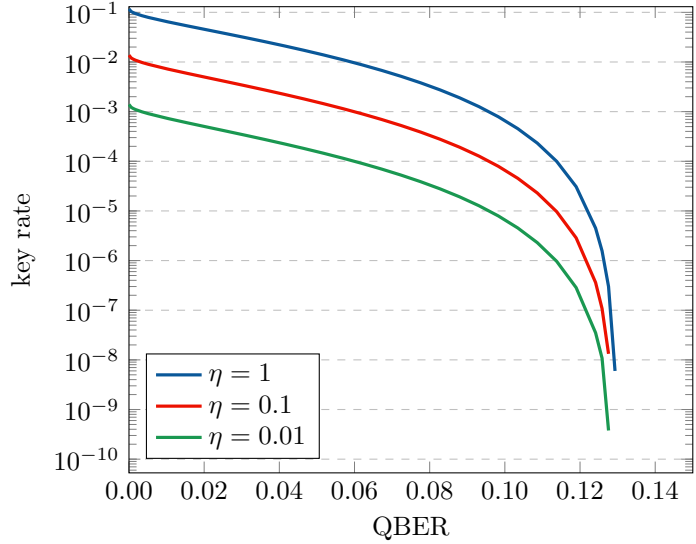

\subsection{Results} \label{sec:rel_qkd_res}

Via the strategy sketched in \cref{sec:proofsketch}. one can compute asymptotic and finite-size key rates of the relativistic QKD protocol under general adversarial attacks that resemble noise. Recall that the key rate is defined as $r=l/n$, where $l$ is the length of the key and $n$ is the number of rounds. Note that there are two laser pulses (reference and signal) per round. The key rate in time is then at most $r/(2\Delta t)$, based on the timing of our protocol.

To further study the behaviour of our protocol under noise, we consider channel losses through a lossy channel with transmittance $\eta\in[0,1]$ and a general quantum-bit error rate $\mathrm{QBER}\in[0,1]$ on the sifted key\footnote{A more in-depth analysis could include effects like detector dark counts.}.

Based on \cref{eq:loss,eq:mmt_ops_Bob,eq:eval,eq:BS_trsf_rel} one finds the statistics for an honest implementation of the protocol with noise to be
\begin{equation}
	\begin{aligned}
		\mathrm{Pr}[\bot|\alpha,\eta,\mathrm{QBER}] &= e^{-2\eta|\alpha|^2}, \\
		\mathrm{Pr}[\mathrm{err}|\alpha,\eta,\mathrm{QBER}] &= \left(1-e^{-2\eta|\alpha|^2}\right)\cdot \mathrm{QBER}, \\
		\mathrm{Pr}[\mathrm{corr}|\alpha,\eta,\mathrm{QBER}] &= \left(1-e^{-2\eta|\alpha|^2}\right)\cdot(1-\mathrm{QBER}),
	\end{aligned}
	\label{eq:statistics}
\end{equation}
which we use as the observed statistics under Eve's attack. The respective key rates can be computed numerically (see \cref{sec:security_rel_qkd}) and are depicted in \cref{fig:qber_plot_relQKD,fig:losses_relQKD}.

The amplitude $\alpha$ of the laser light is chosen to optimize the asymptotic key rates for a given amount of noise. Typical values are $\alpha \approx 0.45$. We emphasize that there are many places in the finite-size analysis where the bound on the key rate could possibly be tightened. The finite-size plots should therefore be viewed only for illustrational purposes. We have chosen a soundness parameter of $\varepsilon^{\mathrm{snd}}=4 \cdot 10^{-12}$ and a completeness parameter of $\varepsilon^{\mathrm{comp}}=10^{-2}$.

An interesting observation is the linear scaling of the asymptotic key rate for the entire parameter range. As a consequence, the asymptotic key rate remains positive for arbitrary amounts of losses. This is important for applications between parties at large distances (i.e., for low transmittances) who aim to establish a secret key. We also highlight that the threshold QBER up to which the protocol stays secure is given by $\mathrm{QBER}\approx 13\%$ and stays there even for $\eta<1$ (compare \cref{fig:qber_plot_relQKD}).

\section{Security of DPS QKD}
\label{sec:SecurityDPS}
In the DPS protocol Alice encodes her raw key in the relative phase between subsequent coherent pulses. Bob then uses a Mach-Zehnder interferometer to measure the relative phase of these pulses to reconstruct Alice's key. This setup is sketched in \cref{fig:dps_setup}. The main motivation for the DPS protocol is that it is both experimentally simple to implement while being resistant against the photon number splitting attack \cite{Bennett1992b,Brassard2000}. The reason for this is that the PNS attack requires Eve to measure the total photon number. This measurement is undetectable by a polarization measurement but does influence the phase coherence (a coherent state is transformed into a mixed state). In this section we present the key steps in proving security of the DPS protocol. The full technical details can be found in \cref{sec:security_general,sec:security_rel_qkd,sec:security_proof_dps}.

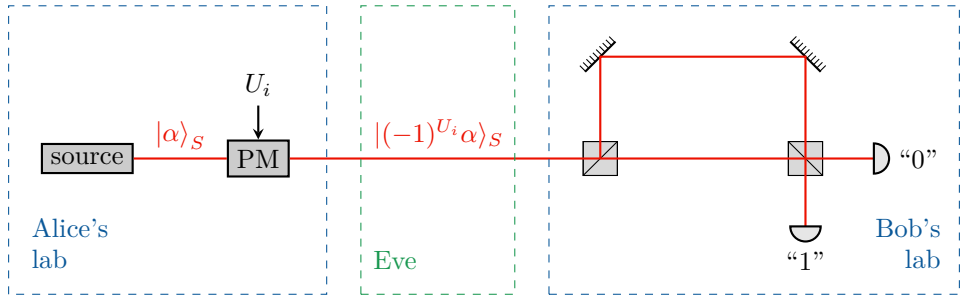
\begin{figure}[t]
	\centering
	\newcommand{\BS}[1]{
    \begin{scope}[shift={(#1)}]
        \draw[fill=gray!30] (-0.25,-0.25) rectangle (0.25,0.25);
        \draw (-0.25,-0.25) rectangle (0.25,0.25);
        \draw (-0.25,-0.25) -- (0.25,0.25);
    \end{scope}
}

\begin{tikzpicture}[
    boxnode/.style={rectangle, draw=black, minimum width=0.8cm, thick, fill=black!20},
    emptynode/.style={},
    mirror/.pic={
        \draw[decorate,decoration={markings, mark=between positions 0.015 and 0.98 step 0.1072 with {
            \draw (0,0)--(90:3pt);}}] (-0.2,-0.2) -- (0.2,0.2);
        \draw[thick] (-0.2,-0.2) -- (0.2,0.2);
    },
    detector/.pic={\draw[draw=black,fill=gray!20] (0,0.2) arc(90:-90:0.2cm and 0.2cm) -- cycle;},
    scale=0.9,
    ]
    \node[boxnode] (source) at (-4.5, 0) {source};
    \node[boxnode] (PM) at (-2, 0) {PM};

    \node (BSBob1) at (3, 0) {};
    \BS{BSBob1.center};

    \node (BSBob2) at (6, 0) {};
    \begin{scope}[rotate=90]
        \BS{BSBob2.center};
    \end{scope}

    \node (MirrorBob1) at (3, 1.5) {};
    \draw (MirrorBob1.center) pic {mirror};

    \node (MirrorBob2) at (6, 1.5) {};
    \draw (MirrorBob2.center) pic[rotate=-90] {mirror};

    \draw[thick,draw=myred] (source.east) -- node[above,text=myred]{${\ket{\alpha}}_S$} (PM.west);
    \draw[thick,draw=myred] (PM.east) -- node[pos=0.48,above,text=myred]{$\ket{(-1)^{U_i} \alpha}_S$} (BSBob1.center) -- node[above]{} (BSBob2.center);
    \draw[thick,draw=myred] (BSBob1.center) -- (MirrorBob1.center) -- node[above]{} (MirrorBob2.center) -- (BSBob2.center);

    \draw[thick,draw=myred] (BSBob2.center) -- ([xshift=+1.0cm] BSBob2.center) pic{detector} node[right,xshift=0.2cm]{``0''};
    \draw[thick,draw=myred] (-1,0) (BSBob2.center) -- ([yshift=-1.0cm] BSBob2.center) pic[rotate=-90]{detector} node[below,yshift=-0.2cm]{``1''};

    \draw[thick,->,>=stealth] ([yshift=0.5cm] PM.north) node[above] {$U_i$} -- (PM.north);
    
    \draw[draw=myblue,dashed] (-5.65,-2) rectangle (-1, 2.25);
    \node[text=myblue,align=left] at (-4.75, -1.25) {Alice's\\lab};
    
    \draw[draw=mygreen,dashed] (-0.5,-2) rectangle (1.75,2.25);
    \node[text=mygreen] at (0, -1.5) {Eve};
    
    \draw[draw=myblue,dashed] (2.25,-2) rectangle (8.25, 2.25);
    \node[text=myblue,align=right] at (7.5, -1.25) {Bob's\\lab};
\end{tikzpicture}
	\caption{\label{fig:dps_setup} The experimental setup of the DPS QKD protocol. In each round, Alice picks a uniformly random bit $U_i$ and uses a phase modulator (PM) to apply a random phase $(-1)^{U_i}$ to a coherent state $\ket{\alpha}$, producing the state $\ket{(-1)^{U_i} \alpha}$ which she sends to Bob. Bob then measures the relative phases between subsequent states using a Mach-Zehnder interferometer. Alice's raw key bit is given by the relative phase $V_i = U_i \oplus U_{i-1}$.}
\end{figure}

Historically, the security of QKD protocols against general attacks (including finite-size effects) was proven using de Finetti type arguments \cite{Renner2007, Renner2008} or the post-selection technique \cite{Christandl2009}. These techniques however require that the protocol of study be permutation invariant. Unfortunately, this is not given for the DPS protocol (permuting the rounds completely changes Bob's raw key bits and does not merely permute them). Thankfully, the EAT does not have this limitation since it applies to any situation where a sequence of channels are applied to some initial state. In fact, a generalised version of the EAT \cite{Metger2022} has recently been used to prove security of QKD protocols \cite{Metger2023}. However, the generalised EAT comes with its own restrictions: In order to apply the generalised EAT we need a well-defined sequence of channels. This then leads us to the following condition:
\begin{cond} 
    \label{cond:sequential}
    Eve does not signal from round $i+1$ to the rounds $1,\ldots,i$.
\end{cond}

A discussion of this condition can be found in \cref{sec:seq_condition}. For the DPS QKD protocol the above condition can be thought of as encompassing both \cref{cond:relativistic} and \cref{cond:sequential_relQKD} of the relativistic protocol. 

\subsection{Protocol}
\label{subsec:DPSprotocol}

\hyperref[protocol:DPS]{Protocol 2} provides a formal description of the steps of the DPS QKD protocol sketched in \cref{fig:dps_setup}. This serves two purposes: firstly, it ensures that the steps of the protocol are clearly laid out. Secondly, it introduces all the registers which will be referenced in the full security proof (see \cref{sec:security_general,sec:security_proof_dps}).

\subsection{Reduction to the relativistic protocol}

\begin{wrapfigure}{l}{0.5\textwidth}
    \begin{tikzpicture}[
    boxnode/.style={rectangle, draw=black, minimum height=0.9cm, minimum width=0.9cm, fill=myblue!30},
    emptynode/.style={}
]
    \draw[very thick,->,>=stealth] (-1.5, -2.0) -- (2.1, -2.0) node[right] {space};
    \draw[very thick,->,>=stealth] (-1.5, -2.0) -- (-1.5, 6.5) node[above] {time};

    \draw[very thick] (-1.5, -2.1) node[below]{Alice} -- (-1.5, -2.0);
    \draw[very thick] (1.5, -2.1) node[below]{Bob} -- (1.5, -1.9);

    \draw[very thick,decorate,decoration={brace,amplitude=0.2cm,raise=0.5cm}] (-1.5, -2.0) -- (-1.5, 2.0) node[pos=0.5,left=0.7cm]{$\mathcal{M}_i$};
    \draw[dashed] (-1.5, 2.0) -- (1.5, 2.0);

    \draw[very thick,decorate,decoration={brace,amplitude=0.2cm,raise=0.5cm}] (-1.5, 2.0) -- (-1.5, 6.0) node[pos=0.5,left=0.7cm]{$\mathcal{M}_{i+1}$};
    \draw[dashed] (-1.5, 6.0) -- (1.5, 6.0);

    \node[rectangle,draw=black,fill=myblue!30] (A1) at (-1.5, -1.5) {$\mathcal{A}$};
    \node[rectangle,draw=black,fill=myblue!30] (A2) at (-1.5, 2.6) {$\mathcal{A}$};

    \node[rectangle,draw=black,fill=myblue!30] (B1) at (1.5, 1.5) {$\mathcal{B}$};
    \node[rectangle,draw=black,fill=myblue!30] (B2) at (1.5, 5.5) {$\mathcal{B}$};

    \node[boxnode,fill=mygreen!30] (E1) at (0, 0) {$\mathcal{E}_i$};
    \node[boxnode,fill=mygreen!30] (E2) at (0, 4.0) {$\mathcal{E}_{i+1}$};

    \draw[thick,->,>=stealth] (A1.north east) -- node[above,sloped]{$S_i$} (E1.south west);
    \draw[thick,->,>=stealth] (E1.north east) -- node[above,sloped]{$S_i$} (B1.south west);

    \draw[thick,->,>=stealth] (A2.north east) -- node[above,sloped]{$S_{i+1}$} (E2.south west);
    \draw[thick,->,>=stealth] (E2.north east) -- node[above,sloped]{$S_{i+1}$} (B2.south west);

    \draw[thick,->,>=stealth] (1.5, -1.7) -- node[right] {$R_{i}$} (B1.south);
    \draw[thick,->,>=stealth] (B1.north) -- node[right] {$R_{i+1}$} (B2.south);
    \draw[thick,->,>=stealth] (B2.north) -- ([yshift=0.7cm] B2.north) node[above]{$S_{i+1}$};

    \draw[thick,->,>=stealth] (B1.east) -- ([xshift=0.4cm] B1.east) node[right]{$B_i$};
    \draw[thick,->,>=stealth] (B2.east) -- ([xshift=0.4cm] B2.east) node[right]{$B_{i+1}$};

    \draw[thick,->,>=stealth,decorate,decoration={snake,amplitude=0.2mm,segment length=2.0mm}]
        ([yshift=-1.8cm] E1.south) node[below]{$E_{i-1}$} -- ([yshift=-0.1cm] E1.south);
    \draw[thick,->,>=stealth,decorate,decoration={snake,amplitude=0.2mm,segment length=2.0mm}]
        ([yshift=0.1cm] E1.north) -- node[pos=0.25,left]{$E_i$} ([yshift=-0.1cm] E2.south);
    \draw[thick,->,>=stealth,decorate,decoration={snake,amplitude=0.2mm,segment length=2.0mm}]
        ([yshift=0.1cm] E2.north) -- ([yshift=1.8cm] E2.north) node[above]{$E_{i+1}$};
    \draw[thick,->,>=stealth] (E2.south east) -- node[sloped, myred, rotate=35]{\large{|}} node[sloped, myred, rotate=-35]{\large{|}} (B1.north west);

\end{tikzpicture}
    \caption{Two rounds of Eve's attack. Alice sends a signal state $S_i$ which gets interrupted by Eve. Eve applies $\mathcal{E}_i$ to this signal state and her previous side-information. Bob measures the signal state together with $S_{i-1}$ to produce his key raw bit $B_i$. Eve's attack cannot signal from $S_{i}$ to $R_i$.}
    \label{fig:dps_seq_channels}
\end{wrapfigure}
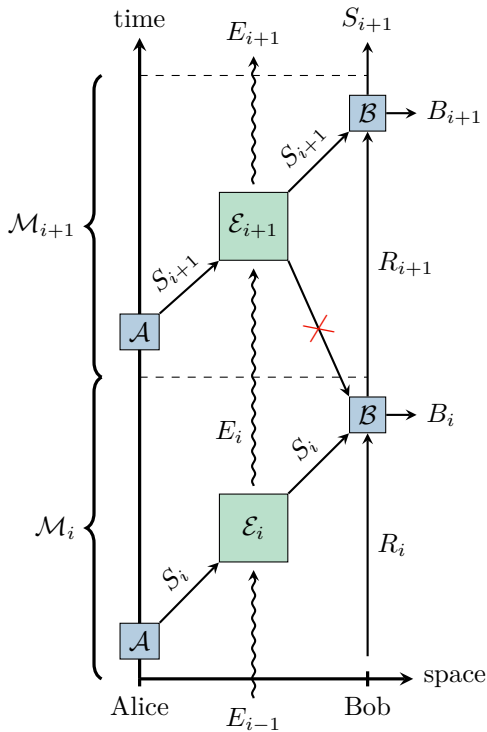

The main idea behind the security proof is to reduce the DPS QKD protocol to the relativistic protocol introduced in \cref{sec:rel_qkd}. As a consequence we can then recycle the results of the relativistic protocol to prove the security of the DPS QKD protocol. For this, we assume that Bob again monitors the measurement times and that Alice and Bob abort if the observed timing suggests that there could be any signalling between neighbouring rounds, i.e., they experimentally enforce \cref{cond:sequential} (see also \cref{sec:seq_condition}).

To see the equivalence between the two protocols we model Eve's attack using a sequence of CPTP maps that take as input Alice's signal states on $S_i$ and some prior (possibly quantum) side-information $E_{i-1}$ and produces a new state on $S_i$ and some new side-information $E_i$ (see \cref{fig:dps_seq_channels}). The collective attack bound is then given by $H(A_i|E_iI^iJ^i\tilde{E}_{i-1})$ (the system $\tilde{E}_{i-1}$ is as defined in \cref{defn:min_tradeoff}). 
The equivalence between the two protocols becomes clear in \cref{fig:dps_seq_channels} on the left: in every round Alice sends a new signal state which gets interrupted by Eve. Due to the sequential condition, Eve cannot hold on to the signal state $S_i$ for too long. In particular, she cannot signal from round $i+1$ back to round $i$. In \cref{fig:dps_seq_channels} this is ensured by the fact that it is impossible to signal backwards in time. For our purposes, however, a simple spacelike separation is sufficient. This non-signalling constraint then allows us to define the channels $\mathcal{M}_i$ for the DPS QKD protocol (for a more detailed description of $\mathcal{M}_i$ see also \cref{sec:security_proof_dps}). Therefore the DPS protocol can be seen as the relativistic protocol where the signal state from round $i$ becomes the reference state in round $i+1$.

%
To formally see the equivalence we apply the same steps in \cref{sec:security_general} to reduce the security analysis to bounding the smooth min-entropy using the generalised EAT. The remaining task then is to evaluate the single-round von Neumann entropy
\begin{equation}
    H(A_i|E_iI^iJ^i\tilde{E}_{i-1}).
\end{equation}
Here we will show that the above quantity can be lower-bounded by the analogous quantity of the relativistic protocol. For this we first separate the situation where Alice keeps her key bit $A_i$ and the situation where she discards it:
\begin{equation}
\begin{aligned}
    H(A_i|E_iI^iJ^i\tilde{E}_{i-1}) =& H(A_i|E_iI^{i-1}J^i\tilde{E}_{i-1},I_i=\bot)\mathrm{Pr}[I_i=\bot] \\
        &+ H(A_i|E_iI^{i-1}J^i\tilde{E}_{i-1},I_i=\top)\mathrm{Pr}[I_i=\top] \\
    =& H(A_i|E_iI^{i-1}J^i\tilde{E}_{i-1},I_i=\top)\mathrm{Pr}[I_i=\top],
\end{aligned}
\end{equation}
where we noted that if $I_i=\bot$ then $A_i=\bot$ is deterministic. Next, we note that by strong subadditivity we have that
\begin{equation}
    H(A_i|E_iI^{i-1}J^i\tilde{E}_{i-1},I_i=\top) \geq H(A_i|E_iI^{i-1}J^i\tilde{E}_{i-1}U_{i-1},I_i=\top).
\end{equation}
Alternatively, one could argue that $\tilde{E}_{i-1}$ can already contain a copy of $U_{i-1}$ and therefore we have equality in the above equation (although the lower bound suffices for our purposes). Since $A_i$ is a deterministic function of $U_i$ and $U_{i-1}$ we may write
\begin{equation}
    H(A_i|E_iI^{i-1}J^i\tilde{E}_{i-1}U_{i-1},I_i=\top) = H(U_i|E_iI^{i-1}J^i\tilde{E}_{i-1}U_{i-1},I_i=\top).
\end{equation}
Finally, we note that $U_i$ is chosen independently from $U_{i-1}I^{i-1}J^{i-1}$ and hence
\begin{equation}
    H(U_i|E_iI^{i-1}J^i\tilde{E}_{i-1}U_{i-1},I_i=\top) = H(U_i|E_iJ_i\tilde{E}_{i-1},I_i=\top).
\end{equation}
Since $U_i$ corresponds to the key register in the relativistic protocol we see that indeed the single-round entropy of the DPS protocol can be evaluated in the same way as for the relativistic protocol (see \cref{sec:security_rel_qkd}). Consequently we also expect the key rates of the DPS protocol to behave almost identically to the relativistic protocol.


Finally we would like to make two comments about our security proof: Firstly, we do not assume that Eve has no phase reference, different to some prior work \cite{Waks2006}. This is justified by the observation that Eve could always sacrifice a small fraction of rounds at the start of the protocol to learn the phase of Alice's laser to arbitrary precision. The second comment is that there are some subtleties when applying the generalised EAT regarding the assignment of the memory system (the upper arm in Bob's Mach-Zehnder interferometer) as well as the construction of the conditioning registers $C_i$. For a more detailed discussion of these issues we refer to \cref{sec:security_proof_dps}.


\vspace{-5pt}
\begin{figure}[H]
	\captionsetup{name=Box}
	\begin{mdframed}[style=MyFrame]
		\small
		{\large \textsf{Protocol 2: Differential phase shift QKD}}\label{protocol:DPS}\\[0.2cm]
		The protocol is defined in terms of the following parameters, which are chosen before the protocol begins: \\[0.1cm]
		\begin{tabular}{rl}
			$\alpha \in \mathbb{C}$: & amplitude of the laser light \\
			$n\in\mathbb{N}$: & number of protocol rounds \\
			$\gamma \in \mathbb{R}$: & testing frequency \\
			$\mathrm{leak_{EC}}$: & maximum length of error correction \\
			$\varepsilon_\mathrm{EC}$: & error tolerance during error correction \\
			$f: \mathbb{P}_\mathcal{C} \rightarrow \mathbb{R}$: & a valid min-tradeoff function \\
			$H_\mathrm{exp} \in \mathbb{R}$: & minimum expected single-round entropy \\
			$l \in \mathbb{N}$: & length of the final secret key \\
		\end{tabular}\\[0.5cm]
		\hrule\vspace{10pt}
		\begin{enumerate}
			\item \textbf{Initialization:} Alice chooses a bit $U_0 \in \{0,1\}$ uniformly at random and sends the state $\ket{(-1)^{U_0}\alpha}_S$ to Bob.
			\item \textbf{Measurement:} For $i\in [n]$:
			\begin{enumerate}
				\item Alice chooses a bit $U_i \in \{0, 1\}$ uniformly at random.
				\item Alice prepares the state $\ket{(-1)^{U_i} \alpha}_S$ and sends it to Bob.
				\item Alice computes her raw key bit $V_i = U_{i} \oplus U_{i-1}$.
				\item Bob receives a state $\rho_{S}$ and sends it through a Mach-Zehnder interferometer (see \cref{fig:dps_setup}).
				\item Bob applies the POVM $\{M_{SR}^{(b)}\}_b$ to the output of the interferometer and records the outcome in the register $B_i \in \{0, 1, \bot\}$.
				\item If $B_i = \bot$ then Bob sets $I_i = \bot$ and $I_i = \top$ otherwise.
				\item Bob chooses $T_i \in \{0,1\}$ randomly with $\mathrm{Pr}[T_i=1] = \gamma$. If $T_i=1$ then Bob records $J_i = B_i$ and $J_i = \varnothing$ otherwise.
                \item Alice waits to enforce \cref{cond:sequential}.
			\end{enumerate}
            \item \textbf{Public announcement:} Bob announces $I^nJ^n$.
			\item \textbf{Sifting:} For all $i \in [n]$ Alice sets $A_i = V_i$ if $I_i = \top$ and $A_i = \bot$ otherwise.
			\item \textbf{Error correction:}
			\begin{enumerate}
				\item Alice and Bob use their outputs $A^{n}$ and $B^{n}$ to perform error correction by communicating at most $\mathrm{leak_{EC}}$ number of bits. Bob stores his guess for Alice's key in $\tilde{A}^{n}$.
				\item Alice chooses a hash function $h \in \mathcal{F}$ uniformly at random from a family of two-universal hash functions of length $\lceil \log(1/\varepsilon_\mathrm{EC}) \rceil$ and applies it to her raw key. She sends the output $h(A^{n})$ and her choice of hash function to Bob.
				\item Bob applies the same hash function to his guess $\tilde{A}^{n}$. If the two hashes disagree, Alice and Bob abort the protocol.
			\end{enumerate}
			\item \textbf{Parameter estimation:} For all $i \in [n]$ Alice computes $C_i=\mathrm{EV}(A_i, J_i)$. If $f({\mathrm{freq}(C^n)}) < H_\mathrm{exp}$ they abort the protocol.
			\item \textbf{Privacy amplification:} Alice and Bob perform privacy amplification on $A^{n}$ and $\tilde{A}^{n}$ to obtain raw keys $K_A^l$ and $K_B^l$.
		\end{enumerate}
	\end{mdframed}
\end{figure}

\subsection{Results}
We now present the results of the security analysis of the DPS QKD protocol. We limit ourselves to loss as the only source of noise in our protocol. Furthermore, we choose a soundness parameter of $\varepsilon^{\mathrm{snd}}=4 \cdot 10^{-12}$ and a completeness parameter of $\varepsilon^{\mathrm{comp}}=10^{-2}$. The amplitude $\alpha$ of the laser light is chosen such that it optimizes the asymptotic key rates. Typical values are $\alpha \approx 0.45$. Again, we emphasize that there are many places in the finite-size analysis where the bound on the key rate could be tightened. The finite-size plots in \cref{fig:key_rates_dps} should therefore be viewed only for illustrational purposes. 

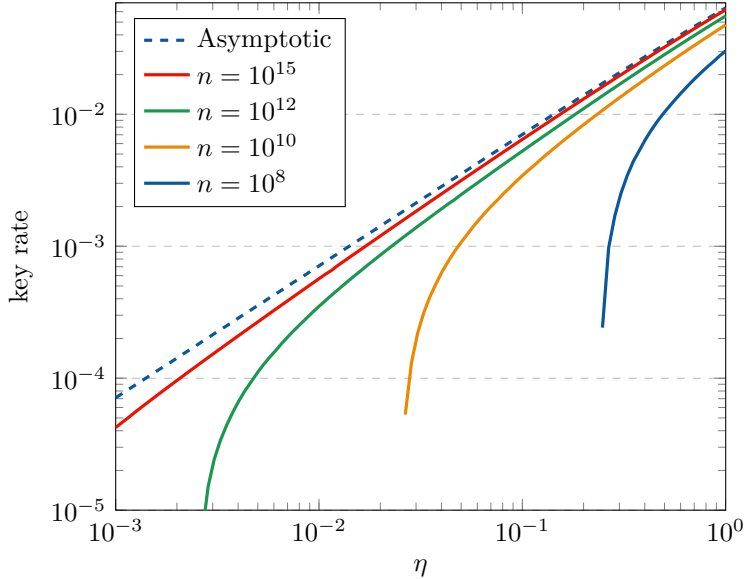
\begin{figure}[t]
    \centering
    \begin{tikzpicture}
        \begin{axis}[
            xmode=log,
            ymode=log,
            width = 0.7\linewidth,
            xlabel={$\eta$},
            ylabel={key rate},
            xmin=1e-3, xmax=1,
            ymin=1e-5, ymax=0.07,
            ymajorgrids=true,
            grid style=dashed,
            legend pos=north west,
            every mark=none,
            cycle list name=mycolors,
            legend cell align={left},
            scale=0.9,
        ]
            \addplot+[dashed,very thick,mark=none] table[x=eta, y=rate, col sep=comma]{data/dps/opt_alphas.csv};
            \addlegendentry{Asymptotic}

            \addplot+[very thick,mark=none] table[x=eta, y=min_ent, col sep=comma]{data/dps/opt_grads_n=1e15.csv};
            \addlegendentry{$n=10^{15}$}

            \addplot+[very thick,mark=none] table[x=eta, y=min_ent, col sep=comma]{data/dps/opt_grads_n=1e12.csv};
            \addlegendentry{$n=10^{12}$}

            \addplot+[very thick,mark=none] table[x=eta, y=min_ent, col sep=comma]{data/dps/opt_grads_n=1e10.csv};
            \addlegendentry{$n=10^{10}$}

            \addplot+[very thick,mark=none] table[x=eta, y=min_ent, col sep=comma]{data/dps/opt_grads_n=1e8.csv};
            \addlegendentry{$n=10^{8}$}
        \end{axis}
    \end{tikzpicture}
    \caption{\label{fig:key_rates_dps} Key rates of the DPS QKD protocol as a function of the transmittance $\eta\in[0,1]$ for different numbers of rounds $n$. Loss is the only type of noise that is considered here.}
\end{figure}

Similarly to the relativistic protocol we observe a linear scaling of the key rate for the entire parameter range of efficiencies. When compared with the relativistic protocol (\cref{sec:rel_qkd}), we observe a modest increase in the asymptotic key rate (for a fair comparison we need to half the key rates of the relativistic protocol since it uses two light pulses per key bit). Other than that, the protocol behaves in the same way as the relativistic protocol. This is expected since, after all, the security proof exploits the equivalence of the two protocols.

\section{Discussion}
\label{sec:Discussion}

There are a some aspects of the security proofs of the two protocols that deserve special attention. First, we would like to discuss how to satisfy \cref{cond:relativistic,cond:sequential_relQKD} for the relativistic QKD protocol. We will also discuss the impact and enforceability of \cref{cond:sequential} for the DPS QKD protocol. 
In the second part of the discussion, we relate our work to a known attack on DPS QKD. In particular, we will make good on the promise made in the introduction by showing that for DPS coherent attacks are indeed stronger than collective attacks.

\subsection{Sequential conditions}
\label{sec:seq_condition}
Here we discuss some implications of \cref{cond:relativistic,cond:sequential_relQKD,cond:sequential}. First, we note that, in practice, this condition can be imposed by Alice and Bob if Bob monitors the arrival time of the quantum systems (or equivalently his detection times). For there to be no signalling between the signal and the reference we require that the arrival of the reference in Bob's lab is outside the future light cone of Alice sending the signal state (grey dotted line in \cref{fig:rel_QKD_spacetime}). From \cref{fig:rel_QKD_spacetime} it follows that for the two signals to be spacelike separated we require that
\begin{equation}
    t_A^{(i)} + \Delta t + \frac{d}{c} > t_B^{(i)} - \Delta t \iff t_B^{(i)} - t_A^{(i)} < 2 \Delta t + \frac{d}{c}.
    \label{eq:spacelike_separation}
\end{equation}
If we assume that Alice and Bob are connected by a fibre of refractive index $n$ and length $d$, we require that $t_B^{(i)} \geq t_A^{(i)} + n d/c + \Delta t$. Inserting this into \cref{eq:spacelike_separation} provides a lower bound on the time delay $\Delta t$ that is required for the protocol to not abort:
\begin{equation}
	\Delta t > (n-1)\frac{d}{c}.
    \label{eq:min_time_delay}
\end{equation}
One way to think about these conditions is that \cref{eq:spacelike_separation} is required for the soundness of the protocol, whereas \cref{eq:min_time_delay} is required for completeness (note that \cref{eq:spacelike_separation} does not make any assumptions about the refractive index of the fibre). To enforce \cref{cond:sequential_relQKD} we choose $t_A^{(i+1)} = t_A^{(i)} + 2\Delta t$. Combining this with \cref{eq:spacelike_separation} we get that $t_A^{(i+1)} + d/c = t_A^{(i)} + 2\Delta t + d/c > t_B^{(i)}$, which says that the reference from round $i+1$ cannot signal to round $i$.

Enforcing these conditions requires Alice and Bob to share a pair of synchronized clocks. This is a reasonable request, as synchronized clocks are already needed in the DPS protocol so that Bob knows which of his detections corresponds to which of Alice's key bits. Enforcing the sequential condition, however, imposes a minimal time delay between signals. This has two consequences: firstly, it limits the repetition rate of protocol rounds. Secondly, this might introduce additional noise (due to the longer arm in Bob's interferometer) and requires better phase coherence of Alice's laser. The first problem can be fixed if Bob measures the relative phase between more temporally distant pulses instead of performing interferometry on neighbouring pulses. Effectively, this corresponds to running many copies of the DPS QKD protocol in parallel. Note that the impact of the sequential condition on the repetition rate varies significantly depending on the transmission channel. For free-space implementations, for instance, the sequential condition can be enforced without significant loss in repetition rate. Lastly, we note that this non-signalling assumption is also implicitly made when considering many restricted sets of attacks such as individual attacks \cite{Waks2006} or collective attacks (for which the security of DPS QKD has not been established before this paper). Therefore, the security statement presented in this paper is stronger than that of prior work.

\subsection{Comparison with upper bounds}

\begin{figure}[t]
	\centering
	\begin{tikzpicture}
		\begin{axis}[
			xmode=log,
			width = 0.8\linewidth,
			xlabel={$\eta$},
			ylabel={QBER},
			xmin=1e-5, xmax=1e-1,
			ymin=1e-4, ymax=0.15,
			ymajorgrids=true,
			grid style=dashed,
			legend pos=north west,
			every mark=none,
			cycle list name=mycolors,
			legend cell align={left},
            scale=0.9,
			]
			\path[name path=top] (current axis.north west) -- (current axis.north east);
			\path[name path=xaxis] (current axis.south west) -- (current axis.south east);
			
			\addplot+[forget plot, solid, myred, very thick, mark=none, name path=curty] table[x=eta, y=qber, col sep=comma]{data/dps/bound_curty.csv};
			
			\addplot+[forget plot, solid, myblue, very thick, mark=*, name path=ns] table[x=eta, y=qber, col sep=comma]{data/dps/threshold_qbers.csv};
			
			\addplot[myred!40, opacity=0.4] fill between[of=curty and top];
			\addlegendentry{Insecure according to \cite{Curty2008}}
			
			\addplot[myblue!40, opacity=0.4] fill between[of=ns and xaxis];
			\addlegendentry{Secure with non-signalling constraint}
		\end{axis}
	\end{tikzpicture}
	\caption{\label{fig:upper_bounds} Comparison between noise thresholds derived in \cite{Curty2008} and the ones derived in this paper. The red region is insecure according to \cite{Curty2008}, whereas the blue region is secure according to our security proof. There is a non-empty overlap of the two regions. Both curves were computed at $\alpha=0.4$ with zero detector dead time.}
\end{figure}
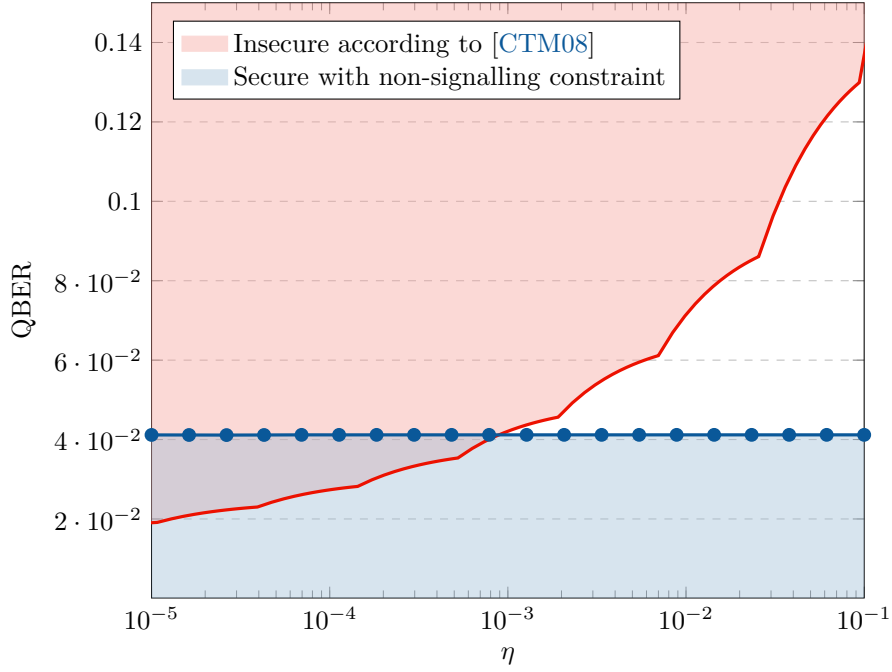

Next, we discuss the relation of our work to the upper bounds on DPS derived in \cite{Curty2008}. In this work, the authors discovered an attack where Eve performs an intercept-resend attack but only resends the pulses if she gets a sufficient number of consecutive conclusive outcomes. This allows Eve to exploit losses to reduce the detectable effects (i.e., the QBER) of her attack. This attack constitutes an entanglement-breaking channel, and as a result, the authors found a parameter regime in which DPS QKD is insecure. Note, however, that this attack violates \cref{cond:sequential}; hence, we do not necessarily expect that this upper bound holds for our security claim. In fact, using our methods, we can derive a parameter regime for which the DPS QKD protocol is secure, as shown in \cref{fig:upper_bounds}. We observe that there exists a region where the two regimes overlap, i.e., the security claims disagree. This shows that the attack in \cite{Curty2008} is stronger than any collective attack since those are covered by our security proof. Furthermore, the converse also holds: any attempt to reduce the security of DPS QKD to collective attacks necessarily requires additional assumptions (since in such a reduction, you need to exclude the type of attack reported in \cite{Curty2008}).

\section{Conclusion and outlook}
\label{sec:outlook}

In this work, we proved the security of DPS QKD against general attacks by exploiting relativistic constraints. In particular, we use relativity to enforce a non-signalling constraint on the eavesdropper, i.e., the eavesdropper can only signal from previous to future rounds but not in the other direction. This strategy then allows us to reduce the DPS QKD protocol to a relativistic protocol. We then applied methods from quantum information theory, relativity, and quantum optics to prove the security of this relativistic protocol. The particular methods of interest are the generalised entropy accumulation theorem (discussed in \cref{sec:qkd_security}), a formal way of treating non-signalling (discussed in \cref{sec:causality}), and the squashing technique (discussed in \cref{sec:quantum_optics}). We observed linear scaling of the secret key rate as a function of the detection efficiency, which is the best we can hope for \cite{Takeoka2014,Pirandola2017}. This observation and the practicality of implementing DPS QKD make the protocol attractive for real-world implementations. Naturally, one may ask whether it would be possible to prove the security of the DPS QKD protocol against general attacks without the non-signalling constraint. By comparing our results with upper-bounds for DPS QKD derived in \cite{Curty2008}, we observed that our security statement can violate these bounds (which do not satisfy the non-signalling condition). Since any collective attack fulfils the non-signalling property, it follows that it is impossible to reduce the security analysis of the DPS QKD protocol to collective attacks without additional assumptions. Since many proof techniques proceed by reducing security against general attacks to the security against collective attacks, they cannot be applied to DPS QKD without such additional assumptions (such as the non-signalling assumption). Furthermore, even if one were to prove the security of DPS without the non-signalling assumption, this would incur a loss of secret key rate and hence could limit the practicality of the protocol.

In addition, the insight that coherent attacks are stronger than collective attacks for DPQ QKD sheds light on the limitations of current security proof techniques. State-of-the-art techniques such as the (generalised) EAT and the quantum de Finetti theorem rely on demonstrating that general attacks are not stronger than collective attacks, thereby simplifying the task to proving security against collective attacks. These techniques have proven effective, as only a handful of protocols are known to exhibit instances where coherent attacks surpass collective attacks, primarily in device-independent protocols where the inner workings of quantum devices remain uncharacterised \cite{Thinh2016,Sandfuchs2023}. Our results thus provide the first example of a device-\emph{dependent} QKD protocol where general attacks are stronger than collective attacks.  This outcome is consistent with the nature of the DPS QKD protocol; for protocols demonstrating an iid structure, it is unsurprising that adversaries cannot gain an advantage by introducing correlations between rounds. This is different for protocols where individual rounds are not independent of each other; here, general attacks can exploit these correlations, something not attainable by collective attacks. Consequently, it is reasonable to speculate that, in general, protocols lacking an iid structure may frequently witness general attacks prevailing over collective ones, thereby rendering standard security proof techniques unsuitable for direct application. As such, the methods presented in our work provide a way to identify under which conditions state-of-the-art techniques still allow to prove security against general attacks for such protocols.

Our work opens up several directions for future research.
A natural question to ask is whether similar techniques could be applied to other distributed phase reference protocols such as the coherent one-way protocol \cite{Stucki2005}. Here, we note that the trick of exploiting the non-signalling condition to cast the protocol into a relativistic protocol works quite generally. The main challenge when trying to apply the techniques to other protocols lies in finding a squashing map which is itself non-signalling. In general, this is a difficult problem which we leave for future work.

The security analysis presented in this paper could be improved in many places. Firstly, we assume here that both detectors have equal efficiency, which is required to be able to push all losses from the detectors into the channel. This allows us to consider ideal detectors when applying the squashing technique. Recently, the scenario with unequal detection efficiencies has been studied in \cite{Zhang2021}. However, it is not obvious how their technique could be applied to our protocols, since we require the ``squashed'' protocol to retain the relativistic constraint on Eve's attack. Similar problems arise when trying to apply different dimension reduction techniques such as the one presented in \cite{Upadhyaya2021} to our protocols. 
Furthermore, the same caveats of typical QKD security analyses apply to this paper: If the implementation does not match with the theoretical model of the devices, the protocol could become insecure. One example of such a limitation are imperfections in the source, which we do not consider here. Lastly, there are some places where the security analysis could be tightened. To carry out the security proof of DPS QKD, we need to give Eve more power than she  has in practice (see \cref{sec:security_proof_dps}), which is required to put the protocol into a form that allows for the application of the generalised EAT. This requirement, however, seems rather artificial; hence, one can hope that it is possible to avoid it. Unfortunately, it seems that reducing Eve's area of influence would also prohibit one from finding (an obvious) squashing map for the DPS QKD protocol.

Finally, there are several other directions that can be considered for the relativistic QKD protocol. Compared to other protocols, the relativistic QKD protocol has quite a low threshold QBER. To remedy this, one could try to allow Alice to choose different bases to encodce her key bit. This is motivated by the fact that the BB84 and six-state protocols can tolerate a higher QBER than the B92 protocol and so we might hope to observe similar effects here. Similarly, there could be opportunities for improved key rates by considering a phase-randomized version of the relativistic protocol. This is motivated by the observation that some other protocols benefit from phase randomization \cite{Lo2006}.

\section*{Code availability}

The code and data necessary to reproduce the results of this paper are available at\\
\href{https://gitlab.phys.ethz.ch/martisan/dps-key-rates}{https://gitlab.phys.ethz.ch/martisan/dps-key-rates}.

\section*{Acknowledgements}

We thank Tony Metger, Renato Renner, and Ernest Tan for helpful comments and discussions. This work was supported by the Air Force Office of Scientific Research (AFOSR), grant No.~FA9550-19-1-0202, the QuantERA project eDICT, the National Centre of Competence in Research SwissMAP, and the Quantum Center at ETH Zurich. VV is supported by an ETH Postdoctoral Fellowship and acknowledges financial support from the Swiss National Science Foundation (SNSF) Grant Number 200021\_188541.

\bibliographystyle{halpha}
\bibliography{RelativisticQKDbib}

\newpage
\appendix

\appendixpage
\startcontents[sections]
{\hypersetup{linkcolor=black}\printcontents[sections]{l}{1}{\setcounter{tocdepth}{2}}}

\section{Technical definitions}
\label{sec:tech_defs}
\begin{defn}
	\label{def:states}
    Let $\rho_{XA} \in \mathcal{S}(XA)$ be a classical-quantum state, i.e., $\rho_{XA}$ can be written in the form
    \begin{align}
        \rho_{XA} = \sum_{x \in \mathcal{X}} p(x) \ketbra{x}{x}_X \otimes \rho_A^{[x]},
    \end{align}
    where $\mathcal{X}$ is some alphabet, $p(x)$ is a probability distribution over $\mathcal{X}$, and $\rho_A^{[x]} \in \mathcal{S}(A)$. For an event $\Omega \subseteq \mathcal{X}$ we define the following states: The subnormalised state $(\rho_{\land\Omega})_{XA}$ conditioned on $\Omega$,
    \begin{align}
        (\rho_{\land\Omega})_{XA} &= \sum_{x \in \Omega} p(x) \ketbra{x}{x}_X \otimes \rho_A^{[x]},
    \end{align}
	and the normalised state $(\rho_{|\Omega})_{XA}$ conditioned on $\Omega$,
	\begin{equation}
		 (\rho_{|\Omega})_{XA} = \frac{(\rho_{\land\Omega})_{XA}}{\rho[\Omega]},\quad \mathrm{where}\;\rho[\Omega] = \tr \left[ (\rho_{\land\Omega})_{XA} \right] = \sum_{x \in \Omega} p(x).
	\end{equation}
\end{defn}

\begin{defn} (von Neumann entropy)
    Let $A$ and $B$ be two quantum systems and $\rho_{A} \in \mathcal{S}(A)$ be a state. The \emph{entropy} of $\rho_{A}$ is defined as
    \begin{align}
        H(A)_\rho = -\tr \left[ \rho_A \log \rho_A \right].
    \end{align}
    For a state $\rho_{AB} \in \mathcal{S}(AB)$ we define the \emph{conditional entropy} as
    \begin{align}
        H(A|B)_\rho = H(AB)_\rho - H(B)_\rho,
    \end{align}
    where $H(B)_\rho$ is the entropy evaluated on $\rho_B = \tr_A \rho_{AB}$.
\end{defn}

\begin{defn} (Generalised fidelity)
    Let $\rho_A, \sigma_A \in \mathcal{S}_\leq(A)$ be two subnormalised states. Define the \emph{fidelity} by
    \begin{align*}
        F(\rho_A, \sigma_A) = \left( \tr | \sqrt{\rho_A}\sqrt{\sigma_A} | + \sqrt{(1-\tr \rho_A)(1 - \tr \sigma_A)} \right)^2.
    \end{align*}
\end{defn}
\begin{defn} \label{def:purified_distance} (Purified distance)
    Let $\rho_A, \sigma_A \in \mathcal{S}_\leq(A)$, define the \emph{purified distance} as
    \begin{align*}
        P(\rho_A, \sigma_A) = \sqrt{1 - F(\rho_A, \sigma_A)}.
    \end{align*}
\end{defn}

\begin{defn} ($\varepsilon$-ball)
    Let $\varepsilon > 0$ and $\rho_A \in \mathcal{S}_\leq(A)$, we define the \emph{$\varepsilon$-ball around $\rho_A$} as
    \begin{equation}
        \mathcal{B}^\varepsilon(\rho_A) = \{ \sigma_A \in \mathcal{S}_\leq(A) \; | \; P(\rho_A, \sigma_A) < \varepsilon \},
    \end{equation}
    where $P(\rho_A, \sigma_A)$ denotes the purified distance (\cref{def:purified_distance}).
\end{defn}

\begin{defn} (Smooth min and max-entropies)
    Let $\varepsilon > 0$ and $\rho_{AB} \in \mathcal{S}_\leq(AB)$ be a quantum state, then the \emph{smooth min-entropy} is defined as
    \begin{align}
        H_\mathrm{min}^\varepsilon (A|B)_\rho = -\log \inf_{\tilde{\rho}_{AB}} \inf_{\sigma_B} \left\| \tilde{\rho}^{1/2}_{AB} \sigma_B^{-1/2} \right\|^2_\infty,
    \end{align}
    where the optimizations are over all $\tilde{\rho}_{AB} \in B^\varepsilon(\rho_{AB})$ and $\sigma_B \in \mathcal{S}(B)$. Similarly we define the \emph{smooth max-entropy} as
    \begin{align}
        H_\mathrm{max}^\varepsilon (A|B)_\rho = \log \inf_{\tilde{\rho}_{AB}} \sup_{\sigma_B} \left\| \tilde{\rho}^{1/2}_{AB} \sigma_B^{1/2} \right\|^2_1,
    \end{align}
    where $\tilde{\rho}_{AB}$ and $\sigma_B$ are as before.
\end{defn}

\section{Proof of Lemma~\ref{lem:choi}}
\label{sec:proof_choi}

To establish the equivalence between \cref{defn:nonsignalling} of signalling and the condition on the Choi state given in \cref{eq: choi_condition}, the intermediate condition given by \cref{eq: nonsignalling2} will be useful. As this condition is shown to be equivalent to \cref{defn:nonsignalling} in \cite{Ormrod2023}, establishing the equivalence between \cref{eq: nonsignalling2} and \cref{eq: choi_condition} would conclude the proof. We repeat \cref{eq: nonsignalling2} below for convenience, it states that $S$ does not signal to $R'$ in the CPTP map $\mathcal{E}_{SR\rightarrow S'R'}$ if there exists a CPTP map  $\mathcal{E}_{R\rightarrow R'}$ such that
\begin{equation}
\label{eq: choi_proof1}
 \tr_{S'} \circ \mathcal{E}_{SR\rightarrow S'R'} = \tr_S \otimes \mathcal{E}_{R\rightarrow R'}.
\end{equation}

 We use the above equation to establish the necessary direction, a diagrammatic version of the proof of this part is given in \cref{fig:choi_proof_necc}. 
 \begin{align}
 \label{eq:choi_proof_nec}
     \begin{split}
       \tr_{S'} [ \mathcal{C}(\mathcal{E}_{SR\rightarrow S'R'})] &= \tr_{S'}[(\mathcal{I}_{\bar{S}\bar{R}}\otimes \mathcal{E}_{SR\rightarrow S'R'}) \ket{\Phi}\bra{\Phi}_{\bar{S}\bar{R}SR}]\\
       &= (\mathcal{I}_{\bar{S}\bar{R}}\otimes \tr_{S'} \circ \mathcal{E}_{SR\rightarrow S'R'}) \ket{\Phi}\bra{\Phi}_{\bar{S}\bar{R}SR}\\
       &= (\mathcal{I}_{\bar{S}\bar{R}}\otimes \tr_S \otimes \mathcal{E}_{R\rightarrow R'}) \ket{\Phi}\bra{\Phi}_{\bar{S}\bar{R}SR}\\
         &= \tr_S[\ket{\Phi}\bra{\Phi}_{\bar{S}S}] \otimes (\mathcal{I}_{\bar{R}}\otimes \mathcal{E}_{R\rightarrow R'})\ket{\Phi}\bra{\Phi}_{\bar{R}R}\\
         &=\frac{\mathds{1}_{\bar{S}}}{d_{\bar{S}}}\otimes (\mathcal{I}_{\bar{R}}\otimes \mathcal{E}_{R\rightarrow R'})\ket{\Phi}\bra{\Phi}_{\bar{R}R}\\
       &= \frac{\mathds{1}_{\bar{S}}}{d_{\bar{S}}}\otimes \tr_{\bar{S}S'} [ \mathcal{C}(\mathcal{E}_{SR\rightarrow S'R'})].
     \end{split}
 \end{align}

In going from the third to the fourth step in the above, we have used the fact that $\ket{\Phi}_{\bar{S}\bar{R}SR}=\sum_{i,j}\ket{ijij}_{\bar{S}\bar{R}S'R'}\equiv \sum_i\ket{ii}_{\bar{S}S}\otimes \sum_j\ket{jj}_{\bar{R}R}=\ket{\Phi}_{\bar{S}S}\otimes \ket{\Phi}_{\bar{R}R}$.

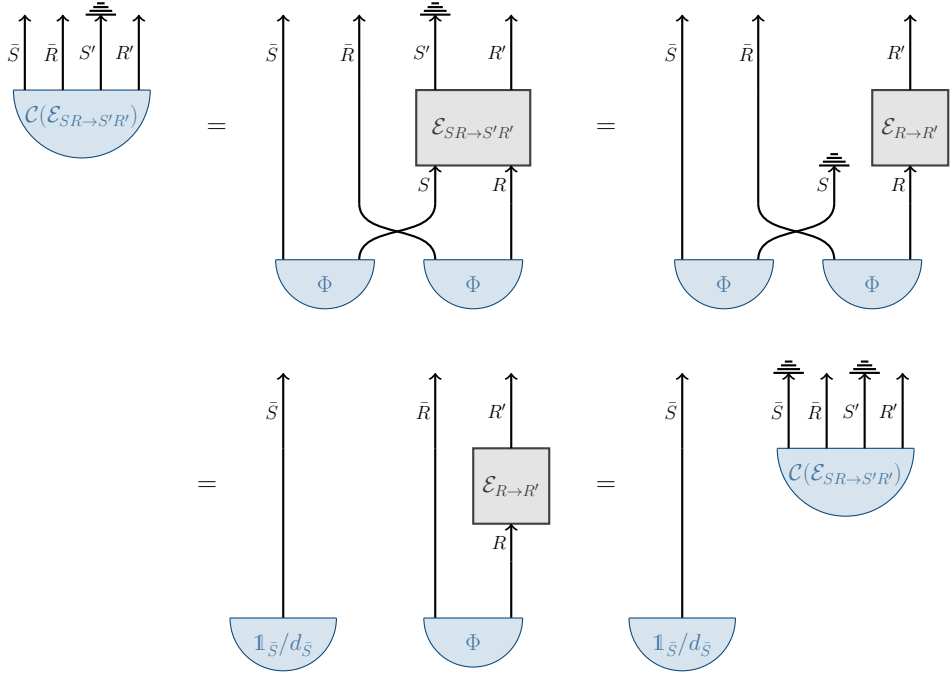
\begin{figure}[h]
    \centering
    \begin{tikzpicture}[trace/.pic={\draw [thick](-0.4,0)--(0.4,0);\draw [thick](-0.3,0.1)--(0.3,0.1);\draw [thick](-0.2,0.2)--(0.2,0.2);\draw [thick](-0.1,0.3)--(0.1,0.3);},phiprep/.pic={\node[myblue!80!black] at (1.3,-0.65) {\LARGE{$\Phi$}};\draw[draw=myblue!80!black,fill=myblue!40!white, fill opacity=0.4] (0,0) arc (180:360:1.3);
\draw[myblue!80!black] (0,0)--(2.6,0); },scale=0.5, transform shape]

\draw[thick, fill=LightGray, opacity=0.75 ] (0,0) rectangle node[align=center]{\LARGE{$\mathcal{E}_{SR\rightarrow S'R'}$}} (3,2);
  
  \draw[thick,black,->] (0.5,-1)--node[anchor=east]{\Large$S$}(0.5,0);   \draw[thick,black,->] (2.5,-1)--node[anchor=east]{\Large$R$}(2.5,0); 
  
   \draw[thick,black,->] (0.5,2)--node[anchor=east]{\Large$S'$}(0.5,4); \draw[thick,black,->] (2.5,2)--node[anchor=east]{\Large$R'$}(2.5,4); 
     \draw[thick] (-1.5,-2.5) to[out=90,in=270] (0.5,-1);
     \draw[thick] (0.5,-2.5) to[out=90,in=270] (-1.5,-1);
\draw[thick] (-1.5,-1)--(-1.5,2);\draw[thick,->] (-1.5,2) --node[anchor=east]{\Large$\bar{R}$}(-1.5,4);
 \draw[thick] (-3.5,-2.5) --(-3.5,2); \draw[thick,->] (-3.5,2) --node[anchor=east]{\Large$\bar{S}$}(-3.5,4);
     \draw[thick] (2.5,-2.5) to[out=90,in=270] (2.5,-1);

\draw (0.2,-2.5) pic {phiprep}; \draw (-3.7,-2.5) pic {phiprep};
\draw (0.5,4) pic {trace};

\node[myblue!80!black] at (-8.8,1.3) {\LARGE{$\mathcal{C}(\mathcal{E}_{SR\rightarrow S'R'})$}};
\draw[draw=myblue!80!black,fill=myblue!40!white, fill opacity=0.4] (-10.6,2) arc (180:360:1.8);
\draw[myblue!80!black] (-10.6,2)--(-7,2);

\draw[thick,->] (-7.3,2)--node[anchor=east]{\Large$R'$}(-7.3,4); \draw[thick,->] (-8.3,2)--node[anchor=east]{\Large$S'$}(-8.3,4); \draw[thick,->] (-9.3,2)--node[anchor=east]{\Large$\bar{R}$}(-9.3,4); \draw[thick,->] (-10.3,2)--node[anchor=east]{\Large$\bar{S}$}(-10.3,4); \draw (-8.3,4) pic {trace};

\node at (-5.25,1) {\huge{$=$}};

\begin{scope}[shift={(10.5,0)}]
\node at (-5.5,1) {\huge{$=$}};

    \draw[thick, fill=LightGray, opacity=0.75 ] (1.5,0) rectangle node[align=center]{\LARGE{$\mathcal{E}_{R\rightarrow R'}$}} (3.5,2);
  
  \draw[thick,black,->] (0.5,-1)--node[anchor=east]{\Large$S$}(0.5,0);   \draw[thick,black,->] (2.5,-1)--node[anchor=east]{\Large$R$}(2.5,0); 
  
   \draw[thick,black,->] (2.5,2)--node[anchor=east]{\Large$R'$}(2.5,4); 
     \draw[thick] (-1.5,-2.5) to[out=90,in=270] (0.5,-1);
     \draw[thick] (0.5,-2.5) to[out=90,in=270] (-1.5,-1);
\draw[thick] (-1.5,-1)--(-1.5,2);\draw[thick,->] (-1.5,2) --node[anchor=east]{\Large$\bar{R}$}(-1.5,4);
 \draw[thick] (-3.5,-2.5) --(-3.5,2); \draw[thick,->] (-3.5,2) --node[anchor=east]{\Large$\bar{S}$}(-3.5,4);
     \draw[thick] (2.5,-2.5) to[out=90,in=270] (2.5,-1);

\draw (0.2,-2.5) pic {phiprep}; \draw (-3.7,-2.5) pic {phiprep};
\draw (0.5,0) pic {trace}; 

\end{scope}

\begin{scope}[shift={(0,-9.5)}]
\node at (-5.5,1) {\huge{$=$}};

    \draw[thick, fill=LightGray, opacity=0.75 ] (1.5,0) rectangle node[align=center]{\LARGE{$\mathcal{E}_{R\rightarrow R'}$}} (3.5,2);
  
    \draw[thick,black,->] (2.5,-1)--node[anchor=east]{\Large$R$}(2.5,0); 
  
   \draw[thick,black,->] (2.5,2)--node[anchor=east]{\Large$R'$}(2.5,4); 
   
\draw[thick] (0.5,-2.5)--(0.5,2);\draw[thick,->] (0.5,2) --node[anchor=east]{\Large$\bar{R}$}(0.5,4);
 \draw[thick] (-3.5,-2.5) --(-3.5,2); \draw[thick,->] (-3.5,2) --node[anchor=east]{\Large$\bar{S}$}(-3.5,4);
     \draw[thick] (2.5,-2.5) to[out=90,in=270] (2.5,-1);

\draw (0.2,-2.5) pic {phiprep}; 
 
\node[myblue!80!black] at (-3.5,-3.2) {\LARGE{$\mathds{1}_{\bar{S}}/d_{\bar{S}}$}};
\draw[draw=myblue!80!black,fill=myblue!40!white, fill opacity=0.4] (-4.9,-2.5) arc (180:360:1.4);
\draw[myblue!80!black] (-4.9,-2.5)--(-2.1,-2.5);

\end{scope}

\begin{scope}[shift={(10.5,-9.5)}]
\node at (-5.5,1) {\huge{$=$}};

 \draw[thick] (-3.5,-2.5) --(-3.5,2); \draw[thick,->] (-3.5,2) --node[anchor=east]{\Large$\bar{S}$}(-3.5,4);

\node[myblue!80!black] at (-3.5,-3.2) {\LARGE{$\mathds{1}_{\bar{S}}/d_{\bar{S}}$}};
\draw[draw=myblue!80!black,fill=myblue!40!white, fill opacity=0.4] (-4.9,-2.5) arc (180:360:1.4);
\draw[myblue!80!black] (-4.9,-2.5)--(-2.1,-2.5);

\node[myblue!80!black] at (0.8,1.3) {\LARGE{$\mathcal{C}(\mathcal{E}_{SR\rightarrow S'R'})$}};
\draw[draw=myblue!80!black,fill=myblue!40!white, fill opacity=0.4] (-1,2) arc (180:360:1.8);
\draw[myblue!80!black] (-1,2)--(2.6,2);

\draw[thick,->] (2.3,2)--node[anchor=east]{\Large$R'$}(2.3,4); \draw[thick,->] (1.3,2)--node[anchor=east]{\Large$S'$}(1.3,4); \draw[thick,->] (0.3,2)--node[anchor=east]{\Large$\bar{R}$}(0.3,4); \draw[thick,->] (-0.7,2)--node[anchor=east]{\Large$\bar{S}$}(-0.7,4); \draw (1.3,4) pic {trace}; \draw (-0.7,4) pic {trace};

\end{scope}
\end{tikzpicture}
    \caption{Diagrammatic proof of the necessary part of \cref{lem:choi} (cf.~\cref{eq:choi_proof_nec}).}
    \label{fig:choi_proof_necc}
\end{figure}


For the sufficiency part, we use the gate teleportation version of the inverse Choi isomophism, given as follows, keeping in mind that $ \mathcal{C}(\mathcal{E}_{SR\rightarrow S'R'})\in \mathcal{S}(\bar{S}\bar{R}S'R')$. 
\begin{equation}
    \mathcal{E}_{SR\rightarrow S'R'}(\rho_{SR})= \bra{\Phi}_{\bar{S}\bar{R}SR} \big(\rho_{SR}\otimes \mathcal{C}(\mathcal{E}_{SR\rightarrow S'R'})\big)\ket{\Phi}_{\bar{S}\bar{R}SR}.
\end{equation}

Then, employing \cref{eq: choi_condition}, we have the following. A diagrammatic version of the proof of the sufficiency part is given in \cref{fig:choi_proof_suff}. 
\begin{align}
 \label{eq:choi_proof_suff}
    \begin{split}
        \tr_{S'}[\mathcal{E}_{SR\rightarrow S'R'}(\rho_{SR})]&=  \bra{\Phi}_{\bar{S}\bar{R}SR} \big(\rho_{SR}\otimes \tr_{S'}[ \mathcal{C}(\mathcal{E}_{SR\rightarrow S'R'})]\big)\ket{\Phi}_{\bar{S}\bar{R}SR}\\
        &= \bra{\Phi}_{\bar{S}\bar{R}SR} \big(\rho_{SR}\otimes \frac{\mathds{1}_{\bar{S}}}{d_{\bar{S}}}\otimes \tr_{\bar{S}S'}[ \mathcal{C}(\mathcal{E}_{SR\rightarrow S'R'})]\big)\ket{\Phi}_{\bar{S}\bar{R}SR}.
    \end{split}
\end{align}

Now, notice that $\tr_{\bar{S}S'}[ \mathcal{C}(\mathcal{E}_{SR\rightarrow S'R'})]\in \mathcal{S}(\bar{R}R')$ can be regarded as the Choi state of some quantum CPTP map $\mathcal{E}_{R\rightarrow R'}: \mathcal{S}(R)\rightarrow \mathcal{S}(R')$, i.e., $\tr_{\bar{S}S'}[ \mathcal{C}(\mathcal{E}_{SR\rightarrow S'R'})]= \mathcal{C}(\mathcal{E}_{R\rightarrow R'})$. This is because of the fact that complete positivity of the map is equivalent to positivity of the Choi state and the trace preserving property of the map is equivalent to the normalisation of the Choi state. $ \mathcal{C}(\mathcal{E}_{SR\rightarrow S'R'})$ being the Choi state of a CPTP map ensures that $\tr_{\bar{S}S'}[ \mathcal{C}(\mathcal{E}_{SR\rightarrow S'R'})]$ too would be the Choi state of a CPTP map between the appropriately reduced spaces. Noticing that the maximally mixed state is the Choi state of the trace map, i.e., $\mathcal{C}(\tr_S)=\frac{\mathds{1}_S}{d_S}$, we have

\begin{align}
    \begin{split}
        \tr_{S'}[\mathcal{E}_{SR\rightarrow S'R'}(\rho_{SR})]&=  
     \bra{\Phi}_{\bar{S}\bar{R}SR} \big(\rho_{SR}\otimes \frac{\mathds{1}_{\bar{S}}}{d_{\bar{S}}}\otimes  \mathcal{C}(\mathcal{E}_{R\rightarrow R'})\big)\ket{\Phi}_{\bar{S}\bar{R}SR}\\
        &= \bra{\Phi}_{\bar{S}\bar{R}SR} \big(\rho_{SR}\otimes \mathcal{C}(\tr_S\otimes \mathcal{E}_{R\rightarrow R'})\big)\ket{\Phi}_{\bar{S}\bar{R}SR}\\
        &=\tr_S\otimes \mathcal{E}_{R\rightarrow R'}(\rho_{SR}).
    \end{split}
\end{align}

This establishes the result as it holds for all input states $\rho_{SR}$.

\begin{figure}
    \centering
    \begin{tikzpicture}[trace/.pic={\draw [thick](-0.4,0)--(0.4,0);\draw [thick](-0.3,0.1)--(0.3,0.1);\draw [thick](-0.2,0.2)--(0.2,0.2);\draw [thick](-0.1,0.3)--(0.1,0.3);},phiprep/.pic={\node[myblue!80!black] at (1.3,-0.65) {\LARGE{$\Phi$}};\draw[draw=myblue!80!black,fill=myblue!40!white, fill opacity=0.4] (0,0) arc (180:360:1.3);
\draw[myblue!80!black] (0,0)--(2.6,0); },phipost/.pic={\node[myblue!80!black] at (1.3,0.65) {\LARGE{$\Phi$}};\draw[draw=myblue!80!black,fill=myblue!40!white, fill opacity=0.4] (0,0) arc (180:0:1.3);
\draw[myblue!80!black] (0,0)--(2.6,0);},scale=0.45, transform shape]

   \draw[thick,black,->] (0.5,2)--node[anchor=east]{\Large$S'$}(0.5,4); \draw[thick,black,->] (2.5,2)--node[anchor=east]{\Large$R'$}(2.5,4); 
   \draw[thick] (-1.5,2) --node[anchor=east]{\Large$\bar{R}$}(-1.5,4); \draw[thick, ->] (-1.5,4)--(-1.5,5.5);
  \draw[thick] (-3.5,2) --node[anchor=east]{\Large$\bar{S}$}(-3.5,4);
 \draw[thick,->] (-3.5,4) to[out=90,in=270] (-5.5,5.5); \draw (0.5,4) pic {trace}; 

     \draw[thick,->] (-5.5,4) to[out=90,in=270] (-3.5,5.5);
     
\node[myblue!80!black] at (-0.5,0.2) {\LARGE{$\mathcal{C}(\mathcal{E}_{SR\rightarrow S'R'})$}};     
\draw[draw=myblue!80!black,fill=myblue!40!white, fill opacity=0.4] (-4.1,2) arc (180:360:3.6); \draw[myblue!80!black] (-4.1,2)--(3.1,2);

\draw (-3.7,5.5) pic {phipost}; \draw (-7.7,5.5) pic {phipost};

\draw[thick] (-5.5,-2.5)--node[anchor=east]{\Large$R$}(-5.5,-1.5); \draw[thick] (-7.5,-2.5)--node[anchor=east]{\Large$S$}(-7.5,-1.5); \draw[thick,->] (-7.5,4)--(-7.5,5.5); \draw[thick] (-7.5,-1.5)--(-7.5,4); \draw[thick] (-5.5,-1.5)--(-5.5,4);

\node at (-9.25,1) {\huge{$=$}};

\begin{scope}[shift={(-14,0)}]
    \draw[thick, fill=LightGray, opacity=0.75 ] (0,0) rectangle node[align=center]{\LARGE{$\mathcal{E}_{SR\rightarrow S'R'}$}} (3,2);
  
  \draw[thick,black,->] (0.5,-2)--node[anchor=east]{\Large$S$}(0.5,0);   \draw[thick,black,->] (2.5,-2)--node[anchor=east]{\Large$R$}(2.5,0); 
  
   \draw[thick,black,->] (0.5,2)--node[anchor=east]{\Large$S'$}(0.5,4); \draw[thick,black,->] (2.5,2)--node[anchor=east]{\Large$R'$}(2.5,4); \draw (0.5,4) pic {trace}; 
\end{scope}

\begin{scope}[shift={(14,0)}]
\node at (-9.25,1) {\huge{$=$}};

\begin{scope}[shift={(0,-3)}]
          \draw[thick,black,->] (0.5,2)--node[anchor=east]{\Large$S'$}(0.5,4); \draw[thick,black,->] (2.5,2)--node[anchor=east]{\Large$R'$}(2.5,4); 
   \draw[thick] (-1.5,2) --node[anchor=east]{\Large$\bar{R}$}(-1.5,7);
  \draw[thick,->] (-3.5,2) --node[anchor=east]{\Large$\bar{S}$}(-3.5,4);
 \draw (0.5,4) pic {trace};  \draw (-3.5,4) pic {trace};

\node[myblue!80!black] at (-0.5,0.2) {\LARGE{$\mathcal{C}(\mathcal{E}_{SR\rightarrow S'R'})$}};     
\draw[draw=myblue!80!black,fill=myblue!40!white, fill opacity=0.4] (-4.1,2) arc (180:360:3.6); \draw[myblue!80!black] (-4.1,2)--(3.1,2);
\end{scope}

\node[myblue!80!black] at (-3.5,3.3) {\LARGE{$\mathds{1}_{\bar{S}}/d_{\bar{S}}$}};
\draw[draw=myblue!80!black,fill=myblue!40!white, fill opacity=0.4] (-4.9,4) arc (180:360:1.4);
\draw[myblue!80!black] (-4.9,4)--(-2.1,4);

 \draw[thick,->] (-3.5,4) to[out=90,in=270] (-5.5,5.5);  \draw[thick, ->] (-1.5,4)--(-1.5,5.5);
     \draw[thick,->] (-5.5,4) to[out=90,in=270] (-3.5,5.5);
\draw (-3.7,5.5) pic {phipost}; \draw (-7.7,5.5) pic {phipost};

\draw[thick] (-5.5,-2.5)--node[anchor=east]{\Large$R$}(-5.5,-1.5); \draw[thick] (-7.5,-2.5)--node[anchor=east]{\Large$S$}(-7.5,-1.5); \draw[thick,->] (-7.5,4)--(-7.5,5.5); \draw[thick] (-7.5,-1.5)--(-7.5,4); \draw[thick] (-5.5,-1.5)--(-5.5,4);
\end{scope}

\begin{scope}[shift={(0,-13)}]
    \draw[thick, fill=LightGray, opacity=0.75 ] (0,0) rectangle node[align=center]{\LARGE{$\mathcal{E}_{SR\rightarrow S'R'}$}} (3,2);
  
  \draw[thick,black,->] (0.5,-1)--node[anchor=east]{\Large$S$}(0.5,0);   \draw[thick,black,->] (2.5,-1)--node[anchor=east]{\Large$R$}(2.5,0); 
    \draw[thick] (2.5,-2.5)--(2.5,-1);
   \draw[thick,black,->] (0.5,2)--node[anchor=east]{\Large$S'$}(0.5,4); \draw[thick,black,->] (2.5,2)--node[anchor=east]{\Large$R'$}(2.5,4); 
     \draw[thick] (-1.5,-2.5) to[out=90,in=270] (0.5,-1);
     \draw[thick] (0.5,-2.5) to[out=90,in=270] (-1.5,-1);
\draw[thick] (-1.5,-1)--(-1.5,2);\draw[thick] (-1.5,2) --node[anchor=east]{\Large$\bar{R}$}(-1.5,4); \draw[thick, ->] (-1.5,4)--(-1.5,5.5);
 
 \draw[thick] (-3.5,2) --node[anchor=east]{\Large$\bar{S}$}(-3.5,4);
\node[myblue!80!black] at (-3.5,1.3) {\LARGE{$\mathds{1}_{\bar{S}}/d_{\bar{S}}$}};
\draw[draw=myblue!80!black,fill=myblue!40!white, fill opacity=0.4] (-4.9,2) arc (180:360:1.4);
\draw[myblue!80!black] (-4.9,2)--(-2.1,2);
\draw[thick,->] (-3.5,-2.5) --node[anchor=east]{\Large$\bar{S}$}(-3.5,-0.5); \draw (-3.5,-0.5) pic {trace}; 
 
 \draw[thick,->] (-3.5,4) to[out=90,in=270] (-5.5,5.5);

     \draw[thick,->] (-5.5,4) to[out=90,in=270] (-3.5,5.5);

\draw (0.2,-2.5) pic {phiprep}; \draw (-3.7,-2.5) pic {phiprep};
\draw (0.5,4) pic {trace}; 

\draw (-3.7,5.5) pic {phipost}; \draw (-7.7,5.5) pic {phipost};

\draw[thick] (-5.5,-2.5)--node[anchor=east]{\Large$R$}(-5.5,-1.5); \draw[thick] (-7.5,-2.5)--node[anchor=east]{\Large$S$}(-7.5,-1.5); \draw[thick,->] (-7.5,4)--(-7.5,5.5); \draw[thick] (-7.5,-1.5)--(-7.5,4); \draw[thick] (-5.5,-1.5)--(-5.5,4);
\node at (-9.25,1) {\huge{$=$}};
\end{scope}

\begin{scope}[shift={(14,-13)}]
    \draw[thick, fill=LightGray, opacity=0.75 ] (0,3) rectangle node[align=center]{\LARGE{$\mathcal{E}_{SR\rightarrow S'R'}$}} (3,5);
  
  \draw[thick,black,->] (0.5,-1)--node[anchor=east]{\Large$S$}(0.5,0);   \draw[thick,black,->] (2.5,2)--node[anchor=east]{\Large$R$}(2.5,3); \draw[thick,black,->] (0.5,2)--node[anchor=east]{\Large$S$}(0.5,3);
    \draw[thick] (2.5,-2.5)--(2.5,2);
  \draw[thick,black,->] (0.5,5)--node[anchor=east]{\large$\mathbf{S'}$}(0.5,7);  \draw[thick,black,->] (2.5,5)--node[anchor=east]{\large$\mathbf{R'}$}(2.5,7); 
     \draw[thick] (-1.5,-2.5) to[out=90,in=270] (0.5,-1);
     \draw[thick] (0.5,-2.5) to[out=90,in=270] (-1.5,-1);
\draw[thick] (-1.5,-1)--(-1.5,2);\draw[thick] (-1.5,2) --node[anchor=east]{\Large$\bar{R}$}(-1.5,4); \draw[thick, ->] (-1.5,4)--(-1.5,5.5);
 
\draw[thick] (-3.5,-2.5)--node[anchor=east,pos=0.85]{\Large$\bar{S}$}(-3.5,4); \draw (0.5,0) pic {trace}; \draw (0.5,7) pic {trace};

\node[myblue!80!black] at (0.5,1.3) {\LARGE{$\mathds{1}_{S}/d_{S}$}};
\draw[draw=myblue!80!black,fill=myblue!40!white, fill opacity=0.4] (-0.9,2) arc (180:360:1.4);
\draw[myblue!80!black] (-0.9,2)--(1.9,2);
 
 \draw[thick,->] (-3.5,4) to[out=90,in=270] (-5.5,5.5);

     \draw[thick,->] (-5.5,4) to[out=90,in=270] (-3.5,5.5);

\draw (0.2,-2.5) pic {phiprep}; \draw (-3.7,-2.5) pic {phiprep};

\draw (-3.7,5.5) pic {phipost}; \draw (-7.7,5.5) pic {phipost};

\draw[thick] (-5.5,-2.5)--node[anchor=east]{\Large$R$}(-5.5,-1.5); \draw[thick] (-7.5,-2.5)--node[anchor=east]{\Large$S$}(-7.5,-1.5); \draw[thick,->] (-7.5,4)--(-7.5,5.5); \draw[thick] (-7.5,-1.5)--(-7.5,4); \draw[thick] (-5.5,-1.5)--(-5.5,4);
\node at (-9.25,1) {\huge{$=$}};
\end{scope}

\begin{scope}[shift={(0,-26)}]
    \draw[thick, fill=LightGray, opacity=0.75 ] (1.5,0) rectangle node[align=center]{\LARGE{$\mathcal{E}_{R\rightarrow R'}$}} (3.5,2);
  
  \draw[thick,black,->] (0.5,-1)--node[anchor=east]{\Large$S$}(0.5,0);   \draw[thick,black,->] (2.5,-1)--node[anchor=east]{\Large$R$}(2.5,0); 
    \draw[thick] (2.5,-2.5)--(2.5,-1);
  \draw[thick,black,->] (2.5,2)--node[anchor=east]{\Large$R'$}(2.5,4); 
     \draw[thick] (-1.5,-2.5) to[out=90,in=270] (0.5,-1);
     \draw[thick] (0.5,-2.5) to[out=90,in=270] (-1.5,-1);
\draw[thick] (-1.5,-1)--(-1.5,2);\draw[thick] (-1.5,2) --node[anchor=east]{\Large$\bar{R}$}(-1.5,4); \draw[thick, ->] (-1.5,4)--(-1.5,5.5);
 
\draw[thick] (-3.5,-2.5)--node[anchor=east,pos=0.85]{\Large$\bar{S}$}(-3.5,4); \draw (0.5,0) pic {trace}; 
 
 \draw[thick,->] (-3.5,4) to[out=90,in=270] (-5.5,5.5);

     \draw[thick,->] (-5.5,4) to[out=90,in=270] (-3.5,5.5);

\draw (0.2,-2.5) pic {phiprep}; \draw (-3.7,-2.5) pic {phiprep};

\draw (-3.7,5.5) pic {phipost}; \draw (-7.7,5.5) pic {phipost};

\draw[thick] (-5.5,-2.5)--node[anchor=east]{\Large$R$}(-5.5,-1.5); \draw[thick] (-7.5,-2.5)--node[anchor=east]{\Large$S$}(-7.5,-1.5); \draw[thick,->] (-7.5,4)--(-7.5,5.5); \draw[thick] (-7.5,-1.5)--(-7.5,4); \draw[thick] (-5.5,-1.5)--(-5.5,4);
\node at (-9.25,1) {\huge{$=$}};
\node at (5.25,1) {\huge{$=$}};
\end{scope}

\begin{scope}[shift={(7,-26)}]
    \draw[thick, fill=LightGray, opacity=0.75 ] (1.5,0) rectangle node[align=center]{\LARGE{$\mathcal{E}_{R\rightarrow R'}$}} (3.5,2);
  
  \draw[thick,black,->] (0.5,-2)--node[anchor=east]{\Large$S$}(0.5,0);   \draw[thick,black,->] (2.5,-2)--node[anchor=east]{\Large$R$}(2.5,0); 
  
\draw[thick,black,->] (2.5,2)--node[anchor=east]{\Large$R'$}(2.5,4); \draw (0.5,0) pic {trace}; 
\end{scope}
\end{tikzpicture}
    \caption{Diagrammatic proof of the sufficient part of \cref{lem:choi} (cf.~\cref{eq:choi_proof_suff}). Here, the upright semi-circles labelled $\Phi$ physically represent a post-selection on the outcome corresponding to the Bell state $\Phi$, following a joint measurement in the Bell basis on the associated systems. Mathematically, they correspond to projectors on to the Bell state $\Phi$.}
    \label{fig:choi_proof_suff}
\end{figure}
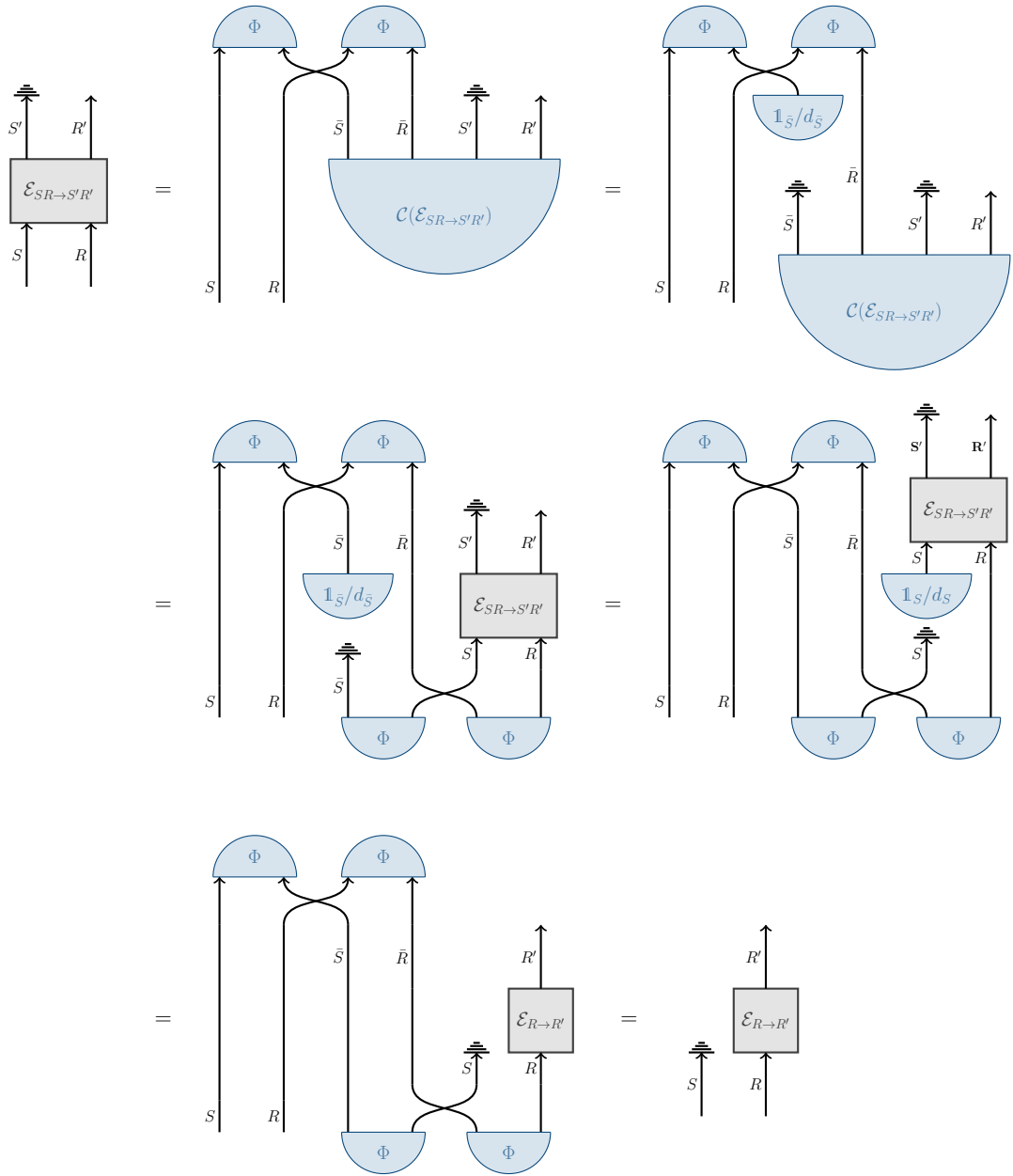

\section{Proof of Theorem~\ref{thm:causality_sequence}}
\label{sec:proof_causality_sequence}
We start by proving a simple version of the theorem (which follows quite easily from previous works such as \cite{Portmann_2017,Ormrod2023}), and then generalise it. Suppose that $B$ does not signal to $C$ in a CPTP map $\mathcal{E}_{AB\rightarrow CD}:\mathcal{S}(AB)\rightarrow \mathcal{S}(CD)$. Then, we can show that there must exist CPTP maps $\mathcal{E}_1: \mathcal{S}(A)\rightarrow \mathcal{S}(CQ) $ and $\mathcal{E}_2:\mathcal{S}(QB)\rightarrow \mathcal{S}(D)$ such that (see also \cref{fig: sequence_proof1})
\begin{equation}
\label{eq:proof_causality_sequence1}
    \mathcal{E}_{AB\rightarrow CD}= \big(\mathcal{I}_{A}\otimes\mathcal{E}_2\big)\circ\big(\mathcal{E}_1\otimes \mathcal{I}_B\big).
\end{equation}

To establish this, we first use the fact that $B$ does not signal to $C$ in $\mathcal{E}_{AB\rightarrow CD}$ is equivalent to the condition that there exists a CPTP map $\mathcal{E}_{A\rightarrow C}: \mathcal{S}(A)\rightarrow \mathcal{S}(C)$ such that
\begin{equation}
    \tr_{D}\circ \mathcal{E}_{AB\rightarrow CD}=\mathcal{E}_{A\rightarrow C}\circ \tr_B.
\end{equation}
This follows from \cref{eq: nonsignalling2} (which is implied by the results of \cite{Ormrod2023}).

Then using a result shown in \cite{Portmann_2017} (Lemma 5.1), it follows that there are Stinespring representations $U_{AB\rightarrow CDP}: \mathcal{S}(AB)\rightarrow \mathcal{S}(CDP)$ and $U_{A\rightarrow CQ}: \mathcal{S}(A)\rightarrow \mathcal{S}(CQ)$ (which are isometries) of $\mathcal{E}_{AB\rightarrow CD}$ and $\mathcal{E}_{A\rightarrow C}$ respectively, and an isometry $V_{QB\rightarrow DP}: \mathcal{S}(QB)\rightarrow \mathcal{S}(DP)$ such that 
\begin{equation}
\label{eq:proof_causality_sequence2}
  U_{AB\rightarrow CDP}= \big(V_{QB\rightarrow DP}\otimes \mathcal{I}_{A}\big)\circ\big(\mathcal{I}_B\otimes U_{A\rightarrow CQ}\big).
\end{equation}

Now applying $\tr_P$ to both sides of the above equation, and using the fact that $\mathcal{E}_{AB\rightarrow CD}=\tr_P\circ U_{AB\rightarrow CDP}$ (as $U_{AB\rightarrow CDP}$ is a Stinespring representation of $\mathcal{E}_{AB\rightarrow CD}$), we have
\begin{equation}
\label{eq:proof_causality_sequence3}
  \mathcal{E}_{AB\rightarrow CD} =\tr_P\circ  \big(V_{QB\rightarrow DP}\otimes \mathcal{I}_{A}\big)\circ\big(\mathcal{I}_B\otimes U_{A\rightarrow QC}\big).
\end{equation}

Defining $\mathcal{E}_1\coloneqq U_{A\rightarrow CQ}$ and $\mathcal{E}_2\coloneqq\tr_P\circ V_{QB\rightarrow DP}$, we immediately obtain the desired equation \cref{eq:proof_causality_sequence1}.

\begin{figure}
    \centering
    \begin{tikzpicture}[trace/.pic={\draw [thick](-0.4,0)--(0.4,0);\draw [thick](-0.3,0.1)--(0.3,0.1);\draw [thick](-0.2,0.2)--(0.2,0.2);\draw [thick](-0.1,0.3)--(0.1,0.3);}, scale=0.6, transform shape]

\draw[thick, fill=LightGray, opacity=0.75 ] (0,0) rectangle node[align=center]{\Large{$\mathcal{E}_{AB\rightarrow CD}$}} (3,2);

 \draw[thick,black,->] (0.5,-2)--node[anchor=east]{\large$A$}(0.5,0);  \draw[thick,black,->] (2.5,-2)--node[anchor=east]{\large$B$}(2.5,0); \draw[thick,black,->] (0.5,2)--node[anchor=east]{\large$C$}(0.5,4); \draw[thick,black,->] (2.5,2)--node[anchor=east]{\large$D$}(2.5,4);

   \node[align=center, black] at (4.5,1) {\Huge{$\mathbf{=}$}};

\draw[thick, fill=LightGray, opacity=0.75 ] (6.5,0) rectangle node[align=center]{\Large{$\mathcal{E}_1$}} (8.5,2);
\draw[thick,black,->] (7.5,-2)--node[anchor=east]{\large$A$}(7.5,0);  
\draw[thick,black,->] (7.5,2)--node[anchor=east]{\large$C$}(7.5,4);

\draw[thick, fill=LightGray, opacity=0.75 ] (10.5,0) rectangle node[align=center]{\Large{$\mathcal{E}_2$}} (12.5,2);

\draw[thick,black,->] (8.5,1)--node[anchor=south]{\large$Q$}(10.5,1);

\draw[thick,black,->] (11.5,-2)--node[anchor=east]{\large$B$}(11.5,0); 
\draw[thick,black,->] (11.5,2)--node[anchor=east]{\large$D$}(11.5,4);

 \node[align=center, black] at (14.5,1) {\Huge{$\mathbf{=}$}};

 \begin{scope}[shift={(10,0)}]
 \draw[thick, fill=LightGray, opacity=0.75 ] (6.5,0) rectangle node[align=center]{\Large{$U_{A\rightarrow CQ}$}} (8.5,2);
\draw[thick,black,->] (7.5,-2)--node[anchor=east]{\large$A$}(7.5,0);  
\draw[thick,black,->] (7.5,2)--node[anchor=east]{\large$C$}(7.5,4);

\draw[thick, fill=LightGray, opacity=0.75 ] (10.5,0) rectangle node[align=center]{\Large{$V_{QB\rightarrow DP}$}} (14.5,2);

\draw[thick,black,->] (8.5,1)--node[anchor=south]{\large$Q$}(10.5,1);

\draw[thick,black,->] (12.5,-2)--node[anchor=east]{\large$B$}(12.5,0); 
\draw[thick,black,->] (11.5,2)--node[anchor=east]{\large$D$}(11.5,4); 

\draw[thick,black,->] (13.5,2)--node[anchor=east]{\large$P$}(13.5,4); 
\draw (13.5,4) pic {trace};

 \end{scope}
\end{tikzpicture}
    \caption{Diagrammatic representation of \cref{eq:proof_causality_sequence1} and~\cref{eq:proof_causality_sequence3}. If $B$ does not signal to $C$ in a CPTP map $\mathcal{E}_{AB\rightarrow CD}$ (left most), then it admits a decomposition as shown in the middle where $\mathcal{E}_1$ and $\mathcal{E}_2$ are CPTP maps, or in terms of isometries $U_{A\rightarrow CQ}$ and $V_{QB\rightarrow DP}$ as shown in the right most diagram. }
    \label{fig: sequence_proof1}
\end{figure}
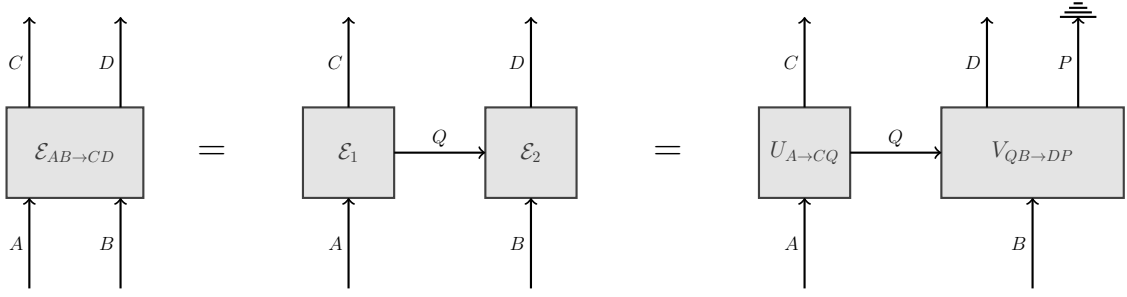

We now apply \cref{eq:proof_causality_sequence1} recursively to prove the general statement of the theorem. Consider the map $\mathcal{E}:\mathcal{S}(E_0S_1...S_n)\rightarrow \mathcal{S}(S_1'...S_n'E_n)$, and recall that we are given that this map is such that $S_i$ does not signal to $S_1',...,S_{i-1}'$ for every $i\in\{2,...,n\}$. Using the non-signalling relation $S_n$ does not signal to $S_1',...,S_{n-1}'$ along with \cref{eq:proof_causality_sequence1} (with $E_0S_1...S_{n-1}$ playing the role of $A$, $S_n$ playing the role of $B$,  $S_1',...,S_{n-1}'$ playing the role of $C$ and $S_n'E_n$ playing the role of $D$), we immediately have that there must exist some CPTP maps $\tilde{\mathcal{E}}_{n-1}:\mathcal{S}(E_0S_1...S_{n-1})\rightarrow \mathcal{S}(S_1'...S_{n-1}'E_{n-1})$ and $\mathcal{E}_n:\mathcal{S}(E_{n-1}S_n)\rightarrow \mathcal{S}(S_n'E_n)$ (where $E_{n-1}$ plays the role of $Q$), such that the following holds (illustrated in \cref{fig:sequence_proof2}):

\begin{equation}
\label{eq:proof_causality_sequence4}
    \mathcal{E}=\big(\mathcal{I}_{E_0S_1...S_{n-1}}\otimes \mathcal{E}_n\big)\circ\big(\tilde{\mathcal{E}}_{n-1}\otimes \mathcal{I}_{S_n}\big).
\end{equation}

\begin{figure}[h]
    \centering
    \begin{tikzpicture}[trace/.pic={\draw [thick](-0.4,0)--(0.4,0);\draw [thick](-0.3,0.1)--(0.3,0.1);\draw [thick](-0.2,0.2)--(0.2,0.2);\draw [thick](-0.1,0.3)--(0.1,0.3);}, scale=0.48, transform shape]
 
\draw[thick, fill=LightGray, opacity=0.75 ] (0,0) rectangle node[align=center]{\LARGE{$\mathcal{E}$}} (5.5,2);

 \draw[thick,black,->] (0.5,-2)--node[anchor=east]{\Large$S_1$}(0.5,0); 

  \draw[thick,black,->] (2,-2)--node[anchor=east]{\Large$S_2$}(2,0); 
  
 \draw[thick,black,->] (5,-2)--node[anchor=east]{\Large$S_n$}(5,0); 

\node at (3.5,-1) {$\dots$};

  \draw[thick,black,->] (0.5,2)--node[anchor=east]{\Large$S_1'$}(0.5,4); 

  \draw[thick,black,->] (2,2)--node[anchor=east]{\Large$S_2'$}(2,4); 
  
 \draw[thick,black,->] (5,2)--node[anchor=east]{\Large$S_n'$}(5,4); 

 \node at (3.5,3) {$\dots$};

  \draw[thick,black,->] (-2,1)--node[anchor=south]{\Large$E_0$}(0,1); 
    \draw[thick,black,->] (5.5,1)--node[anchor=south]{\Large$E_n$}(7.5,1); 

    \node[align=center, black] at (9.5,1) {\Huge{$\mathbf{=}$}};

  \draw[thick, fill=LightGray, opacity=0.75 ] (13.5,0) rectangle node[align=center]{\LARGE{$\tilde{\mathcal{E}}_{n-1}$}} (19.5,2);  
  \draw[thick,black,->] (11.5,1)--node[anchor=south]{\Large$E_0$}(13.5,1); 
    \draw[thick,black,->] (14.5,-2)--node[anchor=east]{\Large$S_1$}(14.5,0); 
 \draw[thick,black,->] (14.5,2)--node[anchor=east]{\Large$S_1'$}(14.5,4);

 \draw[thick,black,->] (18.5,-2)--node[anchor=east]{\Large$S_{n-1}$}(18.5,0); 
 \draw[thick,black,->] (18.5,2)--node[anchor=east]{\Large$S_{n-1}'$}(18.5,4); 

\node at (16.5,-1) {$\dots$};\node at (16.5,3) {$\dots$};

\draw[thick,black,->] (19.5,1)--node[anchor=south]{\Large$E_{n-1}$}(21.5,1);

  \draw[thick, fill=LightGray, opacity=0.75 ] (21.5,0) rectangle node[align=center]{\LARGE{$\mathcal{E}_n$}} (23.5,2);  
    \draw[thick,black,->] (22.5,-2)--node[anchor=east]{\Large$S_n$}(22.5,0); 
 \draw[thick,black,->] (22.5,2)--node[anchor=east]{\Large$S_n'$}(22.5,4); 
 \draw[thick,black,->] (23.5,1)--node[anchor=south]{\Large$E_{n}$}(25.5,1);

\end{tikzpicture}
    \caption{Diagrammatic representation of \cref{eq:proof_causality_sequence4}.}
    \label{fig:sequence_proof2}
\end{figure}

Now consider the non-signalling relation $S_{n-1}$ does not signal to $S_1'...S_{n-2}'$ in $\mathcal{E}$. It easy to see that this implies the same non-signalling relation, $S_{n-1}$ does not signal to $S_1'...S_{n-2}'$ in $\tilde{\mathcal{E}}_{n-1}$. This is because, writing out explicitly the non-signalling condition from $S_{n-1}$ to $S_1'...S_{n-2}'$ in $\mathcal{E}$ we have the following for all CPTP maps $\mathcal{M}_{n-1}:\mathcal{S}(S_{n-1})\rightarrow\mathcal{S}(S_{n-1})$
\begin{equation}
    \tr_{S'_{n-1}S'_nE_n}\circ \mathcal{E}=\tr_{S'_{n-1}S'_nE_n}\circ \mathcal{E}\circ \mathcal{M}_{n-1}.
\end{equation}

Using \cref{eq:proof_causality_sequence4} (illustrated in \cref{fig:sequence_proof2}), it is immediate that $\tr_{S'_nE_n}\circ \mathcal{E}=\tr_{E_{n-1}}\circ \tilde{\mathcal{E}}_{n-1}\circ\tr_{S_n}$ (illustrated in \cref{fig:sequence_proof3}). Therefore, the above is equivalent to 
\begin{equation}
\label{eq:proof_causality_sequence5}
\tr_{S'_{n-1}E_{n-1}}\circ \tilde{\mathcal{E}}_{n-1}\circ \tr_{S_n}=\tr_{S'_{n-1}E_{n-1}}\circ \tilde{\mathcal{E}}_{n-1}\circ \tr_{S_n}\circ \mathcal{M}_{n-1}, \quad \forall \mathcal{M}_{n-1}.
\end{equation}

\begin{figure}[h]
    \centering
    \begin{tikzpicture}[trace/.pic={\draw [thick](-0.4,0)--(0.4,0);\draw [thick](-0.3,0.1)--(0.3,0.1);\draw [thick](-0.2,0.2)--(0.2,0.2);\draw [thick](-0.1,0.3)--(0.1,0.3);}, scale=0.48, transform shape]
 
\draw[thick, fill=LightGray, opacity=0.75 ] (0,0) rectangle node[align=center]{\LARGE{$\mathcal{E}$}} (5.5,2);

 \draw[thick,black,->] (0.5,-2)--node[anchor=east]{\Large$S_1$}(0.5,0); 

  \draw[thick,black,->] (2,-2)--node[anchor=east]{\Large$S_2$}(2,0); 
  
 \draw[thick,black,->] (5,-2)--node[anchor=east]{\Large$S_n$}(5,0); 

\node at (3.5,-1) {$\dots$};

  \draw[thick,black,->] (0.5,2)--node[anchor=east]{\Large$S_1'$}(0.5,4); 

  \draw[thick,black,->] (2,2)--node[anchor=east]{\Large$S_2'$}(2,4); 
  
 \draw[thick,black,->] (5,2)--node[anchor=east]{\Large$S_n'$}(5,4); 

 \node at (3.5,3) {$\dots$};

  \draw[thick,black,->] (-2,1)--node[anchor=south]{\Large$E_0$}(0,1); 
    \draw[thick,black,->] (5.5,1)--node[anchor=south]{\Large$E_n$}(7.5,1); 

    \draw (5,4) pic {trace}; \draw[rotate around={270:(7.5,1)}] (7.5,1) pic {trace};

    \node[align=center, black] at (9.5,1) {\Huge{$\mathbf{=}$}};

  \draw[thick, fill=LightGray, opacity=0.75 ] (13.5,0) rectangle node[align=center]{\LARGE{$\tilde{\mathcal{E}}_{n-1}$}} (19.5,2);  
  \draw[thick,black,->] (11.5,1)--node[anchor=south]{\Large$E_0$}(13.5,1); 
    \draw[thick,black,->] (14.5,-2)--node[anchor=east]{\Large$S_1$}(14.5,0); 
 \draw[thick,black,->] (14.5,2)--node[anchor=east]{\Large$S_1'$}(14.5,4);

 \draw[thick,black,->] (18.5,-2)--node[anchor=east]{\Large$S_{n-1}$}(18.5,0); 
 \draw[thick,black,->] (18.5,2)--node[anchor=east]{\Large$S_{n-1}'$}(18.5,4); 

\node at (16.5,-1) {$\dots$};\node at (16.5,3) {$\dots$};

\draw[thick,black,->] (19.5,1)--node[anchor=south]{\Large$E_{n-1}$}(21.5,1);

    \draw[thick,black,->] (22.5,-2)--node[anchor=east]{\Large$S_n$}(22.5,0); 
    \draw (22.5,0) pic {trace}; \draw[rotate around={270:(21.5,1)}] (21.5,1) pic {trace};

\end{tikzpicture}
    \caption{Diagrammatic representation of the fact that $\tr_{S'_nE_n}\circ \mathcal{E}=\tr_{E_{n-1}}\circ \tilde{\mathcal{E}}_{n-1}\circ\tr_{S_n}$. This readily follows from \cref{fig:sequence_proof2} (or equivalently, \cref{eq:proof_causality_sequence4}) and is used to derive \cref{eq:proof_causality_sequence5}.}
    \label{fig:sequence_proof3}
\end{figure}

From this we obtain the following.

\begin{equation}
\label{eq:proof_causality_sequence6}
\tr_{S'_{n-1}E_{n-1}}\circ \tilde{\mathcal{E}}_{n-1}=\tr_{S'_{n-1}E_{n-1}}\circ \tilde{\mathcal{E}}_{n-1}\circ \mathcal{M}_{n-1}, \quad \forall \mathcal{M}_{n-1}.
\end{equation}
That \cref{eq:proof_causality_sequence5} implies \cref{eq:proof_causality_sequence6} can be established by contradiction. Suppose that \cref{eq:proof_causality_sequence5} holds but \cref{eq:proof_causality_sequence6} fails to hold, i.e., $\exists$ $\mathcal{M}_{n-1}$ and $\rho_{E_0S_1...S_{n-1}}\in \mathcal{S}(E_0S_1...S_{n-1})$ such that
\begin{equation}
 \label{eq:proof_causality_sequence7}
 \tr_{S'_{n-1}E_{n-1}}\circ \tilde{\mathcal{E}}_{n-1}(\rho_{E_0S_1...S_{n-1}})\neq \tr_{S'_{n-1}E_{n-1}}\circ \tilde{\mathcal{E}}_{n-1}\circ \mathcal{M}_{n-1} (\rho_{E_0S_1...S_{n-1}}).
\end{equation} Then considering any normalised state $\sigma_{S_n}\in \mathcal{S}(S_n)$, we have the following which contradicts the assumption that \cref{eq:proof_causality_sequence5} holds.

\begin{align}
    \begin{split}
        &\tr_{S'_{n-1}E_{n-1}}\circ \tilde{\mathcal{E}}_{n-1}\circ \tr_{S_n}(\rho_{E_0S_1...S_{n-1}}\otimes \sigma_{S_n})\\
        =& \tr_{S'_{n-1}E_{n-1}}\circ \tilde{\mathcal{E}}_{n-1}(\rho_{E_0S_1...S_{n-1}}).\tr_{S_n}(\sigma_{S_n})\\
        \neq &\tr_{S'_{n-1}E_{n-1}}\circ \tilde{\mathcal{E}}_{n-1}\circ \mathcal{M}_{n-1} (\rho_{E_0S_1...S_{n-1}}).\tr_{S_n}(\sigma_{S_n})\\
        =&\tr_{S'_{n-1}E_{n-1}}\circ \tilde{\mathcal{E}}_{n-1}\circ \tr_{S_n}\circ \mathcal{M}_{n-1}(\rho_{E_0S_1...S_{n-1}}\otimes \sigma_{S_n}),
    \end{split}
\end{align}

where we have used \cref{eq:proof_causality_sequence7} and the fact that $\sigma_{S_n}$ is a normalised state. Therefore, we have established \cref{eq:proof_causality_sequence6} which is indeed the non-signalling from  $S_{n-1}$ to $S_1'...S_{n-2}'$ in $\tilde{\mathcal{E}}_{n-1}$. This allows us to apply \cref{eq:proof_causality_sequence1} to decompose $\tilde{\mathcal{E}}_{n-1}$ in terms of some maps $\tilde{\mathcal{E}}_{n-2}: \mathcal{S}(E_0S_1...S_{n-2})\rightarrow \mathcal{S}(S_1'...S_{n-2}'E_{n-2})$ and $\mathcal{E}_{n-1}:\mathcal{S}(E_{n-2}S_{n-1})\rightarrow \mathcal{S}(S_{n-1}'E_{n-1})$, as

\begin{equation}
   \tilde{\mathcal{E}}_{n-1}=\big(\mathcal{I}_{E_0S_1...S_{n-2}}\otimes\mathcal{E}_{n-1}\big)\circ\big(\tilde{\mathcal{E}}_{n-2}\otimes \mathcal{I}_{S_{n-1}}\big).
\end{equation}

Plugging this back into \cref{eq:proof_causality_sequence4}, we obtain the decomposition of $\mathcal{E}$ (illustrated in \cref{fig:sequence_proof4}):

\begin{equation}
 \label{eq:proof_causality_sequence8}
  \mathcal{E}=\big(\mathcal{I}_{E_0S_1...S_{n-1}}\otimes \mathcal{E}_n\big)\circ\big(\mathcal{I}_{E_0S_1...S_{n-2}}\otimes \mathcal{E}_{n-1}\otimes\mathcal{I}_{S_n}\big)\circ\big(\tilde{\mathcal{E}}_{n-2}\otimes \mathcal{I}_{S_{n-1}S_n}\big).   
\end{equation}

We can then consider the non-signalling relation $S_{n-2}$ does not signal to $S_1'...S_{n-3}'$ in $\mathcal{E}$, and by the same argument as above, arrive at the fact that this implies the same non-signalling relation in $\tilde{\mathcal{E}}_{n-2}$, which will allow us to decompose $\tilde{\mathcal{E}}_{n-2}$ further using \cref{eq:proof_causality_sequence1}. Repeating the argument recursively, and substituting all these decompositions back into \cref{eq:proof_causality_sequence4}, we get the desired decomposition of $\mathcal{E}$ into a sequence of maps $\{\mathcal{E}_j:\mathcal{S}(E_{j-1}S_j)\rightarrow\mathcal{S}(S_j'E_j)\}_{j=1}^n$ given in \cref{eq: causality_sequence}, which is repeated below for convenience. This establishes the theorem.
\begin{equation}
\mathcal{E}=\big(\mathcal{I}_{E_0S_1...S_{n-1}}\otimes\mathcal{E}_n\big)\circ\dots\circ \big(\mathcal{I}_{E_0S_1}\otimes \mathcal{E}_2\otimes \mathcal{I}_{S_3...S_n}\big)\circ\big(\mathcal{E}_1\otimes \mathcal{I}_{S_2...S_n}\big).
\end{equation}

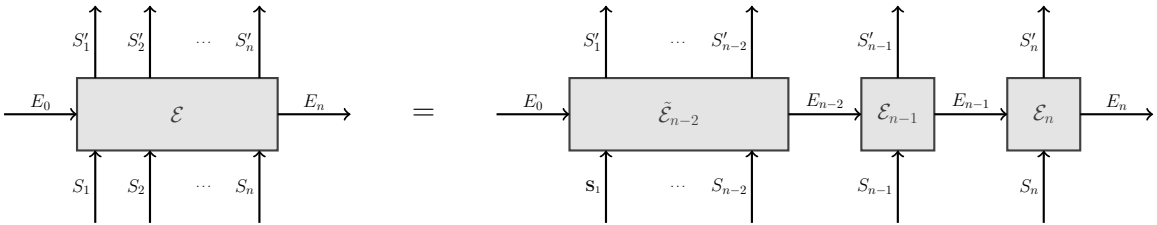
\begin{figure}[h]
    \centering
    \begin{tikzpicture}[trace/.pic={\draw [thick](-0.4,0)--(0.4,0);\draw [thick](-0.3,0.1)--(0.3,0.1);\draw [thick](-0.2,0.2)--(0.2,0.2);\draw [thick](-0.1,0.3)--(0.1,0.3);}, scale=0.48, transform shape]
 
\draw[thick, fill=LightGray, opacity=0.75 ] (0,0) rectangle node[align=center]{\LARGE{$\mathcal{E}$}} (5.5,2);

 \draw[thick,black,->] (0.5,-2)--node[anchor=east]{\Large$S_1$}(0.5,0); 

  \draw[thick,black,->] (2,-2)--node[anchor=east]{\Large$S_2$}(2,0); 
  
 \draw[thick,black,->] (5,-2)--node[anchor=east]{\Large$S_n$}(5,0); 

\node at (3.5,-1) {$\dots$};

  \draw[thick,black,->] (0.5,2)--node[anchor=east]{\Large$S_1'$}(0.5,4); 

  \draw[thick,black,->] (2,2)--node[anchor=east]{\Large$S_2'$}(2,4); 
  
 \draw[thick,black,->] (5,2)--node[anchor=east]{\Large$S_n'$}(5,4); 

 \node at (3.5,3) {$\dots$};

  \draw[thick,black,->] (-2,1)--node[anchor=south]{\Large$E_0$}(0,1); 
    \draw[thick,black,->] (5.5,1)--node[anchor=south]{\Large$E_n$}(7.5,1); 

    \node[align=center, black] at (9.5,1) {\Huge{$\mathbf{=}$}};

  \draw[thick, fill=LightGray, opacity=0.75 ] (13.5,0) rectangle node[align=center]{\LARGE{$\tilde{\mathcal{E}}_{n-2}$}} (19.5,2);  
  \draw[thick,black,->] (11.5,1)--node[anchor=south]{\Large$E_0$}(13.5,1); 
    \draw[thick,black,->] (14.5,-2)--node[anchor=east]{\large$\mathbf{S}_1$}(14.5,0); 
 \draw[thick,black,->] (14.5,2)--node[anchor=east]{\Large$S_1'$}(14.5,4);

 \draw[thick,black,->] (18.5,-2)--node[anchor=east]{\Large$S_{n-2}$}(18.5,0); 
 \draw[thick,black,->] (18.5,2)--node[anchor=east]{\Large$S_{n-2}'$}(18.5,4); 

\node at (16.5,-1) {$\dots$};\node at (16.5,3) {$\dots$};

\draw[thick,black,->] (19.5,1)--node[anchor=south]{\Large$E_{n-2}$}(21.5,1);

  \draw[thick, fill=LightGray, opacity=0.75 ] (21.5,0) rectangle node[align=center]{\LARGE{$\mathcal{E}_{n-1}$}} (23.5,2);  
    \draw[thick,black,->] (22.5,-2)--node[anchor=east]{\Large$S_{n-1}$}(22.5,0); 
 \draw[thick,black,->] (22.5,2)--node[anchor=east]{\Large$S_{n-1}'$}(22.5,4); 
 \draw[thick,black,->] (23.5,1)--node[anchor=south]{\Large$E_{n-1}$}(25.5,1); 

 \begin{scope}[shift={(4,0)}]

  \draw[thick, fill=LightGray, opacity=0.75 ] (21.5,0) rectangle node[align=center]{\LARGE{$\mathcal{E}_n$}} (23.5,2);  
    \draw[thick,black,->] (22.5,-2)--node[anchor=east]{\Large$S_n$}(22.5,0); 
 \draw[thick,black,->] (22.5,2)--node[anchor=east]{\Large$S_n'$}(22.5,4); 
 \draw[thick,black,->] (23.5,1)--node[anchor=south]{\Large$E_{n}$}(25.5,1); 
 \end{scope}

\end{tikzpicture}
    \caption{Diagrammatic representation of \cref{eq:proof_causality_sequence8}.}
    \label{fig:sequence_proof4}
\end{figure}

\section{Proof of Theorem~\ref{thm:squashing}}
\label{sec:proof_squashing}
In this proof we assume that, unless explicitly stated otherwise, sums over $N$ only include terms $N \geq 1$. Similarly we assume that sums over $k$ ($l$) only include odd (even) terms $\leq N$.
We start by rewriting the expressions for $M^{(0)}$ and $M^{(1)}$ in \cref{eq:mmt_ops_Bob} as
\begin{align}
    \label{eq:mmt_ops_Bob2}
    M^{(0,1)} &= \frac{1}{2}\mathds{1} - \frac{1}{2}\ketbra{0,0}{0,0} \pm \frac{1}{2} \sum_{N=1}^{\infty} \left( \ketbra{N, 0}{N, 0}_{AB} - \ketbra{0, N}{0, N}_{AB} \right).
\end{align}
    Using \cref{eq:BS_state_trsf} we can now rewrite the states $\ket{N,0}_{AB}$ and $\ket{0, N}_{AB}$ after the BS in terms of the states on the systems $SR$ before the BS to obtain:
\begin{equation}
\begin{aligned}
    \ket{N,0}_{AB} &= \frac{1}{\sqrt{N!}}\left(a_A^\dagger\right)^N \ket{0,0} = \frac{1}{\sqrt{2^N}}\sum_{m=0}^{N} \sqrt{\binom{N}{m}} \ket{N - m, m}_{SR}, \\
    \ket{0,N}_{AB} &= \frac{1}{\sqrt{N!}}\left(a_B^\dagger\right)^N \ket{0,0} = \frac{1}{\sqrt{2^N}}\sum_{m=0}^{N} \sqrt{\binom{N}{m}} (-1)^m\ket{N - m, m}_{SR}.
\end{aligned}
\end{equation}
Inserting this into the expressions of \cref{eq:mmt_ops_Bob2} yields:
\begin{equation}
\begin{aligned}
    \label{eq:mmt_simplified}
    M^{(0,1)} =& \frac{1}{2}\mathds{1} - \frac{1}{2}\ketbra{0,0}{0,0}_{SR} \\
    &\pm \frac{1}{2} \sum_{N}^{\infty} \sum_{m,n=0}^{N} \frac{1}{2^N}\sqrt{\binom{N}{m}}\sqrt{\binom{N}{n}} (1 - (-1)^{m+n}) \ketbra{N-m, m}{N-n, n}_{SR}.
\end{aligned}
\end{equation}
In a similar fashion we can write:
\begin{align}
    \label{eq:target_mmt_simplified}
    N^{(0, 1)} = \frac{1}{2}\mathds{1} - \frac{1}{2} \ketbra{00}{00} \pm \frac{1}{2} \left( \ketbra{\phi^+}{\phi^+} - \ketbra{\phi^-}{\phi^-}\right).
\end{align}
We are now ready to prove the theorem. We need to show three statements:
\begin{enumerate}
    \item $\Lambda$ is trace-preserving, i.e. $\left(K^{(0)}\right)^* K^{(0)} + \sum_{N, k, l} (K^{(N)}_{k,l})^* K^{(N)}_{k,l} = \mathds{1}$.
    \item $\Lambda$ preserves the statistics, i.e. $\tr\left[M^{(v)}_{SR}\rho_{SR}\right] = \tr\left[N^{(v)}_{S'R'} \Lambda[\rho_{SR}]\right]$.
    \item $\Lambda$ is non-signalling from $S$ to $R'$, i.e. $\tr_{S'}\Lambda[\rho_{SR}]$ only depends on $\rho_{R}$.
\end{enumerate}
We proceed in the order outlined above:
\begin{enumerate}
    \item We have
        \begin{equation}
        \begin{aligned}
            &\left(K^{(0)}\right)^* K^{(0)} + \sum_{N, k, l} \left(K^{(N)}_{k,l}\right)^* K^{(N)}_{k,l} \\
            = &\ketbra{0,0}{0,0} + \sum_{N, k, l} \frac{2}{2^N} \left( \binom{N}{l} \ketbra{N-k, k}{N-k, k} + \binom{N}{k} \ketbra{N-l, l}{N-l, l} \right).
        \end{aligned}
        \end{equation}
        We can now use the formula $\sum_l \binom{N}{l} = \sum_k \binom{N}{k} = 2^{N-1}$ to simplify the expression above:
        \begin{equation}
        \begin{aligned}
            & \ketbra{0,0}{0,0} + \sum_{N, k} \ketbra{N-k, k}{N-k, k} + \sum_{N,l} \ketbra{N-l, l}{N-l, l} \\
            =& \sum_{N = 0}^\infty \sum_{m = 0}^N \ketbra{N-m, m}{N-m, m} = \sum_{m,n} \ketbra{n, m}{n, m} = \mathds{1}.
        \end{aligned}
        \end{equation}
    \item
        We see directly that $\tr\left[ M^{(\bot)}_{SR} \rho_{SR}\right] = \tr\left[ N^{(\bot)}_{S'R'} \Lambda [\rho_{SR}] \right]$. Comparing the expressions in \cref{eq:mmt_simplified} with those in \cref{eq:target_mmt_simplified} we see that for the two remaining measurement outcomes it is sufficient to show that the third term (after the $\pm$) is preserved between \cref{eq:mmt_simplified} and \cref{eq:target_mmt_simplified}. Hence we compute:
        \begin{equation}
        \begin{aligned}
            & \bra{\phi^+}\Lambda[\rho_{SR}]\ket{\phi^+} - \bra{\phi^-}\Lambda[\rho_{SR}]\ket{\phi^-} = \bra{01}\Lambda[\rho_{SR}]\ket{10} + \bra{10}\Lambda[\rho_{SR}]\ket{01} \\
            =& \sum_{N, k, l} \frac{2}{2^N} \sqrt{\binom{N}{l}}\sqrt{\binom{N}{k}} \left(\bra{N-k, k}\rho_{SR}\ket{N-l,l} + \bra{N-l,l}\rho_{SR}\ket{N-k,k} \right) \\
            =& \sum_{N}^{\infty} \sum_{m,n=0}^{N} \frac{1}{2^N}\sqrt{\binom{N}{m}}\sqrt{\binom{N}{n}} (1 - (-1)^{m+n}) \bra{N-n, n} \rho_{SR} \ket{N-m, m},
        \end{aligned}
        \end{equation}
        as desired.
    \item
        We compute
        \begin{equation}
        \begin{aligned}
            \tr_{S'} \Lambda[\rho_{SR}] =& \ketbra{0}{0}_{R'} \bra{0,0}\rho_{SR}\ket{0,0} + \ketbra{1}{1}_{R'} \sum_{N,k,l}\frac{2}{2^N} \binom{N}{l} \bra{N-k,k} \rho_{SR} \ket{N-k,k} \\
            &+ \ketbra{0}{0}_{R'} \sum_{N,k,l} \frac{2}{2^N} \binom{N}{k} \bra{N-l,l} \rho_{SR} \ket{N-l,l} \\
            =& \ketbra{1}{1}_{R'} \sum_{N,k} \bra{N-k,k} \rho_{SR} \ket{N-k,k} \\
            &+ \ketbra{0}{0}_{R'} \left( \bra{0,0}\rho_{SR}\ket{0,0} + \sum_{N,l} \bra{N-l,l} \rho_{SR} \ket{N-l,l} \right) \\
            =& \ketbra{1}{1}_{R'} \sum_{k} \bra{k} \tr_{S} \rho_{SR} \ket{k} + \ketbra{0}{0}_{R'} \sum_{l} \bra{l} \tr_{S} \rho_{SR} \ket{l},
        \end{aligned}
        \end{equation}
        where we again used the expression from before to get rid of the binomials. We now see that the final expression only depends on the state of $\rho_{SR}$ on the system $R$ and hence $\Lambda$ is non-signalling from $S$ to $R'$.
\end{enumerate}

\section{General considerations for the security proofs}
\label{sec:security_general}
In this section we present the steps that are shared between the two security proofs. The remainder of the steps is then presented in the respective chapters. Throughout this section we will make use of the following events (formally defined as subsets of $A^n\tilde{A}^nC^n$):
\begin{center}
\begin{tabular}{c|p{13cm}}
    $\Omega_\mathrm{EC}$ & The protocol did not abort in the EC step, i.e. $h(A^n) = h(\tilde{A}^n)$. \\
    $\Omega_\mathrm{EC}^\mathrm{fail}$ & The protocol did not abort in the EC step and $\tilde{A}^n \neq A^n$. \\
    $\Omega_\mathrm{EC}^\mathrm{cor}$ & The protocol did not abort in the EC step and $\tilde{A}^n = A^n$. This event can also be written as $\Omega_\mathrm{EC}^\mathrm{cor} = \Omega_\mathrm{EC} \cap (\Omega_\mathrm{EC}^\mathrm{fail})^c$. \\
    $\Omega_\mathrm{PE}$ & The protocol did not abort in the parameter estimation step, i.e. $f(\mathrm{freq}(C^n)) \geq H_\mathrm{exp}$. \\
\end{tabular}
\end{center}
We also define the event of not aborting the protocol which is given by $\Omega = \Omega_\mathrm{EC} \cap \Omega_\mathrm{PE}$. Security consists of two components: Completeness and Soundness. Proving those two properties is the content of the following two subsections. The completeness condition will put a lower bound on the admissible value of $\mathrm{leak_{EC}}$ while the soundness condition will put an upper bound on the admissible key length $l$.

\subsection{Completeness}
\label{app:completeness}
In this section we show completeness, i.e., we show that there exists an (honest) implementation that aborts with low probability. The main result is summarised in the following theorem:
\begin{thm}
    Let $\delta, \bar{\varepsilon}_s, \varepsilon^\mathrm{com}_\mathrm{EC} > 0$ and $p \in \mathbb{P}_\mathcal{C}$ be the expected statistics of the honest implementation. If
    \begin{align}
        \label{eq:bound_leak_EC}
        \mathrm{leak_{EC}} \geq n H(A|JB)_\mathrm{hon} + 2 \sqrt{n} \log 7 \sqrt{\log \frac{2}{\bar{\varepsilon}_s^2}} + 2\log \frac{1}{\varepsilon^\mathrm{com}_\mathrm{EC} - \bar{\varepsilon}_s} + 4,
    \end{align}
    and $H_\mathrm{exp} \leq f(p) - \delta$ then there exists an (honest) implementation of the protocol with abort probability
    \begin{align}
        \mathrm{Pr[abort]} \leq \varepsilon^\mathrm{com}_\mathrm{EC} + \exp \left( -n \frac{\delta^2/2}{\mathrm{Var}(f) + (\mathrm{Max}(f) - \mathrm{Min}(f)) \delta / 3} \right).
    \end{align}
    In the expression above $H(A|JB)_\mathrm{hon} = H(A_i|J_iB_i)_\mathrm{hon}$ refers to the entropy of Alice's raw key conditioned on Bob's information for the honest implementation.
\end{thm}
\begin{proof}
Let $\rho_\mathrm{hon}$ be the state at the end of the protocol when Eve is passive. There are two places where the protocol could abort. The first is during error correction, the second is during parameter estimation. Formally, this can be expressed as $\Omega_\mathrm{abort} = \Omega^c = \Omega_\mathrm{EC}^c \cup \Omega_\mathrm{PE}^c$. Using the union bound, we find that
\begin{align}
    \rho_\mathrm{hon}[\Omega_\mathrm{abort}] = \rho_\mathrm{hon}[\Omega_\mathrm{EC}^c \cup \Omega_\mathrm{PE}^c] \leq \rho_\mathrm{hon}[\Omega_\mathrm{EC}^c] + \rho_\mathrm{hon}[\Omega_\mathrm{PE}^c].
    \label{eq:completeness_union_bound}
\end{align}
We will now bound the two terms in this expression separately.

To bound the first term we note that according to \cite{Renes2012}, there exists an EC protocol with failure probability at most $\varepsilon_\mathrm{EC}^\mathrm{com}$ as long as:
\begin{align}
    \mathrm{leak_{EC}} \geq H_\mathrm{max}^{\bar{\varepsilon}_s}(A^n|J^nB^n) + 2\log \frac{1}{\varepsilon^\mathrm{com}_\mathrm{EC} - \bar{\varepsilon}_s} + 4.
\end{align}
    Next, we apply \cite[Corollary 4.10]{Dupuis2020} to bound the smooth max-entropy:
\begin{align}
    H_\mathrm{max}^{\bar{\varepsilon}_s}(A^n|J^nB^n) \leq n H(A|JB)_\mathrm{hon} + 2\sqrt{n} \log (1+2|\mathcal{A}|) \sqrt{\log \frac{2}{\bar{\varepsilon}_s^2}}.
\end{align}
    Therefore we can conclude that since inequality \cref{eq:bound_leak_EC} is satisfied by assumption (we have $|\mathcal{A}| = 3$), there exists an EC protocol such that $\rho_\mathrm{hon}[\Omega_\mathrm{EC}^c] \leq \varepsilon^\mathrm{com}_\mathrm{EC}$.

To bound the second term in \cref{eq:completeness_union_bound} we note that the honest implementation is IID. We follow the steps in \cite{Metger2023} to see that
\begin{align}
    \rho_\mathrm{hon}[\Omega_\mathrm{PE}^c] \leq \exp \left( -n \frac{\delta^2/2}{\mathrm{Var}(f) + (\mathrm{Max}(f) - \mathrm{Min}(f)) \delta / 3} \right).
\end{align}
Combining the two bounds then yields the result.
\end{proof}

\subsection{Soundness}
\label{sec:soundness}
Soundness is the statement that for any attack either the protocol aborts with high probability or Eve's knowledge about the key is small. Our soundness statement is summarized in the following theorem:
\begin{thm}
    Let $\alpha' \in (1, 3/2)$ and $\varepsilon_\mathrm{PA}, \varepsilon_s \in (0,1)$. Let $V$ and $K(\alpha')$ be as in \cref{thm:gen_EAT} for the min-tradeoff function $f$ of the protocol. If
    \begin{equation}
    \begin{aligned}
        l \leq & n H_\mathrm{exp} - n \frac{\alpha'-1}{2-\alpha'}\frac{\ln(2)}{2}V^2 - \frac{g(\varepsilon_s) + \alpha' \log\left(\frac{1}{2\varepsilon_s + \varepsilon_\mathrm{PA}}\right)}{\alpha'-1} - n \left(\frac{\alpha'-1}{2-\alpha'}\right)^2 K(\alpha') \\
        & - \mathrm{leak_{EC}} - \lceil \log (1/\varepsilon_\mathrm{EC}) \rceil - 2\log (1/\varepsilon_\mathrm{PA}) + 2,
    \end{aligned}
    \end{equation}
    then the protocol is $(\varepsilon_\mathrm{EC} + \varepsilon_\mathrm{PA} + 2\varepsilon_s)$-sound.
\end{thm}
\begin{proof}
    The proof more or less follows the steps from \cite{Tan2020a,Metger2023}. Let $E_n$ denote Eve's quantum system at the end of the protocol. Let us denote with $O_\mathrm{EC}$ the (classical) information exchanged during the error correction procedure of the protocol and with $F$ the information exchanged during privacy amplification. Write $\Sigma=E_n I^n J^n O_\mathrm{EC}$ for Eve's total information before privacy amplification and let $\rho_{K_A^l K_B^l A^n \tilde{A}^n C^n F \Sigma}$ be the state at the end of the protocol. The goal is to upper-bound the quantity (see \cref{def:soundness_subnorm}):
\begin{align}
    \label{eq:sec_def}
    \frac{1}{2} \left\| (\rho_{\land\Omega})_{K_A^l K_B^l F\Sigma} - \tau_{K_A^l K_B^l} \otimes (\rho_{\land\Omega})_{F\Sigma} \right\|_1.
\end{align}
    The event $\Omega$ includes the event where the protocol did not abort but error correction produced incorrect outputs, i.e. $\tilde{A}^n \neq A^n$ and hence $K_B^l \neq K_A^l$. To address this we write $\Omega = (\Omega \cap \Omega_\mathrm{EC}^\mathrm{fail}) \cup (\Omega \cap (\Omega_\mathrm{EC}^\mathrm{fail})^c) = \Omega_1 \cup \Omega_2$ with $\Omega_1 = \Omega \cap \Omega_\mathrm{EC}^\mathrm{fail}$ and $\Omega_2 = \Omega \cap (\Omega_\mathrm{EC}^\mathrm{fail})^c$. Since $\Omega_1 \cap \Omega_2 = \emptyset$, we have $\rho_{\land\Omega} = \rho_{\land\Omega_1} + \rho_{\land\Omega_2}$. Inserting this into \cref{eq:sec_def} and applying the triangle inequality we get:
\begin{equation}
\begin{aligned}
    & \frac{1}{2} \left\| (\rho_{\land\Omega})_{K_A^l K_B^l F\Sigma} - \tau_{K_A^l K_B^l} \otimes (\rho_{\land\Omega})_{F\Sigma} \right\|_1 \\
    \leq &\frac{1}{2} \left\| (\rho_{\land\Omega_1})_{K_A^l K_B^l F\Sigma} - \tau_{K_A^l K_B^l} \otimes (\rho_{\land\Omega_1})_{F\Sigma} \right\|_1 + \frac{1}{2} \left\| (\rho_{\land\Omega_2})_{K_A^l K_B^l F\Sigma} - \tau_{K_A^l K_B^l} \otimes (\rho_{\land\Omega_2})_{F\Sigma} \right\|_1 \\
    \leq& \varepsilon_\mathrm{EC} + \frac{1}{2} \left\| (\rho_{\land\Omega_2})_{K_A^l K_B^l F\Sigma} - \tau_{K_A^l K_B^l} \otimes (\rho_{\land\Omega_2})_{F\Sigma} \right\|_1,
\end{aligned}
\end{equation}
    where in the last line we noted that by the properties of universal hashing we have that $\rho[\Omega_1] \leq \rho[\Omega_\mathrm{EC}^\mathrm{fail}] \leq \varepsilon_\mathrm{EC}$. We now focus on bounding the second term. Since this term is conditioned on $\Omega_2 = \Omega \cap (\Omega_\mathrm{EC}^\mathrm{fail})^c$, we know that $\tilde{A}^n = A^n$ and hence $K_A^l = K_B^l$. Therefore we can trace out the systems $K_B^l$ in the trace distance above to recover a situation similar to the one in the leftover hashing lemma.

    We let $\varepsilon_s, \varepsilon_\mathrm{PA} > 0$ and our goal will now be to upper-bound the second term in the expression above by $2\varepsilon_s + \varepsilon_\mathrm{PA}$. We can expand $\Omega_2 = \Omega_\mathrm{PE} \cap \Omega_\mathrm{EC} \cap (\Omega_\mathrm{EC}^\mathrm{fail})^c = \Omega_\mathrm{PE} \cap \Omega_\mathrm{EC}^\mathrm{cor}$. If now either $\rho[\Omega_\mathrm{PE}] < 2\varepsilon_s + \varepsilon_\mathrm{PA}$ or $\rho_{|\Omega_\mathrm{PE}}[\Omega_\mathrm{EC}^\mathrm{cor}] < 2\varepsilon_s + \varepsilon_\mathrm{PA}$ then $\rho[\Omega_2] < 2\varepsilon_s + \varepsilon_\mathrm{PA}$ and the statement holds trivially. Hence we can from now on assume that both $\rho[\Omega_\mathrm{PE}] \geq 2\varepsilon_s + \varepsilon_\mathrm{PA}$ and $\rho_{|\Omega_\mathrm{PE}}[\Omega_\mathrm{EC}^\mathrm{cor}] \geq 2\varepsilon_s + \varepsilon_\mathrm{PA}$. 
    We can write $\rho_{\land\Omega_2} = \rho_{\land(\Omega_\mathrm{PE} \cap \Omega_\mathrm{EC}^\mathrm{cor})} = \rho[\Omega_\mathrm{PE}] (\rho_{|\Omega_\mathrm{PE}})_{\land\Omega_\mathrm{EC}^\mathrm{cor}}$ and to simplify the notation we introduce $\sigma = \rho_{|\Omega_\mathrm{PE}}$. With this we get
\begin{equation}
\begin{aligned}
    \left\| (\rho_{\land\Omega_2})_{K_A^l F\Sigma} - \tau_{K_A^l} \otimes (\rho_{\land\Omega_2})_{F\Sigma} \right\|_1 
    & = \rho[\Omega_\mathrm{PE}] \left\| (\sigma_{\land\Omega_\mathrm{EC}^\mathrm{cor}})_{K_A^lF\Sigma} - \tau_{K_A^l} \otimes (\sigma_{\land\Omega_\mathrm{EC}^\mathrm{cor}})_{F\Sigma}\right\|_1 \\
    & \leq \left\| (\sigma_{\land\Omega_\mathrm{EC}^\mathrm{cor}})_{K_A^lF\Sigma} - \tau_{K_A^l} \otimes (\sigma_{\land\Omega_\mathrm{EC}^\mathrm{cor}})_{F\Sigma} \right\|_1.
\end{aligned}
\end{equation}
We now apply \cref{lem:leftoverhashing} to upper-bound the trace distance:
\begin{align}
    \frac{1}{2} \left\| (\sigma_{\land\Omega_\mathrm{EC}^\mathrm{cor}})_{K_A^l F\Sigma} - \tau_{K_A^l} \otimes (\sigma_{\land\Omega_\mathrm{EC}^\mathrm{cor}})_{F\Sigma} \right\|_1
    \leq 2 \varepsilon_s + 2^{\frac{1}{2}\left(l - H_\mathrm{min}^{\varepsilon_s}(A^n|\Sigma)_{\sigma_{\land\Omega_\mathrm{EC}^\mathrm{cor}}} - 2 \right)} \overset{!}{\leq} 2 \varepsilon_s + \varepsilon_\mathrm{PA},
\end{align}
where the smooth min-entropy is evaluated on the state before the PA part of the protocol. To arrive at our desired result we now need to upper-bound the second term by $\varepsilon_\mathrm{PA}$:
\begin{align}
    2^{\frac{1}{2}\left(l - H_\mathrm{min}^{\varepsilon_s}(A^n|\Sigma)_{\sigma_{\land\Omega_\mathrm{EC}^\mathrm{cor}}} - 2 \right)} \leq \varepsilon_\mathrm{PA} \iff l \leq H_\mathrm{min}^{\varepsilon_s}(A^n|\Sigma)_{\sigma_{\land\Omega_\mathrm{EC}^\mathrm{cor}}} - 2\log \frac{1}{\varepsilon_\mathrm{PA}} + 2.
\end{align}
    To find an upper bound on the admissible values of $l$, we now need to lower-bound the smooth min-entropy. Since the EAT applies to the state before the error correction procedure, we first need to eliminate this side-information.
This can be achieved using the chain rule in \cite[Lemma 6.8]{Tomamichel2016}
\begin{equation}
\begin{aligned}
    H_\mathrm{min}^{\varepsilon_s}(A^n|\Sigma)_{\sigma_{\land\Omega_\mathrm{EC}^\mathrm{cor}}} \geq& H_\mathrm{min}^{\varepsilon_s}(A^n|E_n I^n J^n)_{\sigma_{\land\Omega_\mathrm{EC}^\mathrm{cor}}} - \log (\dim O_\mathrm{EC}) \\
    \geq & H_\mathrm{min}^{\varepsilon_s}(A^n|E_n I^n J^n)_{\sigma_{\land\Omega_\mathrm{EC}^\mathrm{cor}}} - \mathrm{leak_{EC}} - \lceil \log (1/\varepsilon_\mathrm{EC}) \rceil,
\end{aligned}
\end{equation}
    where we noted that $\log(\dim O_\mathrm{EC}) \leq \mathrm{leak_{EC}} + \lceil \log(1/\varepsilon_\mathrm{EC}) \rceil$.
    We now note that since $\sqrt{\sigma[\Omega_\mathrm{EC}^\mathrm{cor}]} \geq \sqrt{2\varepsilon_s + \varepsilon_\mathrm{PA}} > \varepsilon_s$ we can apply \cite[Lemma 10]{Tomamichel2017}:
\begin{equation}
    H_\mathrm{min}^{\varepsilon_s}(A^n|E_n I^n J^n)_{\sigma_{\land\Omega_\mathrm{EC}^\mathrm{cor}}} \geq H_\mathrm{min}^{\varepsilon_s}(A^n|E_n I^n J^n)_{\sigma}.
\end{equation}
    Remembering that $\sigma = \rho_{|\Omega_\mathrm{PE}}$ and applying \cref{thm:gen_EAT} with $\rho[\Omega_\mathrm{PE}] \geq 2\varepsilon_s + \varepsilon_\mathrm{PA}$ to the remaining smooth min-entropy term then yields the desired result.
\end{proof}

\section{Security of relativistic QKD}
\label{sec:security_rel_qkd}

Here we show the remaining steps for proving security of the relativistic protocol. Namely this includes finding a valid min-tradeoff function (the parameter $f$ of the protocol).
A min-tradeoff function is an affine function $f: \mathbb{P}_\mathcal{C} \rightarrow \mathbb{R}$ satisfying:
\begin{equation}
    f(p) \leq \inf_{\nu \in \Sigma_i(p)} H(A_i|E_i I^i J^i\tilde{E}_{i-1})_{\nu} \qquad \forall p \in \mathbb{P}_\mathcal{C}, i \in [n],
    \label{eq:full_trdf}
\end{equation}
where $\tilde{E}_{i-1}$ is a system isomorphic to $E_{i-1}I^{i-1}J^{i-1}$ and $\Sigma_i(p)$ is the set of all states that can be produced by our channels and that are compatible with the statistics $p$. To simplify the construction of our min-tradeoff function we follow the steps outlined in \cite{Depuis2019} i.e. we split our protocol into key rounds where $T_i=0$ and test rounds where $T_i=1$. For this we separate $\mathcal{C} = \mathcal{C'} \cup \{ \varnothing \}$ with $\mathcal{C'}=\{\mathrm{corr}, \mathrm{err}, \bot \}$. We then apply \cite[Lemma V.5]{Depuis2019} which states that if
\begin{equation}
    g(p') \leq \inf H(A_i|E_iI^iJ^i\tilde{E}_{i-1}) \qquad \forall p' \in \mathbb{P}_\mathcal{C'}, i \in [n]
    \label{eq:simplified_trdf}
\end{equation}
is a valid min-tradeoff function for the protocol rounds constrained on obtaining statistics $p'$ in the test rounds, then the affine function defined by
\begin{equation}
\begin{aligned}
    f(\delta_c) &= \mathrm{Max}(g) + \frac{1}{\gamma}(g(\delta_c) - \mathrm{Max}(g)) \quad \forall c \in \mathcal{C'},\\
    f(\delta_\varnothing) &= \mathrm{Max}(g)
\end{aligned}
\end{equation}
is a valid min-tradeoff function for the full protocol (which includes both the key and the test rounds). The main difference between the min-tradeoff functions given in \cref{eq:full_trdf} and \cref{eq:simplified_trdf} is that in the former the optimization is constrained on obtaining the correct statistics in all the rounds (including key rounds), whereas in the latter the optimization is constrained only on obtaining the correct statistics in the test rounds. The properties of the two min-tradeoff functions can be related by
\begin{equation}
\begin{aligned}
    \mathrm{Max}(f) &= \mathrm{Max}(g), \\
    \mathrm{Min}_\Sigma(f) &\geq \mathrm{Min}(g), \\
    \mathrm{Var}(f) &\leq \frac{1}{\gamma} (\mathrm{Max}(g) - \mathrm{Min}(g))^2.
\end{aligned}
\end{equation}

The remaining task now is to find a min-tradeoff function as in \cref{eq:simplified_trdf}. To ease this task there are a number of simplifications that can be made to Eve's attack channel (see \cref{fig:rel_qkd_attack_channel}). By the existence of the squashing map (see \cref{thm:squashing}) we conclude that we can replace our photonic systems $S_i$ and $R_i$ with qubit systems $S'_i$ and $R'_i$. Additionally we can absorb the systems $I^{i-1}J^{i-1}$ and $\tilde{E}_{i-1}$ into Eve's attack map. Next, we note that since the state of the reference is identical for all rounds (and hence is uncorrelated with the key), we can allow Eve to prepare this state herself. Finally we note that, in principle, we would also need to consider the full Fock space on the input to Eve's channel. However because in our protocol the inputs live in $T_i = \mathrm{span}\{\ket{\alpha}, \ket{-\alpha}\}$ we can restrict Eve's input to this reduced subspace (restricting the real attack channel to this subspace yields a new channel with the same key rate).

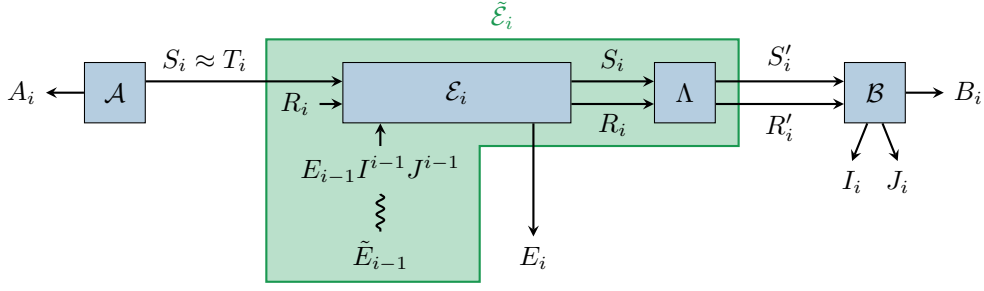
\begin{figure}[t]
    \centering
    \begin{tikzpicture}[
    boxnode/.style={rectangle, draw=black, minimum height=0.8cm, minimum width=0.8cm, fill=myblue!30},
    emptynode/.style={}
]

    \node[boxnode] (labA) at (-4.5, 0) {$\mathcal{A}$};
    \node[boxnode, minimum width=3cm] (channel) at (0, 0) {$\mathcal{E}_i$};
    \node[boxnode] (squash) at (3, 0) {$\Lambda$};
    \node[boxnode] (labB) at (+5.5, 0) {$\mathcal{B}$};

    \draw[thick,->,>=stealth] (labA.west) -- ([xshift=-0.5cm] labA.west) node[left] {$A_i$};

    \draw[thick,->,>=stealth] ([xshift=-0.1cm] labB.south) -- ([xshift=-0.3cm,yshift=-0.5cm] labB.south) node[below] {$I_i$};
    \draw[thick,->,>=stealth] ([xshift=0.1cm] labB.south) -- ([xshift=+0.3cm,yshift=-0.5cm] labB.south) node[below] {$J_i$};
    \draw[thick,->,>=stealth] (labB.east) -- ([xshift=+0.5cm] labB.east) node[right] {$B_i$};

    \draw[thick,->,>=stealth] ([yshift=+0.15cm] labA.east) -- node[pos=0.3,above]{$S_i \approx T_i$} ([yshift=+0.15cm] channel.west);
    \draw[thick,->,>=stealth] ([xshift=-0.3cm,yshift=-0.15cm] channel.west) node[left]{$R_i$} -- ([yshift=-0.15cm] channel.west);
    \draw[thick,->,>=stealth] ([yshift=+0.15cm] channel.east) -- node[above]{$S_i$} ([yshift=+0.15cm] squash.west);
    \draw[thick,->,>=stealth] ([yshift=-0.15cm] channel.east) -- node[below]{$R_i$} ([yshift=-0.15cm] squash.west);

    \draw[thick,->,>=stealth] ([yshift=+0.15cm] squash.east) -- node[above]{$S'_i$} ([yshift=+0.15cm] labB.west);
    \draw[thick,->,>=stealth] ([yshift=-0.15cm] squash.east) -- node[below]{$R'_i$} ([yshift=-0.15cm] labB.west);

    \draw[thick,->,>=stealth] ([xshift=+0.5cm,yshift=-0.3cm] channel.south west)  node[below] (Eprev) {$E_{i-1}I^{i-1}J^{i-1}$} -- ([xshift=+0.5cm] channel.south west);
    \draw[thick,->,>=stealth] ([xshift=-0.5cm] channel.south east) -- ([xshift=-0.5cm,yshift=-1.5cm] channel.south east)  node[below] {$E_{i}$};
    \draw[thick,decorate,decoration={snake,amplitude=0.5mm,segment length=1.5mm}] (Eprev.south) -- ([yshift=-0.5cm] Eprev.south) node[below] (purif) {$\tilde{E}_{i-1}$};

    \begin{scope}[on background layer]
        \draw[thick,draw=mygreen,fill=mygreen!30] ([xshift=-1.0cm,yshift=+0.3cm] channel.north west) --
        ([xshift=-1.0cm,yshift=-2.1cm] channel.south west) --
        ([xshift=0.3cm,yshift=-2.1cm] channel.south) -- ([xshift=0.3cm,yshift=-0.3cm] channel.south) --
        ([xshift=0.3cm,yshift=-0.3cm] squash.south east) --
        ([xshift=0.3cm,yshift=0.3cm] squash.north east) -- node[above,text=mygreen] {$\tilde{\mathcal{E}}_i$}
        cycle;
    \end{scope}
\end{tikzpicture}
    \caption{A diagrammatic representation of the channel $\mathcal{M}_{i}$ in the $i$-th round of the protocol. Formally we need to consider attack channels with inputs $R_i$ and $E_{i-1}I^{i-1}J^{i-1}$ (which might be entangled with a system $\tilde{E}_{i-1}$) and outputs $E_i$, $R_i$ and $S_i$. However, the systems $R_iE_{i-1}I^{i-1}J^{i-1}\tilde{E}_{i-1}$ can be absorbed into the definition of $\tilde{\mathcal{E}}_i$. Similarly we can include the squashing map $\Lambda$ into Eve's attack (this can only decrease the key rate). Since both $\mathcal{E}_i$ and $\Lambda$ are non-signalling, $\tilde{\mathcal{E}_i}$ will also be non-signalling (from $T_i$ to $R'_i$).}
    \label{fig:rel_qkd_attack_channel}
\end{figure}

In conclusion: we can restrict ourselves to attack channels which take a single qubit as input (the system $T_i$) and produce qubit outputs $S'_i$ and $R'_i$ together with some side-information $E_i$. Furthermore, this simplified attack channel is non-signalling from $T_i$ to $R'_i$. By strong subadditivity, we can also assume that $E_i$ is a purifying system.

Assuming $E_i$ to be a purifying system gives Eve too much power in general (Eve cannot purify an unknown state). To resolve this, we switch to an entanglement based version of the protocol. Explicitly we note that the post-measurement state would be identical if Alice had instead prepared the entangled state (for brevity we drop the index $i$)
\begin{align}
    \ket{\psi}_{\tilde{V}T} = \frac{1}{\sqrt{2}} \ket{0}_{\tilde{V}} \otimes \ket{\alpha}_T + \frac{1}{\sqrt{2}} \ket{1}_{\tilde{V}} \otimes \ket{-\alpha}_T,
\end{align}
followed by measuring $\tilde{V}$ locally to obtain her key bit $V$. However, in contrast to the usual literature, we do not give this source to Eve. The reason for this is that it is not obvious how the non-signalling constraint would look like in that scenario. In this entanglement based picture we can then define the POVM associated to the evaluation function in \cref{eq:eval} (now restricted to only test rounds where $J\neq\varnothing$):
\begin{equation}
\begin{aligned}
    \Gamma_{\tilde{V}S'R'}^{(\mathrm{cor})} &= \ketbra{0}{0}_{\tilde{V}} \otimes N^{(0)}_{S'R'} + \ketbra{1}{1}_{\tilde{V}} \otimes N^{(1)}_{S'R'} \\
    \Gamma_{\tilde{V}S'R'}^{(\mathrm{err})} &= \ketbra{0}{0}_{\tilde{V}} \otimes N^{(1)}_{S'R'} + \ketbra{1}{1}_{\tilde{V}} \otimes N^{(0)}_{S'R'} \\
    \Gamma_{\tilde{V}S'R'}^{(\bot)} &= \mathds{1}_{\tilde{V}} \otimes N^{(\bot)}_{S'R'},
\end{aligned}
\end{equation}
where $\{ N_{S'R'}^{(b)} \}_b$ are as defined in \cref{eq:mmt_ops_Bob_qubits}. With this we can compute the expected statistics of any state $\sigma_{\tilde{V}S'R'}$ as
\begin{align}
    \sigma_C(c) &= \tr \left[ \sigma_{\tilde{V}S'R'} \Gamma^{(c)}_{\tilde{V}S'R'} \right].
\end{align}
Next, we parametrize $g: \mathbb{P}_\mathcal{C'} \rightarrow \mathbb{R}$ as $g(p') = c_\lambda + \lambda \cdot p'$ with $\lambda \in \mathbb{R}^{|\mathcal{C'}|}$ and note that if
\begin{align}
    \label{eq:opt_dps_simple}
    c_\lambda \leq \inf_{\nu} \left\{H(A|IJE)_\nu - \lambda \cdot \nu_C\right\},
\end{align}
then $g$ is a valid min-tradeoff function (the minimization is over all states compatible with our channels $\mathcal{M}_i$). The parameter $\lambda$ is optimized using standard numerical optimization techniques (choosing a non-optimal $\lambda$ decreases the key rate but does not compromise security). Conservatively, we assume that in the test rounds ($J \neq \varnothing$) no entropy is produced. This means that we can get rid of the conditioning system $J$ at the cost of reducing the entropy by a factor of $1-\gamma$. Hence the final optimization problem that we need to solve is:

\begin{equation}
\begin{aligned}
    \label{eq:opt_dps_full}
    &\inf_{\rho_{TS'R'}} ((1-\gamma) H(A|IE,J=\varnothing)_{\nu(\rho)} - \lambda \cdot \nu_C(\rho)) \\
    & \quad \begin{aligned}
        \mathrm{s.t.} \; && \rho_{TS'R'} & \geq 0, \\
        && \tr[\rho_{TS'R'}] &= 1, \\
        && \rho_{TR'} &= \frac{\mathds{1}_{T}}{d_{T}} \otimes \rho_{R'},
    \end{aligned}
\end{aligned}
\end{equation}
where $\nu(\rho)$ is the state after Alice and Bob measure a purification of $\mathcal{I}_{\tilde{V}} \otimes \mathcal{C}^{-1}(\rho)\left[\ketbra{\psi}{\psi}_{\tilde{V}T}\right]$ ($\rho_{TS'R'}$ is the Choi state of Eve's attack channel). It is also worth noting that by the Stinespring representation theorem, giving Eve a purification does not give her more power (the input states are pure). The last constraint in \cref{eq:opt_dps_full} originates from the non-signalling condition (see \cref{lem:choi}). The optimization problem in \cref{eq:opt_dps_full} can then be evaluated using the methods in \cite{Araujo2023}\footnote{Technically we need a slight generalization of the method where we absorb the conditioning on $I$ into the normalization of $\nu$.}.

\section{A note about memory in DPS QKD}
\label{sec:security_proof_dps}

\begin{figure}[t]
    \centering
    \newcommand{\plusbinomial}[3][2]{(#2 + #3)^#1}
\newcommand{\BS}[1]{
    \begin{scope}[shift={(#1)}]
        \draw[fill=gray!30] (-0.25,-0.25) rectangle (0.25,0.25);
        \draw (-0.25,-0.25) rectangle (0.25,0.25);
        \draw (-0.25,-0.25) -- (0.25,0.25);
    \end{scope}
}

\begin{tikzpicture}[
    boxnode/.style={rectangle, draw=black, minimum width=0.8cm, thick, fill=black!20},
    emptynode/.style={},
    mirror/.pic={
        \draw[decorate,decoration={markings, mark=between positions 0.015 and 0.98 step 0.1072 with {
            \draw (0,0)--(90:3pt);}}] (-0.2,-0.2) -- (0.2,0.2);
        \draw[thick] (-0.2,-0.2) -- (0.2,0.2);
    },
    detector/.pic={\draw[draw=black,fill=gray!20] (0,0.2) arc(90:-90:0.2cm and 0.2cm) -- cycle;},
    scale=0.9,
    ]
    \node[boxnode] (source) at (-4, 0) {Source};
    \node[boxnode] (PM) at (-2.5, 0) {PM};

    \node (BSBob1) at (3, 0) {};
    \BS{BSBob1.center};

    \node (BSBob2) at (6, 0) {};
    \begin{scope}[rotate=90]
        \BS{BSBob2.center};
    \end{scope}

    \node (MirrorBob1) at (3, 1.5) {};
    \draw (MirrorBob1.center) pic {mirror};

    \node (MirrorBob2) at (6, 1.5) {};
    \draw (MirrorBob2.center) pic[rotate=-90] {mirror};

    \draw[thick,draw=myred] (source.east) -- (PM.west);
    \draw[thick,draw=myred] (PM.east) -- (BSBob1.center) -- node[above]{$S$} (BSBob2.center);
    \draw[thick,draw=myred] (BSBob1.center) -- (MirrorBob1.center) -- node[above]{$R$} (MirrorBob2.center) -- (BSBob2.center);

    \draw[thick,draw=myred] (BSBob2.center) -- ([xshift=+1.0cm] BSBob2.center) pic{detector};
    \draw[thick,draw=myred] (-1,0) (BSBob2.center) -- ([yshift=-1.0cm] BSBob2.center) pic[rotate=-90]{detector};

    \begin{scope}[on background layer]
        \draw[thick,dotted,draw=myblue] ([xshift=0.5cm,yshift=-1.3cm] PM.east) --
            ([xshift=0.5cm,yshift=2.5cm] PM.east) -- node[below,text=myblue]{Situation 2}
            ([xshift=0.5cm,yshift=2.5cm] BSBob2.east) --
            ([xshift=0.5cm,yshift=+0.6cm] BSBob2.east) --
            ([xshift=-0.6cm,yshift=+0.6cm] BSBob2.west) --
            ([xshift=-0.6cm,yshift=-1.3cm] BSBob2.west) --
            cycle;
        \node[fit={([xshift=0.8cm,yshift=-1cm] PM.east) ([xshift=-0.8cm,yshift=1.7cm] BSBob1.west)},thick,draw=mygreen] (Box1) {};
        \draw (Box1.north) node[below,text=mygreen] {Situation 1};
    \end{scope}
\end{tikzpicture}
    \caption{The different possible regions of Eve's influence. Situation 1 is the actual situation of the DPS protocol which does not satisfy the non-signalling constraint. So instead we give Eve access to the memory system $R$ (the upper arm of the interferometer) as well as one of Bob's beam splitters. This is shown as situation 2. This change can only decrease the secret key rate. }
	\label{fig:dps_attack_regions}
\end{figure}
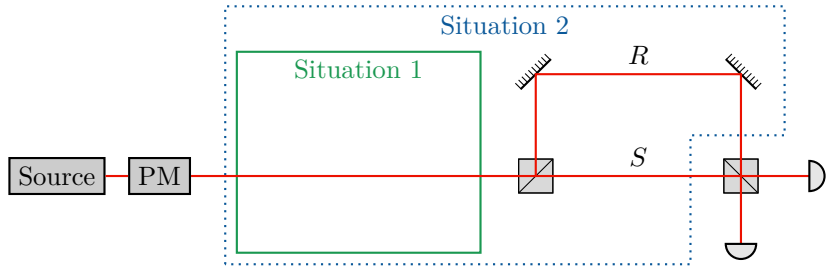

\begin{figure}[t]
    \centering
    \newcommand{\channels}[1]{
    \node[green_box] (E#1) at (0, 0) {$\mathcal{E}_{#1}$};
    \node[green_box] (F#1) at (1.5, 0) {$\mathcal{F}_{#1}$};
    \node[blue_box] (A#1) at ([yshift=+1cm] E#1.north) {$\mathcal{A}$};
    \node[blue_box,minimum height=1.5cm] (K#1) at ([yshift=0.4cm] F#1.north|-A#1.east) {$\mathcal{K}$};

    \node[blue_box] (BS#1) at ([yshift=-1cm] E#1.south) {BS};
    \node[blue_box] (B#1) at ([yshift=-2.5cm] E#1.south) {$\mathcal{B}$};

    \draw[thick,->,>=stealth,rounded corners] (B#1.east) -- node[above]{$I_{#1}J_{#1}$} (B#1.east-|F#1.south) -- (F#1.south);
    \draw[thick,->,>=stealth] (F#1.north) -- (K#1.south);
    \draw[thick,->,>=stealth] (E#1.east) -- (F#1.west);
    \draw[thick,->,>=stealth] (A#1.east) -- node[above]{$U_{#1}$} (A#1.east-|K#1.west);
    \draw[thick,->,>=stealth] ([xshift=-0.2cm] K#1.north) -- ([xshift=-0.4cm,yshift=0.7cm] K#1.north) node[above] {$A_#1$};
    \draw[thick,->,>=stealth] ([xshift=+0.2cm] K#1.north) -- ([xshift=+0.4cm,yshift=0.7cm] K#1.north) node[above] {$C_#1$};

    \draw[thick,->,>=stealth] (A#1.south) -- (E#1.north);
    \draw[thick,->,>=stealth] (E#1.south) -- (BS#1.north);
    \draw[thick,->,>=stealth] (BS#1.south) -- (B#1.north);
    \draw[thick,->,>=stealth] (B#1.south) -- ([yshift=-0.5cm] B#1.south) node[below]{$B_#1$};

    \draw[thick,-,>=stealth] (BS#1.east) -- ([xshift=-0.07cm] F#1.north|-BS#1.east);
    \node (BSOut#1) at (BS#1.east-|F#1.east) {};
    \draw[thick,-,>=stealth] ([xshift=+0.07cm] F#1.north|-BS#1.east) -- (BSOut#1);

    \begin{scope}[on background layer]
        \draw[fill=LightGray!50,rounded corners=0.1cm]
            ([xshift=-0.2cm,yshift=-1.1cm] B#1.south west) -- node[below]{$\mathcal{M}_{#1}$}
            ([xshift=+0.2cm,yshift=-1.1cm] B#1.south east-|F#1.east) -- 
            ([xshift=+0.2cm,yshift=+0.2cm] K#1.north east) -- 
            ([xshift=-0.2cm,yshift=+0.2cm] K#1.north east-|A#1.north west) --
            cycle;

    \end{scope}[on background layer]
}

\begin{tikzpicture}[
    green_box/.style={
        text=black,
        draw=black,
        minimum height=1.5cm,
        minimum width=0.8cm,
        fill=mygreen!30,
    },
    blue_box/.style={
        align=center,
        draw=black,
        minimum height=0.8cm,
        minimum width=0.8cm,
        fill=myblue!30,
    },
]
    \channels{1}
    \draw[thick,->,>=stealth] ([xshift=-0.9cm] E1.west) -- node[above,pos=0.4]{$\hat{E}_0$} (E1.west);
    \draw[thick,->,>=stealth] ([xshift=-0.9cm] B1.west) -- node[above,pos=0.4]{$S_0$} (B1.west);
    \draw[thick,->,>=stealth] ([xshift=-2.4cm,yshift=0.4cm] K1.west) -- node[above,pos=0.15]{$R_0$} ([yshift=0.4cm] K1.west);

    \begin{scope}[shift={(4cm,0)}]
        \channels{2}
    \end{scope}
    \draw[thick,->,>=stealth,rounded corners] (BSOut1.west) -- ([xshift=0.3cm] BSOut1.east) -- ([xshift=0.3cm] BSOut1.east|-B2.west) -- node[above,pos=0.4](tmp){$S_1$} (B2.west);
    \draw[thick,->,>=stealth] ([yshift=0.4cm] K1.east) -- node[above=0.10cm](tmp2){} ([yshift=0.4cm] K2.west);
    \node at (tmp2-|tmp) {$R_1$};
    \draw[thick,->,>=stealth] (F1.east) -- node[above=0.15cm](tmp2){} (E2.west);
    \node at (tmp2-|tmp) {$\hat{E}_1$};

    \begin{scope}[shift={(10.2cm,0)}]
        \channels{n}
    \end{scope}

    \node (dots) at ([xshift=1.9cm] F2.east) {$\cdots$};
    \node at ([yshift=0.4cm] K2.east-|dots) {$\cdots$};
    \node at (Bn.east-|dots) {$\cdots$};
    \draw[thick,->,>=stealth,rounded corners] (BSOut2.west) -- ([xshift=0.3cm] BSOut2.east) -- ([xshift=+0.3cm] BSOut2.east|-Bn) -- node[above](tmp){$S_2$} ([xshift=-0.3cm] Bn-|dots.west);
    \draw[thick,->,>=stealth] ([yshift=0.4cm] K2.east) -- node[above=0.10cm](tmp2){} ([xshift=-0.3cm,yshift=0.4cm] K2-|dots.west);
    \node at (tmp2-|tmp) {$R_2$};
    \draw[thick,->,>=stealth] (F2.east) -- node[above=0.15cm](tmp2){} ([xshift=-0.3cm] F2-|dots.west);
    \node at (tmp2-|tmp) {$\hat{E}_2$};

    \draw[thick,->,>=stealth] ([xshift=+0.3cm,yshift=0.4cm] dots.east|-Kn.west) -- node[above,pos=0.2]{$R_{n-1}$} ([yshift=0.4cm] Kn.west);
    \draw[thick,->,>=stealth] ([xshift=+0.3cm] dots.east|-En.west) -- node[above,pos=0.4]{$\hat{E}_{n-1}$} (En.west);
    \draw[thick,->,>=stealth] ([xshift=+0.3cm] dots.east|-Bn.west) -- node[above,pos=0.4]{$S_{n-1}$} (Bn.west);

    \draw[thick,->,>=stealth,rounded corners] (BSOutn.west) -- ([xshift=0.3cm] BSOutn.east) -- ([xshift=+0.3cm] BSOutn.east|-Bn) -- node[above](tmp){$S_n$} ([xshift=+1.3cm] Bn-|Fn.east);
    \draw[thick,->,>=stealth] ([yshift=0.4cm] Kn.east) -- node[above=0.10cm](tmp2){} ([xshift=+1.3cm,yshift=0.4cm] Kn.east-|Fn.east);
    \node at (tmp2-|tmp) {$R_n$};
    \draw[thick,->,>=stealth] (Fn.east) -- node[above=0.15cm](tmp2){} ([xshift=+1.3cm] Fn.east);
    \node at (tmp2-|tmp) {$\hat{E}_n$};
\end{tikzpicture}
    \caption{\label{fig:dps_eat_channels}The construction of the EAT channels for the DPS protocol. The channels $\mathcal{F}_i$ simply store a copy of $I_iJ_i$ in Eve's side-information. The systems $R_i$ store a copy of the previous outputs, i.e., $R_i=U^i$. To apply the generalised EAT we make the substitution $S_i\hat{E}_i \rightarrow E_i$ and keep $R_i$ as the memory register. Note that $\mathcal{M}_i$ does not signal from $R_{i-1}$ to $E_i$ and therefore the conditions of the generalised EAT are satisfied. The only reason for having the system $R_i$ is so that $A_i$ and $C_i$ can be computed. To compute single-round entropies we include the beam splitter (denoted as BS in the figure) in Eve's attack channel (see also Situation 2 in \cref{fig:dps_attack_regions}).}
\end{figure}

One of the main technical challenges when trying to prove security of the DPS protocol is that we need to consider memory effects (the system $R$ in \cref{fig:dps_attack_regions}). Naively one might hope to be able to include these memory effects in the EAT using the register $R_i$. After all, this is precisely what this register is for. There is however a problem: Bob's outputs and hence the sifting information $I_i$ depends on the state of the memory system $R$. But the register $I_i$ will also become available to Eve thus breaking the non-signalling condition of the EAT. We circumvent this problem by including the memory system $R$ in Eve's register $E_i$ (situation 2 in \cref{fig:dps_attack_regions}). Now the non-signalling constraint is trivially satisfied since the register $R_i$ is the trivial (one dimensional) system. On the other hand this also means that we now have to condition on this additional system (to which Eve in reality has no access). It also turns out that when computing entropies for situation 1 in \cref{fig:dps_attack_regions} we observe that the entropy decreases with increasing photon number. This indicates that it is not possible to find (an obvious) squashing map for the scenario where we include the first beam splitter in Bob's measurement operator. Lastly to compute the registers $A_i$ and $C_i$ we need access to $U_{i-1}$. We can achieve this by carrying $U_{i-1}$ through the channels using the memory register of the generalised EAT. The final construction of the EAT channels is shown in \cref{fig:dps_eat_channels}.

\end{document}